\documentclass[12pt,fleqn]{article}
\usepackage{lscape}
\usepackage{amsfonts,ulem,url}
\usepackage{hyperref}
\usepackage{xr-hyper}
\usepackage{amssymb}
\usepackage{amsmath}
\usepackage{amstext,verbatim,graphicx,ifthen,multirow,amsthm,algorithm2e, mathtools}
\usepackage{bm}
\usepackage{natbib}
\usepackage[bf]{caption}
\usepackage{setspace}
\usepackage{enumitem}
\usepackage{placeins}
\usepackage{xcolor}

\usepackage{rotating}
\usepackage{adjustbox}

\newcommand{\tick}{0}
\newcommand{\tock}{1}

\newcommand{\syst}{systematic\ }

\newcommand{\Neff}{N_{eff}}
\newcommand{\Teff}{T_{eff}}

\newcommand{\R}{\mathbb{R}}

\newcommand{\pre}{{\rm pre}}
\newcommand{\post}{{\rm post}}

\newcommand{\homega}{\hat{\omega}}
\newcommand{\hlambda}{\hat{\lambda}}

\newcommand{\mmr}{\mathbb{R}}
\newcommand{\bl}{\mathbf{L}}

\newcommand{\sumi}{\sum_{i=1}^{N-1}}

\newcommand{\sumit}{\sum_{i=1}^N\sum_{t=1}^T }

\newcommand{\did}{{\rm did}}
\newcommand{\asc}{{\rm asc}}
\newcommand{\sdid}{{\rm sdid}}

\newcommand{\scc}{{\rm sc}}

\newcommand{\by}{\mathbf{Y}}
\newcommand{\bbb}{\mathbf{\Upsilon}}
\newcommand{\ba}{\mathbf{\Gamma}}

\newcommand{\btau}{\mathbf{\tau}}
\newcommand{\bw}{\mathbf{W}}

\newcommand{\hmu}{\hat{\mu}}
\newcommand{\hDelta}{\widehat{\Delta}}
\newcommand{\hV}{\widehat{V}}

\newcommand{\ccc}{{\rm co}}

\newcommand{\norm}[1]{\lVert#1\rVert}
\newcommand{\Norm}[1]{\left\lVert#1\right\rVert}

\newcommand{\ttt}{{\rm tr}}

\newcommand{\p}[1]{\left(#1\right)}
\newcommand{\sqb}[1]{\left[#1\right]}
\newcommand{\cb}[1]{\left\{#1\right\}}

\newcommand{\EE}[2][]{\mathbb{E}_{#1}\left[#2\right]}

\newcommand{\PP}[2][]{\mathbb{P}_{#1}\left[#2\right]}
\newcommand{\Var}[2][]{\operatorname{Var}_{#1}\left[#2\right]}

\newcommand{\Cov}[2][]{\operatorname{Cov}_{#1}\left[#2\right]}

\newcommand{\abs}[1]{\left\lvert#1\right\rvert}
\newcommand{\hdelta}{\hat{\delta}}
\newcommand{\tlambda}{\tilde{\lambda}}
\newcommand{\tomega}{\tilde{\omega}}
\newcommand{\slambda}{\tilde \lambda}
\newcommand{\somega}{\tilde \omega}

\newcommand{\sLambda}{\Lambda^{\star}}
\newcommand{\sOmega}{\Omega^{\star}}

\newcommand{\be}{\mathbf{E}}

\newcommand{\oomega}{\tilde\omega}

\newcommand{\olambda}{\tilde\lambda}

\renewcommand{\mathbf}{\boldsymbol}
\renewcommand{\thepage}{}
\renewcommand{\appendix}{\footnotesize\parindent 0cm\setcounter{equation}{0}
\renewcommand{\theequation}{A.\arabic{equation}}
\setcounter{lemma}{0}\renewcommand{\thelemma}{A.\arabic{lemma}}}

\newtheorem{assumption}{Assumption}
\newtheorem{theorem}{Theorem}
\newtheorem{lemma}[theorem]{Lemma}

\newtheorem{corollary}[theorem]{Corollary}

\def\monthname{\ifcase\month\or
  January\or February\or March\or April\or May\or June\or July\or
  August\or September\or October\or November\or December\fi}
\numberwithin{equation}{section}

\DeclareMathOperator*{\argmin}{arg\,min}

\DeclareMathOperator{\rank}{rank}

\DeclareMathOperator{\width}{w}
\DeclareMathOperator{\rad}{rad}

\DeclareMathOperator{\var}{Var}

\DeclareMathOperator{\E}{E}
\DeclareMathOperator{\trace}{tr}

\DeclarePairedDelimiter\set{\{}{\}}

\DeclarePairedDelimiter{\smallabs}{\lvert}{\rvert}

\newcommand{\htau}{\hat{\tau}}

\newcommand\blfootnote[1]{%
  \begingroup
  \renewcommand\thefootnote{}\footnote{#1}%
  \endgroup
}

\def\monthname{\ifcase\month\or
January\or February\or March\or April\or May\or June\or
July\or August\or September\or October\or November\or December\fi}
\renewcommand{\appendix}{\small\parindent 0cm\setcounter{equation}{0}
\renewcommand{\theequation}{A.\arabic{equation}}
\setcounter{lemma}{0}\renewcommand{\thelemma}{A.\arabic{lemma}}
\setcounter{theorem}{0}\renewcommand{\thetheorem}{A.\arabic{theorem}}}

\begin{document}

\title{\textbf{Synthetic Difference in Differences}\blfootnote{{\small We are grateful for
helpful comments and feedback from a co-editor and referees, as well as from Alberto Abadie, Avi Feller, Paul Goldsmith-Pinkham,
Liyang Sun, Yiqing Xu, Yinchu Zhu, and  seminar participants at several venues. This research was generously supported
by ONR grant N00014-17-1-2131 and the Sloan Foundation.
The R package for implementing the methods developed here is available at https://github.com/synth-inference/synthdid.
The associated vignette is at https://synth-inference.github.io/synthdid/.}}}
\author{Dmitry Arkhangelsky\thanks{{\small Associate Professor, CEMFI, Madrid, \texttt{darkhangel@cemfi.es}.}}
\and
Susan Athey\thanks{{\small Professor of
Economics, Graduate School of Business, Stanford University, SIEPR, and NBER,
\texttt{athey@stanford.edu}. }}  
  \and David A.~Hirshberg\thanks{{\small Postdoctoral Fellow, Graduate School of Business and Department of Statistics, Stanford University, \texttt{davidahirshberg@stanford.edu}.}}
 \and Guido W.~Imbens\thanks{{\small Professor of Economics,
Graduate School of Business, and Department of Economics, Stanford University, SIEPR, and NBER,
\texttt{imbens@stanford.edu}.}}
\and Stefan Wager\thanks{{\small Associate Professor of Operations, Information and Technology,
Graduate School of Business, and of Statistics (by courtesy), Stanford University, \texttt{swager@stanford.edu}.}} }
\date{Version \ifcase\month\or January\or February\or March\or April\or May\or June\or July\or August\or September\or October\or November\or December\fi \ \number \year\ \  
}
\maketitle\thispagestyle{empty}

\linespread{1.5}

\begin{abstract}
\singlespacing
\noindent
We present a new estimator for causal effects with panel data that builds on insights behind the widely used difference in differences and synthetic control methods. Relative to these methods we find, both theoretically and empirically, that this ``synthetic difference in differences'' estimator has desirable robustness properties, and that it performs well in settings where the conventional estimators are commonly used in practice. We study the asymptotic behavior of the estimator when the systematic part of the outcome model includes latent unit factors interacted with latent time factors,  and we present conditions for consistency and asymptotic normality.

\vspace{0.4\baselineskip}


\end{abstract}

\baselineskip=20pt
\setcounter{page}{1}
\renewcommand{\thepage}{\arabic{page}}
\renewcommand{\theequation}{\arabic{section}.\arabic{equation}}


\newpage

\section{Introduction}


Researchers are often interested in evaluating
the effects of policy changes using panel data, {\it i.e.}, using repeated observations of units across time, in a setting where
some units are exposed to the policy in some time periods but not others.
These policy changes are frequently not random---neither across units of analysis, nor across
time periods---and  even unconfoundedness given observed covariates may not be credible ({\it e.g., } \citet{imbens2015causal}). In the absence of exogenous variation researchers have focused on statistical models that connect observed data to unobserved counterfactuals.
Many approaches have been developed for this setting but, in practice, a handful of methods
are dominant in empirical work. As documented by \citet*{currie2020technology}, Difference in Differences (DID)
methods have been widely used in applied economics over the last three decades; see also
\citet*{ashenfelter1985using}, \citet*{Bertrand2004did}, and \citet{angrist2008mostly}.
More recently, Synthetic Control (SC) methods, introduced in a
series of seminal papers by Abadie and coauthors \citep*{abadie2003,Abadie2010, abadie2014, abadie2016},
have emerged as an important alternative method for comparative case studies.

Currently these two strategies are often viewed as targeting different types of empirical applications.
In general, DID methods are applied in cases where we have a substantial number of units that are exposed
to the policy, and researchers are willing to make a ``parallel trends'' assumption which implies that we can
adequately control for selection effects by accounting for additive unit-specific and time-specific
 fixed effects. 
In contrast, SC methods, introduced in a setting with only a single (or small number)
of units exposed,  seek to compensate for the lack of parallel trends
by re-weighting units to match their pre-exposure trends. 
 

In this paper, we argue that although the empirical settings where DID and 
SC methods
are typically used differ, the fundamental assumptions that justify both methods are closely related.
We then propose a new method, Synthetic Difference in Differences (SDID), that combines attractive features of both. Like
SC, our method re-weights and matches pre-exposure trends to weaken the reliance on  parallel trend type 
assumptions. Like DID, our method is invariant to additive unit-level shifts, and allows
for valid large-panel inference. Theoretically, we establish consistency and asymptotic normality of our
estimator. Empirically, we find that our method is competitive with (or dominates) DID
in applications where DID methods have been used in the past, and likewise is competitive with (or dominates)
SC in applications where SC methods have been used in the past.

To introduce the basic ideas, 
consider a balanced panel with $N$ units and $T$ time periods, where the outcome for unit $i$ in period $t$ is denoted by $Y_{it}$,
and exposure to the binary treatment is denoted by $W_{it}\in\{0,1\}$. Suppose moreover that the first $N_\ccc$ (control) units are never
exposed to the treatment, while the last $N_\ttt = N - N_\ccc$ (treated) units are exposed after time $T_\pre$.\footnote{Throughout the main part of our analysis, we focus on the block treatment assignment case where
$W_{it} = 1\p{\cb{i > N_\ccc, \,t > T_\pre}}$. In the closely related staggered adoption case (\citet*{athey2021design}) where
units adopt the treatment at different times, but remain exposed after they first adopt the treatment, one can modify the methods
developed here.
See Section \ref{sec:staggered} in the Appendix for details.}
Like with
SC methods, we start by finding weights $\homega^\sdid$ that align pre-exposure trends in the outcome of unexposed units with those
for the exposed units, {\it e.g.,}
$\sum_{i = 1}^{N_\ccc} \homega^\sdid_i Y_{it} \approx N_\ttt^{-1} \sum_{i = N_\ccc + 1}^{N} Y_{it}$ for all $t=1,\,\ldots,\,T_\pre$.
We also look for time weights $\hlambda^\sdid_t$ that  balance pre-exposure time periods with
post-exposure ones (see Section \ref{sec:calif} for details).
Then we use these weights in a basic two-way fixed effects regression
to estimate the average causal effect of exposure (denoted by $\tau$):\footnote{This estimator 
also has an interpretation as a difference-in-differences of weighted averages of observations.
See Equations~\ref{eq:tauhat_def}-~\ref{eq:delta_def} below.}
\begin{equation}
\label{main_sdid}
\p{\hat\tau^\sdid, \, \hat\mu, \, \hat\alpha, \, \hat\beta}=
\argmin_{\tau,\mu,\alpha,\beta}  \cb{ \sumit \Bigl( Y_{it}-\mu-\alpha_i-\beta_t-W_{it}\tau\Bigr)^2\homega_i^\sdid\hlambda_t^\sdid}.
\end{equation}
In comparison,  DID estimates the effect of treatment exposure by solving the same two-way fixed effects regression problem without either time or unit weights:
\begin{equation}
\label{main_did}
\p{\hat\tau^\did, \, \hat\mu, \, \hat\alpha, \, \hat\beta}=
\argmin_{\alpha,\beta,\mu,\tau} \cb{\sumit \Bigl( Y_{it}-\mu-\alpha_i-\beta_t-W_{it}\tau\Bigr)^2}.
\end{equation}
The  use of weights in the SDID estimator  effectively makes the two-way fixed effect regression ``local,'' in that it emphasizes (puts more weight on) units that on average are similar
in terms of their past to the target (treated) units, and it emphasizes periods that are on average similar to the target (treated)  periods.

This localization can bring two benefits relative to the standard DID estimator. Intuitively, using only similar units and similar periods makes the estimator more
robust. For example, if one is interested in estimating the effect of anti-smoking legislation on California (\citet*{Abadie2010}),
or the effect of German reunification on West Germany (\citet*{abadie2014}), or the effect of the Mariel boatlift on Miami (\citet*{cardmariel}, 
\citet*{peri2019labor}), it is natural to emphasize  states, countries or cities that are similar to California, West Germany,
or Miami respectively relative to states, countries or cities that are not. Perhaps less intuitively, the use of the weights can also improve the estimator's precision by implicitly removing
systematic (predictable) parts of the outcome. However, the latter is not guaranteed: If there is little systematic heterogeneity in outcomes
by either units or time periods, the unequal weighting of units and time periods may worsen the precision of the estimators relative to the DID estimator.

Unit weights are designed so that the average outcome for the treated units is approximately parallel to the weighted average for control units. Time weights are designed so that
the average post-treatment outcome for each of the control units differs by a constant from   the weighted average of the pre-treatment outcomes for the same control units.
Together, these weights make the DID strategy more plausible. This idea is not far from the current empirical practice. Raw data rarely exhibits parallel time trends for treated and control units, and researchers use different techniques, such as adjusting for covariates or selecting appropriate time periods to address this problem ({\it e.g.}, \citet{abadie2005semiparametric, callaway2020difference}). Graphical evidence that is used to support the parallel trends assumption is then based on the adjusted data. SDID makes this process automatic and applies a similar logic to weighting both units and time periods, all while retaining statistical guarantees. From this point of view, SDID addresses pretesting concerns recently expressed in \citet{roth2018pre}.

In comparison with the SDID estimator, the SC estimator  omits the unit fixed effect and the time weights from the regression function:
\begin{equation}
\label{main_sc}
\p{\hat\tau^\scc, \, \hat\mu, \, \hat\beta}=
\argmin_{\mu,\beta,\tau} \cb{\sumit  \Bigl( Y_{it}-\mu-\beta_t-W_{it}\tau\Bigr)^2\homega^\scc_i}.
\end{equation} 
The argument for including time weights in the SDID estimator is the same as the argument for including the unit weights presented earlier:
The time weight can both remove bias and
improve precision by eliminating the role of time periods that are very different from the post-treatment periods.
Similar to the argument for the use of weights, the argument for the inclusion of the unit fixed effects is twofold. First, by making
the model more flexible, we strengthen its robustness properties. Second, as demonstrated in the application 
and simulations based on real data,
these unit fixed effects often explain much of the variation in
outcomes and can improve precision. Under some conditions, SC weighting can account for the
unit fixed effects on its own. In particular, this happens when the weighted average of the outcomes for the control
units in the pre-treatment periods is exactly equal to the average of outcomes for the treated
units during those pre-treatment periods. In practice, this equality holds only approximately,
in which case including the unit fixed effects in the weighted regression will remove some of the
remaining bias. The benefits of including unit fixed effects in the SC regression \eqref{main_sc} can also be obtained by applying the synthetic control method after centering the data by subtracting, from each unit's trajectory, its pre-treatment mean. 
This estimator was previously suggested 
in \citet{doudchenko2016balancing} and \citet{ferman2019synthetic}.
To separate out the benefits of allowing for fixed effects from those stemming from the use of time-weights, we include in our application and simulations  this DIFP  estimator.

\section{An Application}
\label{sec:calif}

To get a better understanding of how $\htau^\did$, $\htau^\scc$ and $\htau^\sdid$ compare to each other,
we first revisit the California smoking cessation program example of \citet*{Abadie2010}.
The goal of their analysis was to estimate the effect of increased cigarette taxes on smoking in California.
We consider observations for 39 states (including California) from 1970 through 2000. California
passed Proposition 99 increasing cigarette taxes (i.e., is treated) from 1989 onwards.
Thus, we have $T_\pre=19$ pre-treatment periods, $T_\post=T-T_\pre=12$ post-treatment periods, $N_\ccc=38$ unexposed states,
and $N_\ttt=1$ exposed state (California).

\subsection{Implementing SDID} \label{sec:implementing}

Before presenting results on the California smoking case, we discuss in detail how we choose the synthetic control
type weights $\homega^\sdid$ and $\hlambda^\sdid$ used for our estimator as specified in \eqref{main_sdid}. Recall that,
at a high level, we want to choose the unit weights to roughly match pre-treatment trends of unexposed units with those for the exposed ones,
$\sum_{i = 1}^{N_\ccc} \homega^\sdid_i Y_{it} \approx N_\ttt^{-1} \sum_{i = N_\ccc + 1}^{N} Y_{it}$ for all $t=1,\,\ldots,\,T_\pre$,
and similarly we want to choose the time weights to balance pre- and post-exposure periods for unexposed units.

In the case of the unit weights  $\homega^\sdid$, we  implement this by solving the optimization problem
\begin{equation}
\label{unit_weights}
\begin{aligned}
&\p{\homega_0, \, \homega^\sdid} = \argmin_{ \omega_0 \in \mmr, \omega\in \Omega} \ell_{unit}(\omega_0, \omega) \quad \text{ where } \\
& \quad \quad \quad\ell_{unit}(\omega_0, \omega) = \sum_{t=1}^{T_\pre} 
\left( \omega_0 + \sum_{i=1}^{N_\ccc} \omega_{i} Y_{it} -  
\frac{1}{N_\ttt} \sum_{i = N_\ccc + 1}^N Y_{it}\right)^2 + \zeta^2 T_{\pre} \left\|\omega\right\|_2^2, \\
& \quad \quad \quad \Omega = \cb{\omega \in \mmr_+^N :   \sum_{i=1}^{N_\ccc} \omega_i=1, \,  \omega_i=N_\ttt^{-1} \text{ for all } i=N_\ccc+1, \, \ldots \, ,N},
\end{aligned}
\end{equation}
where $\mmr_+$ denotes the positive real line. We set the regularization parameter $\zeta$ as
\begin{equation}
\label{eq:zeta_calif}
\begin{split}
&\zeta= (N_\ttt T_\post)^{1/4} \ \hat \sigma \ \text{ with }\  \hat \sigma^2 = \frac{1}{N_\ccc (T_\pre-1)}\sum_{i=1}^{N_\ccc}\sum_{t=1}^{T_\pre-1} \p{\Delta_{it}-\overline{\Delta}}^2, \\
&\mathrm{where}\ \  \Delta_{it}=Y_{i(t+1)}-Y_{it},\hskip0.5cm {\rm and}\ \  \overline{\Delta}=\frac{1}{N_\ccc (T_\pre-1)}\sum_{i=1}^{N_\ccc}\sum_{t=1}^{T_\pre-1} \Delta_{it}.
\end{split}
\end{equation}
That is, we choose the regularization parameter $\zeta$ to match the size of a typical one-period outcome change $\Delta_{it}$ for unexposed units in the pre-period, multiplied by a theoretically motivated scaling $(N_\ttt T_\post)^{1/4}$.
The SDID weights $\homega^\sdid$ are closely related to the weights  used in \citet*{Abadie2010}, with  two minor differences.
First, we allow for an intercept term $\omega_0$, meaning that the weights $\homega^\sdid$ no longer need to make
the unexposed pre-trends perfectly match the exposed ones; rather, it is sufficient that the weights make the trends parallel. The reason we
can allow for this extra flexibility in the choice of weights is that our use of fixed effects $\alpha_i$ will absorb any constant
differences between different units. Second, following \citet{doudchenko2016balancing}, we add a regularization penalty
to increase the dispersion, and ensure the uniqueness, of the weights. If we were to omit the intercept
$\omega_0$ and set $\zeta = 0$, then \eqref{unit_weights} would correspond exactly to a choice of weights discussed
in  \citet{Abadie2010} in the case where $N_\ttt = 1$.

We implement this for the  time weights $\hlambda^\sdid$ by solving\footnote{The weights $\hlambda^\sdid$ may not be uniquely defined, as $\ell_{time}$ can have multiple minima. In principle our results hold for any argmin of $\ell_{time}$. These tend to be similar in the setting we consider, as they all converge to unique `oracle weights' $\slambda^\sdid$ that are discussed in Section~\ref{sec:oracle-and-adaptive}. In practice, to make the minimum defining our time weights unique, we add a very small regularization term $\zeta^2 N_{\ccc} \norm{\lambda}^2$ 
to $\ell_{time}$, taking $\zeta = 10^{-6} \ \hat \sigma$ for $\hat \sigma$ as in \eqref{eq:zeta_calif}. 
}
\begin{equation}
\label{time_weights}
\begin{aligned}
&\p{\hlambda_0, \, \hlambda^\sdid} = \argmin_{ \lambda_0 \in \mmr, \lambda \in \Lambda}  \ell_{time}(\lambda_0, \lambda) \quad \text{ where } \\
&\quad \quad \quad \ell_{time}(\lambda_0, \lambda) = \sum_{i=1}^{N_\ccc} 
\left( \lambda_0 + \sum_{t=1}^{T_\pre} \lambda_{t} Y_{it} -  
\frac{1}{T_\post} \sum_{t = T_\pre + 1}^T Y_{it}\right)^2 , \\
&\quad \quad \quad \Lambda=\cb{\lambda\in\mmr_+^T :  \sum_{t=1}^{T_\pre} \lambda_t=1, \, \lambda_t=T_\post^{-1} \text{ for all } t=T_\pre+1,\, \ldots \, ,T}.
\end{aligned}
\end{equation}
The main difference between \eqref{unit_weights} and \eqref{time_weights} is that we use regularization for the
former but not the latter. This choice is motivated by our formal
results, and reflects the fact we allow for correlated observations within time periods for the same unit, but not across units within 
a time period, beyond what is captured by the systematic component of outcomes as represented by a latent factor model.

\RestyleAlgo{boxruled}
\LinesNumbered
\begin{algorithm}[t]
 \KwData{$\mathbf{Y}, \mathbf{W}$}
 \KwResult{Point estimate $\hat \tau^\sdid$}
Compute regularization parameter $\zeta$ using \eqref{eq:zeta_calif}\;
Compute unit weights $\homega^\sdid$ via \eqref{unit_weights}\;
Compute time weights $\hlambda^\sdid$ via \eqref{time_weights}\;
Compute the SDID estimator via the weighted DID regression
\begin{equation*}
\label{main_sdid_cov}
\p{\hat\tau^\sdid, \, \hat\mu, \, \hat\alpha, \, \hat\beta}=
\argmin_{\tau,\mu,\alpha,\beta}  \cb{ \sumit \Bigl( Y_{it}-\mu-\alpha_i-\beta_t-W_{it}\tau\Bigr)^2\homega_i^\sdid\hlambda_t^\sdid};
\end{equation*}
 \caption{Synthetic Difference in Differences (SDID)}
 \label{alg:sdid}
\end{algorithm}

We summarize our procedure as Algorithm \ref{alg:sdid}.\footnote{Some applications feature time-varying
exogenous covariates $X_{it} \in \mathbb{R}^p$. 
We can incorporate adjustment for these covariates by applying SDID to the residuals $Y_{it}^{\mathrm{res}} = Y_{it}-X_{it}\hat\beta$ of the regression of $Y_{it}$ on $X_{it}$.}
In our application and simulations we also report 
the SC and DIFP estimators. Both of these use 
weights solving \eqref{unit_weights} without regularization. The SC estimator also omits the intercept $\omega_0$.\footnote{Like the time weights $\hat\lambda^{\sdid}$, the unit weights 
for the SC and DIFP estimators may not be uniquely defined. To ensure uniqueness in practice, we take $\zeta = 10^{-6} \ \hat \sigma$, not $\zeta=0$, in $\ell_{unit}$. 
In our simulations, SC and DIFP with this minimal form of regularization outperform more strongly regularized variants with $\zeta$ as in \eqref{eq:zeta_calif}. We show this comparison in Table \ref{table_reg_vs_unreg}.} 
Finally, we report results for the matrix completion (MC)  estimator  proposed by \citet{athey2017matrix}, which is based on imputing the missing $Y_{it}(0)$ using a low rank factor model with nuclear norm regularization.

\subsection{The California Smoking Cessation Program}

The results from running this analysis
are shown in Table \ref{estimates_smoking}. As argued in \citet{Abadie2010}, the assumptions underlying
the DID estimator are suspect here, and the -27.3 point estimate likely overstates the
effect of the policy change on smoking. SC provides a reduced (and generally considered more credible) estimate
of -19.6. The other methods, our proposed SDID, the DIFP and the MC estimator are all smaller than the DID estimator with the SDID and DIFP estimator substantially smaller than the SC estimator. At the very least, this difference in point estimates implies that the
use of time weights and unit fixed effects in \eqref{main_sdid} materially affects conclusions; and, throughout this
paper, we will argue that when $\htau^\scc$ and $\htau^\sdid$ differ, the latter is often more credible.
Next, and perhaps surprisingly, we see that the  standard errors obtained for SDID
(and also for SCIFP, and MC) are smaller than those for DID, despite our method
being more flexible. This is a result of the local fit of SDID (and SC) being improved by the weighting.

\begin{table}[t]
\begin{center}
\begin{tabular}{|l|ccccc|}
  \hline
 & SDID & SC & DID & MC & DIFP \\ 
     \hline 
Estimate &
-15.6 & -19.6 &-27.3	& -20.2	& -11.1 \\
Standard error &
  (8.4)	  &  (9.9)& (17.7) & (11.5)	  	&  (9.5) 
\\
\hline
\end{tabular}
\caption{Estimates for average effect of increased cigarette taxes on California per capita cigarette sales over twelve post-treatment years,
based on 
synthetic difference in differences (SDID),
 synthetic controls (SC),
difference in differences (DID),  
 matrix completion (MC), synthetic control with intercept (DIFP), 
along with estimated standard errors. We use the `placebo method' standard error estimator discussed in Section~\ref{sec:inference}.}
\label{estimates_smoking}
 \end{center}
  \end{table}

To facilitate direct comparisons, we observe that each of the three estimators can be rewritten as a weighted average difference in adjusted outcomes $\hdelta_i$ for appropriate sample weights $\homega_i$:
\begin{equation}
    \label{eq:tauhat_def}
\htau =  \hdelta_\ttt - \sum_{i = 1}^{N_\ccc} \homega_i \hdelta_i 
\quad \text{ where }\quad \hdelta_\ttt = \frac{1}{N_{\ttt}}\sum_{i=N_\ccc+1}^N \hdelta_i.
\end{equation}
DID uses constant weights $\homega_i^\did = N_\ccc^{-1}$, while the construction of SDID and SC weights is outlined in Section \ref{sec:implementing}.
For the adjusted outcomes $\hdelta_i$, SC uses unweighted treatment period averages,
DID uses unweighted differences between average treatment period and pre-treatment outcomes, 
and SDID uses weighted differences of the same.
\begin{equation}
\label{eq:delta_def}
\begin{split}
&\hdelta_i^\scc = \frac{1}{T_\post} \sum_{t = T_\pre + 1}^T Y_{it}, \\
&\hdelta_i^\did = \frac{1}{T_\post} \sum_{t = T_\pre + 1}^T Y_{it} - \frac{1}{T_\pre} \sum_{t = 1}^{T_\pre} Y_{it}, \\
&\hdelta_i^\sdid = \frac{1}{T_\post} \sum_{t = T_\pre + 1}^T Y_{it} -  \sum_{t = 1}^{T_\pre} \hlambda_t^\sdid Y_{it}.
\end{split}
\end{equation}

The top panel of Figure \ref{californiadiagrams} illustrates how each method operates.
As is well known \citep{ashenfelter1985using}, DID relies on the assumption that cigarette sales in different states would have evolved
in a parallel way absent the intervention. Here, pre-intervention trends are obviously not parallel, so the DID
 estimate should be considered suspect. In contrast, SC re-weights
the unexposed states so that the weighted of outcomes for these states match California pre-intervention as close as possible, and then attributes
any post-intervention divergence of California from this weighted average to the intervention.
What SDID does here is to re-weight the unexposed control units to make their time trend
parallel (but not necessarily identical) to California pre-intervention, and then applies a DID
 analysis to this re-weighted panel. Moreover, because of the time weights, we only
focus on a subset of the  pre-intervention time periods when carrying out this last step. These time periods were
selected so that the weighted average of historical outcomes predict average treatment period outcomes for control units, up to a constant.
It is useful to contrast the data-driven SDID approach to selecting the time weights to both DID, where all pre-treatment periods are given equal weight, and to event studies where typically the last pre-treatment period is used as a comparison and so implicitly gets all the weight ({\it e.g.,} \citet{borusyak2016revisiting, freyaldenhoven2019pre}).

\begin{figure}[p]
\begin{center}
\begin{tabular}{cccc} & \
\makebox[0.3\textwidth]{Difference in Differences} & 
\makebox[0.3\textwidth]{Synthetic Control} & 
\makebox[0.3\textwidth]{Synthetic Diff.~in Differences} 
\end{tabular}

\begin{minipage}{.02\textwidth}
\begin{tabular}{c}
\rotatebox{90}{\tiny \qquad\qquad cigarette consumption (packs/year)} \\
\rotatebox{90}{\tiny \quad difference in consumption (packs/year) \quad}
\end{tabular}
\end{minipage}
\begin{minipage}{.97\textwidth}
\includegraphics[width=\textwidth]{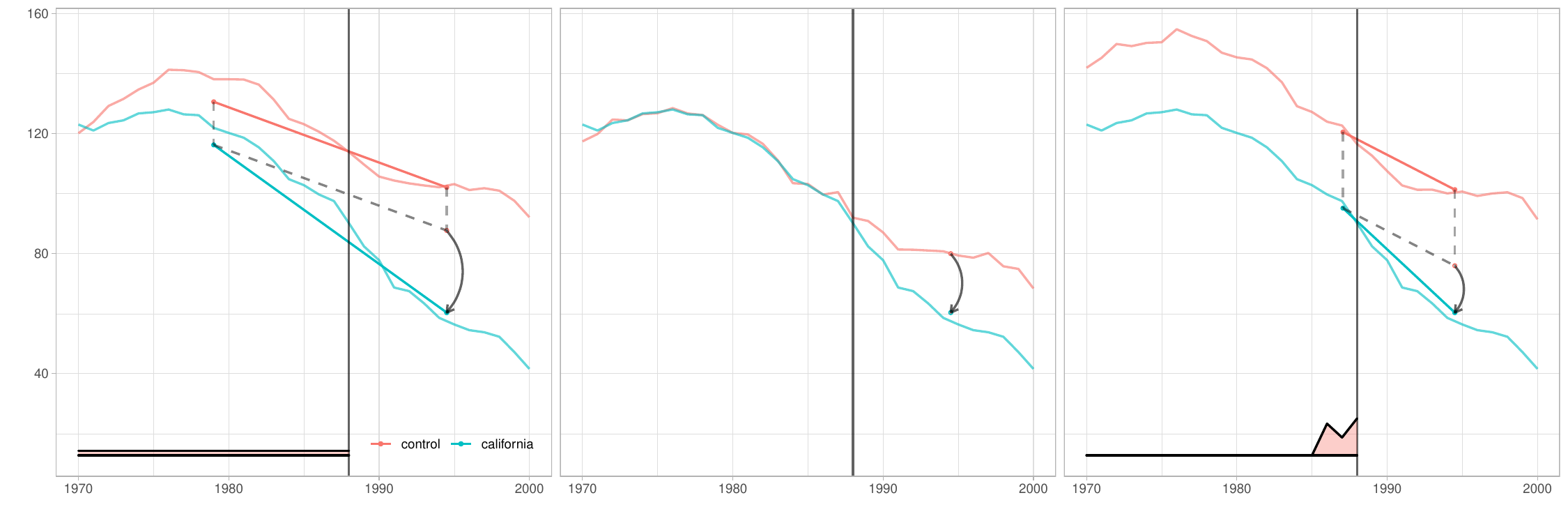} 
\includegraphics[width=\textwidth]{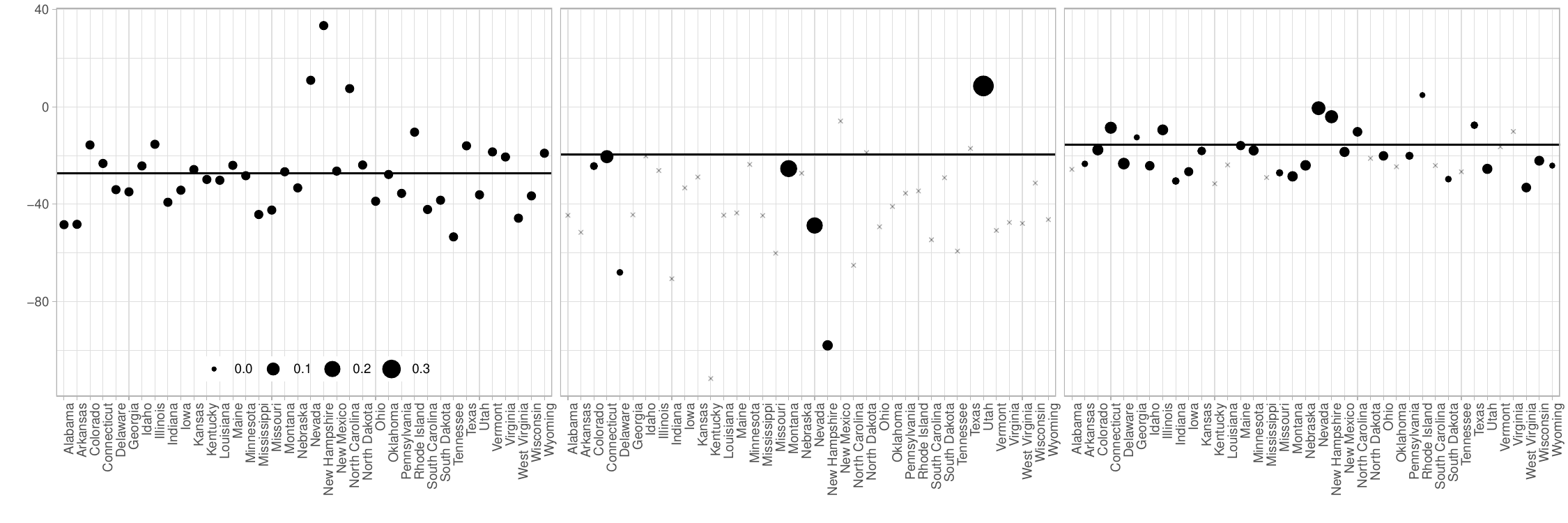}
\end{minipage}
\caption{A comparison between difference-in-differences, synthetic control, and synthetic differences-in-differences estimates 
for the effect of California Proposition 99 on per-capita annual cigarette consumption (in packs/year). 
In the first row, we show trends in consumption over time for California and the relevant weighted average of control states,
with the weights used to average pre-treatment time periods at the bottom of the graphs.
The estimated effect is indicated by an arrow.
In the second row, we show the state-by-state adjusted outcome difference $\hdelta_\ttt -\hdelta_i$ as specified in \eqref{eq:tauhat_def}-\eqref{eq:delta_def}, with the weights $\homega_i$ indicated by dot size and the weighted average of these differences --- the estimated effect --- indicated by a horizontal line. Observations with zero weight are denoted by an $\times$-symbol.
\label{californiadiagrams}}
\end{center}
\end{figure}

The lower panel of Figure \ref{californiadiagrams} plots $\hdelta_\ttt - \hdelta_i$ for each method and for each unexposed state,
where the size of each point corresponds to its weight $\homega_i$; observations with zero weight are denoted by
an $\times$-symbol. As discussed in \citet*{Abadie2010}, the
SC weights $\homega^\scc$ are sparse. The SDID weights $\homega^\sdid$ are also sparse---but
less so. This is due to regularization and the use of the intercept $\omega_0$, which allows greater flexibility in solving \eqref{unit_weights}, enabling more balanced weighting. Observe that both DID
and SC have some very high influence states, that is, states with large absolute values of $\homega_i (\hdelta_\ttt - \hdelta_i)$
({\it e.g.}, in both cases, New Hampshire). In contrast, SDID does not give any state particularly high influence, 
suggesting that after weighting, we have achieved the desired ``parallel trends'' as illustrated in the top panel of Figure \ref{californiadiagrams} without inducing excessive variance in the estimator by using concentrated weights.

\section{Placebo  Studies}
\label{sec:simu}

So far, we have relied on conceptual arguments to make the claim that SDID
inherits good robustness properties from both traditional DID and SC methods,
and shows promise as a method that can be used is settings where either DID and SC
would traditionally be used. The goal of this section is to see how these claims play out in realistic empirical settings. To this end, we consider
two carefully crafted simulation studies,  calibrated to datasets representative of those typically used for panel data studies. The first simulation study mimics settings where DID would be used in practice (Section \ref{sec:cps}), while the second mimics settings suited to SC (Section \ref{sec:penn}). Not only do we base the outcome model of our simulation study on real datasets, we further ensure
that the treatment assignment process is realistic by seeking to emulate the distribution of real policy initiatives. To be specific,
in Section \ref{sec:cps}, we consider a panel of US states. We estimate several alternative treatment assignment models to create the hypothetical treatments, where the models are based on the state laws related to minimum wages,
abortion or gun rights.

In order to run such a simulation study, we first need to commit to an econometric specification
that can be used to assess the accuracy of each method. Here, we work with the following latent factor model
(also referred to as an ``interactive fixed-effects model'', \citet{xu2017generalized}, see also \citet{athey2017matrix}),
\begin{equation}
\label{basic_model_scalar} 
Y_{it}=\boldsymbol{\gamma}_i \boldsymbol{\upsilon}_t^\top +\tau W_{it} +\varepsilon_{it},
\end{equation}
where $\boldsymbol{\gamma}_i$ is a vector of latent unit factors of dimension $R$, and $\boldsymbol{\upsilon}_t$ is a vector of latent time factors of dimension $R$. In matrix form, this can be written 
\begin{equation}
\label{basic_model} 
\by=\bl+\tau\bw+\be
\quad \text{ where } \quad
\bl=\boldsymbol{\Gamma \Upsilon}^\top.
\end{equation}
  We refer to $\be$ as the idiosyncratic component or error matrix, and to $\bl$ as the \syst component.
We assume that
the conditional expectation of the error matrix $\be$ given the assignment matrix $\bw$ and the \syst component $\bl$ is zero. That is,
 the treatment assignment cannot depend on $\be$. However, the treatment assignment may in general depend on the
systematic component $\bl$ ({\it i.e.}, we do not take $\bw$ to be randomized). We assume that $\be_i$ is independent of $\be_{i'}$ for each pair of units $i,i'$, but we allow for correlation across time periods within a unit.  Our goal is to estimate the treatment effect $\tau$.

The model \eqref{basic_model} captures several qualitative challenges that have received considerable attention
in the recent panel data literature. When the matrix $\bl$ takes on an additive form, i.e., $L_{it} = \alpha_i + \beta_t$,
then the DID regression will consistently recover $\tau$.  Allowing for interactions
in $\bl$ is a natural way to generalize the fixed-effects specification and discuss inference in settings where
DID is misspecified \citep{bai2009panel,moon2015linear,moon2017dynamic}. In our formal results given in
Section \ref{sec:formal}, we show how, despite not explicitly fitting the model \eqref{basic_model}, SDID can consistently estimate $\tau$ in this design under reasonable conditions. Finally, accounting for correlation over time within observations of the same unit is widely considered to be an important ingredient to credible inference
using panel data \citep*{angrist2008mostly, Bertrand2004did}.

In our experiments, we compare DID, SC, SDID, and DIFP, all implemented exactly as in Section \ref{sec:calif}.
We also compare these four estimators to an alternative that
estimates $\tau$ by directly fitting both $\bl$ and $\tau$ in \eqref{basic_model}; specifically, we consider
the matrix completion (MC) estimator recommended in \citet*{athey2017matrix} which uses nuclear norm
penalization to regularize its estimate of $\bl$.
In the remainder of this section, we focus on comparing the bias and root-mean-squared error of the estimator.
We discuss questions around inference and coverage in Section \ref{sec:inference}.

\subsection{Current Population Survey Placebo Study}
\label{sec:cps}

Our first set of simulation experiments revisits the landmark placebo study of \citet*{Bertrand2004did} using
the Current Population Survey (CPS). The main goal of \citet{Bertrand2004did} was to study the behavior
of different standard error estimators for DID. To do so, they randomly assigned a subset of states in the CPS
dataset to a placebo treatment and the rest to the control group, and examined how well different approaches
to inference for  DID estimators covered the true treatment effect of zero. Their main finding was that only methods that were robust to serial correlation of
repeated observations for a given unit ({\it e.g.,} methods that clustered observations by unit) attained valid coverage.

We modify  the placebo analyses in \citet{Bertrand2004did} in two ways. First, we no longer assigned exposed states
completely at random, and instead use a non-uniform assignment mechanism that is inspired by different policy
choices actually made by different states. Using a non-uniformly random assignment is important because it allows us to differentiate between various estimators in ways that completely random assignment would not. Under completely random assignment, a number of methods, including DID, perform well because the presence of $\bl$ in
\eqref{basic_model} introduces zero bias. In contrast, with a non-uniform random assignment ({\it i.e.},
treatment assignment is correlated with systematic effects), methods that
do not account for the presence of $\bl$ will be biased. 
Second, we simulate values for the outcomes based on a model estimated on the CPS data, in order to have more control over the data generating process.

\subsubsection{The Data Generating Process}
\label{sec:cps_dgp}

For the first set of simulations we use as the starting point data on wages for women with positive wages in the March outgoing rotation
groups in the Current Population Survey (CPS) for the years 1979 to 2019. We first transform these by taking logarithms and then average them by state/year cells.
Our simulation design has two components, an outcome model and an assignment model. We generate outcomes
via a simulation that seeks to capture the behavior of the average by state/year of the logarithm of wages for those
with positive hours worked in the CPS data as in \citet{Bertrand2004did}. Specifically, we simulate data using the
model \eqref{basic_model}, where the rows $\be_{i}$ of $\be$
have a multivariate Gaussian distribution $\be_{i} \sim \mathcal{N}(0, \Sigma)$, and
we choose both $\bl$ and $\Sigma$ to fit the CPS data as follows. First, we fit a rank four factor model for $\bl$:
\begin{equation}
\bl :=\argmin_{L: \text{rank}(L) = 4} \sum_{it}(Y^*_{it} - L_{it})^2,
\end{equation}
where $Y^*_{it}$ denotes the true state/year average of  log-wage in the CPS data. We then estimate
$\Sigma$ by fitting an AR(2) model to the residuals of $Y^*_{it} - L_{it}$. For purpose of interpretation,
we further decompose the systematic component $\bl$ into an additive (fixed effects) term $\mathbf{F}$ and
an interactive term $\mathbf{M}$, with
\begin{equation}
\begin{aligned}
&F_{it} = \alpha_i+\beta_t=\frac{1}{T}\sum_{l=1}^T L_{il} + \frac{1}{N}\sum_{j=1}^N L_{jt} - \frac{1}{NT} \sum_{it}L_{it},\\
& M_{it}=L_{it}-F_{it}.
\end{aligned}
\end{equation}
This decomposition of $\bl$ into an additive two-way fixed effect component  $\mathbf{F}$ and an interactive component $\mathbf{M}$  enables us to study the sensitivity
of different estimators to the presence of different types of systematic effects.

Next we discuss generation of the treatment assignment. Here, we are designing a ``null effect'' study, meaning that
treatment has no effect on the outcomes and all methods should estimate zero. However, to make this more challenging, we choose the treated units so that the assignment mechanism is correlated with the systematic component $\bl$.
We set $W_{it} = D_i \mathbf{1}_{t > T_0}$, where $D_i$ is a binary exposure indicator generated as
\begin{equation}
\label{eq:assignment_lr}
D_i \,\big|\, \be_{i}, \alpha_i, \mathbf{M}_i \sim \text{Bernoulli}\p{\pi_i}, \ \ \ \
\pi_i = \pi(\alpha_i, \mathbf{M}_i;\phi) = \frac{\exp(\phi_\alpha \alpha_i+\phi_M \mathbf{M}_i)}
{1+\exp(\phi_\alpha \alpha_i+\phi_M \mathbf{M}_i)}.
\end{equation}
In particular, the distribution of $D_i$ may depend on $\alpha_i$ and $\mathbf{M}_i$; however, $D_i$ is independent of $\be_i$, i.e., the assignment is strictly exogenous.\footnote{In the simulations below, we restrict the maximal number of treated units (either to $10$ or $1$). To achieve this, we first sample $D_i$ independently and accept the results if the number of treated units satisfies the constraint. If it does not, then we choose the maximal allowed number of treated units from those selected in the first step uniformly at random.}
To construct probabilities $\{\pi_i\}$ for this assignment model, we choose $\phi$ as the coefficient estimates from a logistic regression of an observed  binary characteristic of the state $D_i$ on $\mathbf{M}_i$ and $\alpha_i$. We consider three different choices for $D_i$, relating to minimum wage laws, abortion rights, and gun control laws.\footnote{See Section~\ref{sec:placebo-study-details} in the appendix for details.} As a result, we get assignment probability models that reflect actual differences across states with respect to important economic variables. In practice the $\alpha_i$ and $\mathbf{M}_i$ that we construct predict a sizable part of variation in $D_i$, with $R^2$ varying from $15\%$ to $30\%$.

\begin{table}[t]
\begin{center}
\begin{adjustbox}{width=1\textwidth}
\begin{tabular}{|l|rrrr|rrrrr|rrrrr|}
\hline
& \multirow{2}{*}{$\frac{\|\mathbf{F}\|_F}{\sqrt{NT}}$} &  \multirow{2}{*}{$\frac{\|\mathbf{M}\|_F}{\sqrt{NT}}$} & \multirow{2}{*}{$\sqrt{\frac{\trace(\Sigma)}{T}}$} & \multirow{2}{*}{AR(2)} &
\multicolumn{5}{c|}{RMSE}& \multicolumn{5}{c|}{Bias}\\ 
& &&  & &SDID & SC & DID &MC & DIFP & SDID & SC & DID &MC& DIFP \\ \hline
&&&& &&&& &&&& &&\\
Baseline & 0.992 & 0.100 & 0.098 & (.01,-.06) & \bf 0.028 & 0.037 & 0.049 & 0.035 & 0.032 & 0.010 & 0.020 & 0.021 & 0.015 & 0.007 \\ \hline
\multicolumn{15}{|l|}{\it Outcome Model}\\
No Corr& 0.992 & 0.100 & 0.098 & (.00, .00) & \bf  0.028 & 0.038 & 0.049 & 0.035 & 0.032 & 0.010 & 0.020 & 0.021 & 0.015 & 0.007 \\ 
No $\mathbf{M}$ & 0.992 & 0.000 & 0.098 & (.01, -.06) & 0.016 & 0.018 & \bf  0.014 & \bf  0.014 & 0.016 & 0.001 & 0.004 & 0.001 & 0.001 & 0.001 \\ 
No $\mathbf{F}$ & 0.000 & 0.100 & 0.098 & (.01, -.06) & 0.028 & \bf  0.023 & 0.049 & 0.035 & 0.032 & 0.010 & 0.004 & 0.021 & 0.015 & 0.007\\ 
Only Noise&  0.000 & 0.000 & 0.098 & (.01, -.06) & 0.016 & \bf 0.014 & \bf  0.014 & \bf 0.014 & 0.016 & 0.001 & 0.001 & 0.001 & 0.001 & 0.001 \\ 
No Noise & 0.992 & 0.100 & 0.000 & (.00, .00) & 0.006 & 0.017 & 0.047 & \bf 0.004 & 0.011 & 0.004 & 0.004 & 0.020 & 0.000 & 0.001 \\ \hline
 \multicolumn{15}{|l|}{\it Assignment Process}\\
 Gun Law &0.992 & 0.100 & 0.098 & (.01, -.06) &\bf  0.026 & 0.027 & 0.047 & 0.035 & 0.030 & 0.008 & -0.003 & 0.015 & 0.015 & 0.009 \\ 
 Abortion & 0.992 & 0.100 & 0.098 & (.01, -.06) & \bf 0.023 & 0.031 & 0.045 & 0.031 & 0.027 & 0.004 & 0.016 & 0.003 & 0.003 & 0.001 
 \\ Random & 0.992 & 0.100 & 0.098 & (.01, -.06)& \bf 0.024 & 0.025 & 0.044 & 0.031 & 0.027 & 0.001 & -0.001 & 0.002 & 0.001 & -0.000 \\ \hline
 \multicolumn{15}{|l|}{\it Outcome Variable}\\
 Hours & 0.789 & 0.402 & 0.575 & (.06, .00) & 0.190 & 0.203 & 0.206 & \bf 0.185 & 0.197 & 0.111 & -0.049 & 0.085 & 0.100 & 0.099  \\ 
 U-rate &  0.752 & 0.441 & 0.593 & (-.02, -.01) & 0.191 & \bf 0.184 & 0.353 & 0.247 & 0.187 & 0.100 & 0.080 & 0.304 & 0.187 & 0.078\\
 \hline
 \multicolumn{15}{|l|}{\it Assignment Block Size}\\
  $T_\post =1$ & 0.992 & 0.100 & 0.098 &  (.01, -.06) & \bf 0.050 & 0.059 & 0.070 & 0.051 & 0.054 & 0.019 & 0.017 & 0.038 & 0.021 & 0.012\\ 
  $N_\ttt = 1$ & 0.992 & 0.100 & 0.098 & (.01, -.06) & \bf 0.063 & 0.072 & 0.126 & 0.081 & 0.083 & 0.002 & 0.014 & 0.011 & 0.004 & -0.002  \\ 
 $ T_\post = N_\ttt = 1$ & 0.992 & 0.100 & 0.098 & (.01, -.06) & 0.112 & 0.124 & 0.153 & \bf 0.108 & 0.117 & 0.014 & 0.024 & 0.033 & 0.016 & 0.011   \\ \hline
\end{tabular}
\end{adjustbox}
\caption{Simulation Results for CPS Data.
The baseline case uses state minimum wage laws to simulate treatment assignment, and generates outcomes using the full
data-generating process described in Section \ref{sec:cps_dgp}, with $T_\post=10$ post-treatment periods and at most $N_\ttt=10$
treatment states. In subsequent settings, we omit parts of the data-generating process (rows 2-6), consider different distributions
for the treatment exposure variable $D_i$  (rows 7-9), different distributions for the outcome variable (rows 10-11), and vary the number of treated cells (rows 12-14).
The full dataset has $N=50$, $T=40$, and outcomes are normalized to have mean zero and unit variance. 
All results are based on 1000 simulation replications.
  \label{table1}}
 \end{center}
  \end{table}						       

\subsubsection{Simulation Results}

Table \ref{table1} compares the performance of the four aforementioned estimators in the simulation design described above. We
 consider various choices for the number of treated units and the treatment assignment distribution. Furthermore,
we also consider settings where we drop various components of the outcome-generating process, such as the fixed effects $\mathbf{F}$ or
the interactive component $\mathbf{M}$, or set the noise correlation matrix $\Sigma$ to be diagonal.
The magnitude of the $\mathbf{F}$, $\mathbf{M}$ and $\be$ components as well as the strength of the
autocorrelation effects in $\Sigma$ captured by the first two autoregressive coefficients are shown in the first four columns of Table \ref{table1}.

At a high level, we find that SDID has excellent performance relative to the benchmarks ---both in terms of bias
and root-mean squared error. This holds in the baseline simulation design and  over a number of other designs where we vary the treatment assignment (from being based on minimum wage laws to gun laws, abortion laws, or completely random), the outcome (from average of log wages to average hours and unemployment rate), and the maximal number of treated units (from 10 to 1) and the number of exposed periods (from 10 to 1). We find that when the treatment assignment is uniformly random, all methods are essentially
unbiased, but SDID is more precise. Meanwhile, when the treatment assignment is not uniformly random, SDID
is particularly successful at mitigating bias while keeping variance in check. 

In the second panel of
Table \ref{table1} 
we provide some additional insights into the superior performance of the SDID estimator by 
sequentially dropping some of the components of the model that generates the potential outcomes.
If we drop the interactive component $\mathbf{M}$ from the outcome model (``No $\mathbf{M}$''), so that the fixed effect specification is correct, the DID estimator performs best (alongside MC).
In contrast, if we drop the fixed effects component $\bf$ (``No $\mathbf{F}$'') but keep the interactive component, the SC estimator does best.
If we drop both parts of the systematic component, and there is only noise, 
the superiority of the SDID estimator vanishes and all estimators are essentially equivalent.
On the other hand, if we remove the noise component so that there is only signal, the increased flexibility of the SDID estimator allows it (alongside MC) to outperform the SC and DID estimators dramatically.


\begin{figure}[t]
\begin{center}
\includegraphics[width=.7\textwidth]{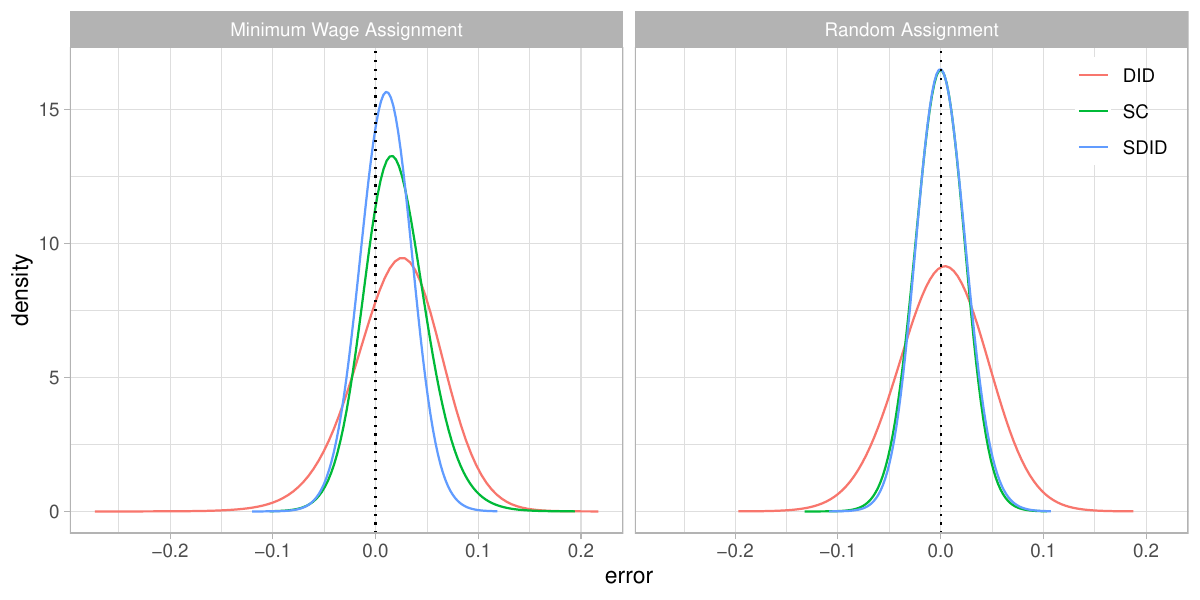} 
\caption{Distribution of the errors of SDID, SC and DID in the setting of the ``baseline'' (i.e., with minimum wage)
and random assignment rows of Table \ref{table1}.}
\label{fig_mw_rand}
\end{center}
\end{figure}

Next, we focus  on two designs of interest: One with the assignment probability model based on parameters estimated in the minimum
wage law model and one where the treatment exposure $D_i$ is assigned uniformly at random.
Figure \ref{fig_mw_rand} shows the errors of the DID, SC and SDID estimators in both settings, and reinforces our observations above.
When assignment is not uniformly random, the distribution of the DID errors is visibly off-center, showing the bias of the estimator. In contrast,
the errors from SDID are nearly centered. Meanwhile, when treatment assignment is uniformly random, both estimators are centered but the errors
of DID are more spread out. We note that the right panel of Figure \ref{fig_mw_rand} is closely related to the simulation
specification of \citet*{Bertrand2004did}. From this perspective, \citet{Bertrand2004did} correctly argue that the error
distribution of DID is centered, and that the error scale can accurately be recovered using appropriate robust estimators.
Here, however, we go further and show that this noise can be substantially reduced by using an estimator like SDID
that can exploit predictable variation by matching on pre-exposure trends.

Finally, we note that Figure \ref{fig_mw_rand} shows that the error distribution of SDID is nearly unbiased and Gaussian in both designs,
thus suggesting that it should be possible to use $\htau^\sdid$ as the basis for valid inference. We postpone a discussion
of confidence intervals until Section \ref{sec:inference}, where we  consider various strategies for inference based on SDID
and show that they attain good coverage here.

\subsection{Penn World Table Placebo Study}
\label{sec:penn}

The simulation based on the CPS  is a natural benchmark for applications that traditionally rely on DID-type methods to estimate the policy effects. In contrast, SC methods are often used in applications where
units tend to be more heterogeneous and are observed over a longer timespan as in, e.g., \citet*{abadie2014}.
To investigate the behavior of SDID in this type of setting, we propose a second set of simulations based on the Penn World
Table. This dataset contains observations on annual real GDP for $N = 111$ countries for $T = 48$ consecutive years, starting from 1959;
we end the dataset in 2007 because we do not want the treatment period to coincide with the Great Recession.
We construct the outcome and the assignment model following the same
procedure outlined in the previous subsection. We select $\log(\mathrm{real\ GDP})$ as
the primary outcome. As with the CPS dataset, the two-way fixed effects explain most of the variation; however, the interactive
component plays a larger role in determining outcomes for this dataset than for the CPS data. We again derive treatment assignment via an exposure
variable $D_i$, and consider both a uniformly random distribution for $D_i$ as well as two non-uniform ones 
 based on predicting  Penn World Table indicators of democracy and education respectively.

Results of the simulation study are presented in Table \ref{table2}. At a high level, these results mirror the ones above: SDID
again performs well in terms of both bias and root-mean squared error and across all simulation settings dominates the other estimators.
In particular, SDID is nearly unbiased, which is important for constructing confidence intervals with accurate coverage rates. The main difference
between Tables \ref{table1} and \ref{table2} is that DID does substantially worse here relative to SC than before. This appears to be due to
the presence of a stronger interactive component in the Penn World Table dataset, and is in line with the empirical practice of preferring
SC over DID in settings of this type. We again defer a discussion of inference to Section \ref{sec:inference}.

\begin{table}[t]
\begin{center}
\begin{adjustbox}{width=1\textwidth}
\begin{tabular}{|l|rrrr|rrrrr|rrrrr|}
\hline
& \multirow{2}{*}{$\frac{\|\mathbf{F}\|_F}{\sqrt{NT}}$} &  \multirow{2}{*}{$\frac{\|\mathbf{M}\|_F}{\sqrt{NT}}$} & \multirow{2}{*}{$\sqrt{\frac{\trace(\Sigma)}{T}}$} & \multirow{2}{*}{AR(2)} &
\multicolumn{5}{c|}{RMSE}& \multicolumn{5}{c|}{Bias}\\ 
& &&  & &SDID & SC & DID &MC& DIFP & SDID & SC & DID &MC& DIFP\\ \hline
Democracy &   0.972 & 0.229 & 0.070 & (.91, -.22) & \bf 0.031 & 0.038 & 0.197 & 0.058 & 0.039 & -0.005 & -0.004 & 0.175 & 0.043 & -0.007\\ 
  Education &  0.972 & 0.229 & 0.070 & (.91, -.22) & \bf 0.030 & 0.053 & 0.172 & 0.049 & 0.039 & -0.003 & 0.025 & 0.162 & 0.040 & -0.005 \\ 
  Random & 0.972 & 0.229 & 0.070 & (.91, -.22) & \bf 0.037 & 0.046 & 0.129 & 0.063 & 0.045 & -0.002 & -0.011 & -0.006 & -0.004 & -0.004 \\ 
   \hline
  \end{tabular}									       
\end{adjustbox}
\caption{Simulation results based on the the Penn World Table dataset.
We use $\log(GDP)$ as the outcome, with $N_\ttt = 10$ out of $N = 111$ treatment countries,
and $T_\post = 10$ out of $T = 48$ treatment periods.
In the first two rows we consider treatment assignment distributions based on democracy status and education metrics,
while in the last row the treatment is assigned completely at random.
All results are based on 1000 simulations.
 \label{table2}}
 \end{center}									       							       
 \end{table}

\section{Formal Results}
\label{sec:formal}

In this section we discuss the formal results.
For the remainder of the paper, we assume that the data generating process follows a generalization of the latent factor model \eqref{basic_model},
\begin{equation}
\label{eq:model}
\by = \bl + \bw \circ \btau + \be, \qquad \mathrm{where} \qquad (\bw \circ \btau)_{it} = \bw_{it} \btau_{it}. 
\end{equation}
The model allows for heterogeneity in treatment effects $\tau_{it}$, as in 
\citet{de2020two}. As above, we
assume block assignment $W_{it} = 1\p{\cb{i > N_\ccc, \, t > T_\pre}}$, where
the subscript ``$\ccc$'' stands for control group,
``$\ttt$'' stands for treatment group, ``$\pre$'' stands for pre-treatment,
and ``$\post$'' stands for post-treatment. 
It is useful to characterize the systematic component $\bl$ as a factor model $\bl=\ba \bbb^{\top}$
as in \eqref{basic_model}, where we define factors $\ba=\mathbf{UD}^{1/2}$ and $\bbb^{\top}=\mathbf{D}^{1/2}\mathbf{V}^{\top}$ in terms of the singular value decomposition $\bl=\mathbf{U D V}^{\top}$. 
Our target estimand is the average treatment effect for the treated units during the periods they were treated,
which under block assignment is
\begin{equation}
\label{eq:estimand}
\tau = \frac{1}{N_\ttt T_\post} \sum_{i=N_\ccc+1}^N\ \sum_{t=T_\pre+1}^T \btau_{it}.
\end{equation}
For notational convenience, we partition the matrix $\by$  as
\[ \by =\left(
\begin{array}{cc}
\by_{\ccc,\pre} & 
\by_{\ccc,\post}\\
\by_{\ttt,\pre} &  
\by_{\ttt,\post}
\end{array}\right),\]
with $\by_{\ccc,\pre}$ a $N_\ccc\times T_\pre$ matrix, 
$\by_{\ccc,\post}$ a $N_\ccc\times T_\post$ matrix,
$\by_{\ttt,\pre}$ a $N_\ttt\times T_\pre$ matrix, and
$\by_{\ttt,\post}$ a $N_\ttt\times T_\post$ matrix,
and similar for $\bl$, $\bw$, $\tau$, and $\be$.
Throughout our analysis, we will assume that the errors $\be_{i.}$ are homoskedastic across units (but not across time),
{\it i.e.}, that $\Var{\be_{i.}} = \Sigma \in \mmr^{T \times T}$ for all units $i = 1, \, \ldots, \, n$. We partition $\Sigma$ as
\[ \Sigma=\left(
\begin{array}{cc} \Sigma_{\pre,\pre} & \Sigma_{\pre,\post}\\
\Sigma_{\post,\pre} & \Sigma_{\post,\post}\end{array}
\right).\]
Given this setting, we are interested in guarantees on how accurately SDID can recover $\tau$.

A simple, intuitively appealing approach to estimating $\tau$ in \eqref{eq:model} is to directly fit both $\bl$ and $\tau$
via methods for low-rank matrix estimation, and several variants of this approach have been proposed in the literature
\citep*[e.g.,][]{athey2017matrix,bai2009panel,xu2017generalized, agarwal2019robustness}.  However, our main interest is in  $\tau$ and not in $\bl$,
and so one might suspect that approaches that provide consistent estimation of $\bl$ may rely on assumptions
that are stronger than what  is necessary for consistent estimation of $\tau$.

Synthetic control methods address confounding bias without  explicitly  estimatin $\bl$ in  \eqref{eq:model}.
Instead, they take an indirect approach more akin to balancing as
in \citet{zubizarreta2015stable} and \citet*{athey2018approximate}. Recall that the SC weights
$\homega^\scc$ seek to balance out the pre-intervention trends in $\by$. Qualitatively, one might hope that doing so also leads us to balance out the unit-factors $\ba$ from \eqref{basic_model}, rendering
$\sum_{i = N_\ccc + 1}^N \homega_i^\scc \ba_{i.} - \sum_{i = 1}^{N_\ccc} \homega_i^\scc \ba_{i.} \approx 0$.
\citet*{Abadie2010} provide some arguments for why this should be the case, and our formal analysis outlines
a further set of conditions under which this type of phenomenon holds. Then, if $\homega^\scc$ in fact succeeds
in balancing out the factors in $\ba$, the SC estimator can be approximated as $\htau^\scc \approx \tau + \sum_{i = 1}^N (2W_i -1) \homega_i^\scc \bar \varepsilon_{i}$ with
$\bar \varepsilon_i=
T_\post^{-1} \sum_{t=T_{pre}+1}^T \varepsilon_{it}$
; in words, SC weighting has succeeded in removing the bias associated with the systematic component $\bl$
and in delivering a nearly unbiased estimate of $\tau$.

Much like the SC estimator, the SDID estimator seeks to recover $\tau$ in \eqref{eq:model} by reweighting to remove the bias associated with $\bl$. However, the SDID estimator takes a 
two--pronged approach. First, instead of only
making use of unit weights $\homega$ that can be used to balance out $\ba$, the estimator also incorporates time weights $\hlambda$
that seek to balance out $\bbb$. This provides a type of double robustness property, whereby if
one of the balancing approaches is effective, the dependence on $\bl$ is approximately removed.
Second, the use of two-way fixed effects in \eqref{main_sdid} and intercept terms in \eqref{unit_weights} and
\eqref{time_weights} makes the SDID estimator invariant to additive shocks to any row or column, i.e.,
if we modify $\bl_{it} \leftarrow \bl_{it} + \alpha_i + \beta_t$ for any choices $\alpha_i$ and $\beta_t$ the
estimator $\htau^\sdid$ remains unchanged. The estimator shares this invariance property with DID
(but not SC).\footnote{More specifically, as suggested by \eqref{main_sc}, SC is invariant to shifts in $\beta_t$
but not $\alpha_i$. In this context, we also note that the DIFP estimator  proposed by \citet{doudchenko2016balancing}
and \citet{ferman2019synthetic} that center each unit's trajectory before applying the synthetic control method is also invariant to shifts in $\alpha_i$.}

The goal of our formal analysis is to understand how and when the SDID weights
succeed in removing the bias due to $\bl$. As discussed below, this requires assumptions
on the signal to noise ratio. The assumptions require that $\be$ does not incorporate too much serial correlation within units, so that we can
attribute persistent patterns in $\by$ to patterns in $\bl$; furthermore, $\ba$ should be stable over time, particularly through the treatment periods. Of course, these are non-trivial assumptions.
However, as discussed further in Section \ref{sec:relworks}, they are considerably weaker than what is required in results of
\citet{bai2009panel} or \citet{moon2015linear,moon2017dynamic} for methods that require explicitly estimating $\bl$ in \eqref{eq:model}.
Furthermore, these assumption are aligned with standard practice in the literature; for example, we can assess the claim that we  balance
all components of $\ba$ by examining the extent to which the method succeeds in balancing pre-intervention periods. Historical context may be needed to justify the assumption that that there were no
other shocks disproportionately affecting the treatment units at the time of the treatment.

\subsection{Weighted Double-Differencing Estimators}
\label{sec:double_diff}

We introduced the SDID estimator \eqref{main_sdid} as the
solution to a weighted two-way fixed effects regression. For the purpose of our formal results,
however, it is convenient to work with the alternative characterization described above
in Equation~\ref{eq:tau-weighted-differences}.
For any weights $\omega \in \Omega$ and $\lambda \in \Lambda$, we can define a weighted
double-differencing estimator\footnote{This weighted double-differencing structure plays a key
role in understanding the behavior of SDID. As discussed further in Section \ref{sec:relworks},
despite relying on a different motivation, certain specifications of the recently proposed ``augmented synthetic
control'' method of \citet*{ben2018augmented} also result in a weighted double-differencing estimator.}
\begin{equation}
\label{eq:tau-weighted-differences}
\htau(\omega,\lambda) =  \omega_{\ttt}^\top \by_{\ttt,\post} \lambda_{\post} - \omega_{\ccc}^\top \by_{\ccc,\post} \lambda_{\post} - \omega_{\ttt}^\top \by_{\ttt,\pre} \lambda_{\pre} +
 \omega_{\ccc}^\top \by_{\ccc,\pre} \lambda_{\pre}.  
\end{equation}
One can verify that the basic DID estimator is of the form \eqref{eq:tau-weighted-differences}, with
constant weights $\omega_\ttt = 1/N_\ttt$, etc. The proposed SDID estimator \eqref{main_sdid} can also be written as
\eqref{eq:tau-weighted-differences}, but now with weights $\homega^\sdid$ and $\hlambda^\sdid$ solving \eqref{unit_weights}
and \eqref{time_weights} respectively. When there is no risk of ambiguity, we will omit the SDID-superscript from the weights
and simply write $\homega$ and $\hlambda$.

Now, note that for any choice of weights $\omega \in \Omega$ and $\lambda \in \Lambda$, we have $\omega_{\ttt} \in \R^{N_\ttt}$
and $\lambda_{\post} \in \R^{T_\post}$ with all elements equal to $1/N_\ttt$ and $1/T_\post$  respectively, and so
$\omega_{\ttt}^\top \btau_{\ttt,\post} \lambda_{\post} = \tau$. Thus, we can decompose the error of any weighted double-differencing
estimator with weights satisfying these conditions as the sum of a bias and a noise component:
\begin{equation}
\label{eq:sdid_error}
\begin{split}
\htau(\omega,\lambda) - \tau 
&= \underbrace{\omega_{\ttt}^\top \bl_{\ttt,\post} \lambda_{\post} - \omega_\ccc^\top \bl_{\ccc,\post} \lambda_{\post} - \omega_{\ttt}^\top \bl_{\ttt,\pre} \lambda_\pre + \omega_\ccc^\top \bl_{\ccc,\pre} \lambda_\pre}_{\text{bias}\ B(\omega,\lambda)} \\
&\quad\quad\quad\quad + \underbrace{\omega_{\ttt}^\top \be_{\ttt,\post} \lambda_{\post}  - \omega_\ccc^\top \be_{\ccc,\post} \lambda_{\post} - \omega_{\ttt}^\top\be_{\ttt,\pre} \lambda_\pre +
 \omega_\ccc^\top \be_{\ccc,\pre} \lambda_\pre.}_{\text{noise}\ \varepsilon(\omega,\lambda)}
\end{split}
\end{equation}
In order to characterize the distribution of $\htau^\sdid - \tau$, it thus remains to carry out two tasks. First, we need
to understand the scale of the errors $ B(\omega,\lambda)$ and $\varepsilon(\omega,\lambda)$, and second, we need
to understand how data-adaptivity of the weights $\homega$ and $\hlambda$  affects the situation.

\subsection{Oracle and Adaptive Synthetic Control Weights}
\label{sec:oracle-and-adaptive}

To address the adaptivity of the SDID weights $\homega$ and $\hlambda$ chosen via \eqref{unit_weights}
and \eqref{time_weights}, we construct alternative ``oracle'' weights that have similar properties to $\homega$ and $\hlambda$
in terms of eliminating bias due to $\bl$, but are deterministic. We can then further decompose the error of $\htau^\sdid$ into
the error of a weighted double-differencing estimator with the oracle weights and the difference between the oracle and feasible
estimators. Under appropriate conditions, we  find  the latter term negligible relative to the error of the
oracle estimator,  opening the door to a simple asymptotic characterization of the error distribution of $\htau^\sdid$.

We define such oracle weights $\tomega$ and $\tlambda$ by minimizing the expectation of the objective functions
$\ell_{unit}(\cdot)$ and $\ell_{time}(\cdot)$ used in \eqref{unit_weights} and \eqref{time_weights} respectively, and set
\begin{equation}
\label{eq:oracle_gen}
\begin{split}
&\p{\tomega_0, \tomega}  = \argmin_{\omega_0 \in \R, \omega \in \Omega} \EE{\ell_{unit}(\omega_0, \omega)}, \ \ \ \
\p{\tlambda_0, \, \tlambda} = \argmin_{\lambda_0 \in \R, \lambda \in \Lambda}\EE{\ell_{time}(\lambda_0, \lambda)}.
\end{split}
\end{equation}
In the case of our model \eqref{eq:model} these weights admit a simplified characterization
\begin{align}
&\p{\tomega_0, \tomega} 
= \argmin_{\omega_0 \in \R, \omega \in \Omega}  \Norm{\omega_0 + \omega_{\ccc}^{\top} \bl_{\ccc,\pre} - \omega_{\ttt}^{\top} \bl_{\ttt, \pre}}^2_2 +  \p{\trace(\Sigma_{\pre,\pre}) + \zeta^2 T_\pre} \Norm{\omega}_2^2,
\label{eq:row_oracle} \\
&\p{\tlambda_0, \, \tlambda} 
= \argmin_{\lambda_0 \in \R, \lambda \in \Lambda} \Norm{\lambda_0 + \bl_{\ccc,\pre} \lambda_{\pre} - \bl_{\ccc,\post}\lambda_{\post}}_2^2 + \Norm{\tilde\Sigma \lambda}_2^2, \ 
\label{eq:col_oracle} \\
& \quad\quad\quad\quad\quad\quad \text{ where }\ \ \ \tilde\Sigma=\begin{pmatrix}\ \Sigma_{\pre,\pre} & -\Sigma_{\pre,\post} \\
-\Sigma_{\post,\pre} & \ \ \ \Sigma_{\post,\post}\end{pmatrix}. \nonumber
\end{align}
The error of the synthetic difference in differences estimator can now be decomposed as follows,
\begin{equation}
\label{eq:full_decomp}
\htau^\sdid - \tau 
            = \underbrace{\varepsilon(\tomega,\tlambda)}_{\text{oracle noise}}
            + \underbrace{B(\tomega,\tlambda)}_{\text{oracle confounding bias}}
            +\quad  \underbrace{\htau(\homega,\hlambda) - \htau(\tomega,\tlambda),}_{\text{deviation from oracle}}
\end{equation}
and our task is to characterize all three terms.

First, the oracle noise term tends to be small when the weights are not too concentrated,
{\it i.e.,} when $\norm{\tomega}_2$ and $\norm{\tlambda}_2$ are small, and we have a sufficient number of exposed units and time periods.
In the case with $\Sigma = \sigma^2 I_{T \times T}$, {\it i.e.,} without any cross-observation correlations, we note that
$\Var{\varepsilon(\tomega,\tlambda)} = \sigma^2 \p{N_\ttt^{-1} + \norm{\tomega}_2^2}\p{T_\post^{-1} + \norm{\tlambda}_2^2}$.
When we move to our asymptotic analysis below, we  work under assumptions that make this oracle noise
term dominant relative to the other error terms in \eqref{eq:full_decomp}.

Second, the oracle confounding bias will be small either when the pre-exposure oracle row regression 
fits well and generalizes to the exposed rows, i.e., 
$\tomega_0 + \tomega_\ccc^\top \bl_{\ccc,\pre} \approx \tomega_{\ttt}^\top \bl_{\ttt,\pre }$ and  
$\tomega_0 + \tomega_\ccc^\top \bl_{\ccc,\post} \approx \tomega_{\ttt}^{\top} \bl_{\ttt,\post}$,
or when the unexposed oracle column regression fits well and generalizes
to the exposed columns,
$\tlambda_0 + \bl_{\ccc,\pre}\tlambda_\pre \approx \bl_{\ccc,\post}\tlambda_{\post}$ and  
$\tlambda_0 + \bl_{\ttt,\pre}\tlambda_\pre \approx \bl_{\ttt,\post} \tlambda_{\post}$.
Moreover, even if neither model generalizes sufficiently well on its own, it suffices for one model to predict the generalization error of the other:
\begin{align*}
B(\omega,\lambda) &= (\omega_{\ttt}^\top \bl_{\ttt,\post} - \omega_{\ccc}^\top \bl_{\ccc,\post})\lambda_{\post} -
                         (\omega_{\ttt}^\top \bl_{\ttt,\pre}    - \omega_{\ccc}^\top \bl_{\ccc,\pre})\lambda_{\pre} \\
                    &= \omega_{\ttt}^\top (\bl_{\ttt,\post}\lambda_{\post} - \bl_{\ttt,\pre}\lambda_{\pre}) -
                         \omega_{\ccc}^\top(\bl_{\ccc,\post}\lambda_{\post} - \bl_{\ccc,\pre}\lambda_\pre).
\end{align*}
The upshot is even if one of the sets of weights fails to remove the bias from the presence of $\bl$, the combination  of weights $\tomega$ and $\tlambda$ can
compensate for such  failures. This double robustness property is similar to that of the augmented inverse
probability weighting estimator, whereby one can trade off between accurate estimates of the outcome
and treatment assignment models \citep*{ben2018augmented, scharfstein1999adjusting}.

We note that although poor fit in the oracle regressions on the unexposed rows and columns of $\bl$ will often be indicated by
a poor fit in the realized regressions on the unexposed rows and columns of $\by$, 
the assumption that one of these regressions generalizes to exposed rows or columns 
is an identification assumption without clear testable implications. 
It is essentially an assumption of no unexplained confounding:
any exceptional behavior of the exposed observations, whether due to exposure or not, can be ascribed to it.

Third, our core theoretical claim, formalized in our asymptotic analysis, is that the SDID estimator will be close to the oracle
when the oracle unit and time weights look promising on their respective training sets, i.e, 
when $\tomega_0 + \tomega_\ccc^{\top} \bl_{\ccc,\pre} \approx \tomega_{\ttt}^{\top} \bl_{\ttt,\pre}$ and $\norm{\tomega}_2$ is not too large and $\tlambda_0 + \bl_{\ccc,\pre}\tlambda_\pre \approx \bl_{\ccc,\post}\tlambda_{\post}$ and $\norm{\tlambda}_2$ is not too large. Although the details differ, as described above
these qualitative properties are also criteria for accuracy of the oracle estimator itself.

Finally, we comment briefly on the behavior of the oracle time weights $\tlambda$ in the presence of
autocorrelation over time. When $\Sigma$ is not diagonal, the effective regularization term in \eqref{eq:col_oracle}
does not shrink $\tlambda_{\pre}$  towards zero, but rather toward an autoregression vector 
\begin{equation}
\label{eq:autoregression-vector}
\psi =  \argmin_{v \in \R^{T_\pre}} \Norm{\tilde\Sigma \begin{pmatrix}v \\ \lambda_{\post} \end{pmatrix}} 
     = \Sigma_{\pre,\pre}^{-1}\Sigma_{\pre,\post}\lambda_{\post}.
\end{equation}
Here $\lambda_{\post}$ is the $T_\post$-component column vector with all elements equal to $1/T_\post$
and $\psi$ is the population regression coefficient in a regression of the average of the post-treatment errors on the pre-treatment errors.
In the absence of autocorrelation, $\psi$ is zero, but when autocorrelation is present, shinkage toward $\psi$
reduces the variance of the SDID estimator---and enables us to gain  precision over the basic DID estimator \eqref{main_did} even when the two-way fixed effects model is correctly specified. This explains some of the behavior noted in the simulations.

\subsection{Asymptotic Properties}
\label{section:asym}

To carry out the analysis plan sketched above, we need to embed our problem into an asymptotic setting.
First, we require the error matrix $\be$ to satisfy some regularity properties. 

\begin{assumption}\label{ass:noise}{\sc (Properties of Errors)}
The rows  $\be_{i}$ of the noise matrix are independent and identically distributed Gaussian
vectors and the eigenvalues of its covariance matrix $\Sigma$ are bounded and bounded away from zero.
\end{assumption}

Next, we spell out assumptions about the sample size. At a high level, we want the panel to be large
({\it i.e.}, $N, \, T \rightarrow \infty$), and for the number of treated cells of the panel to grow to infinity but
slower than the total panel size. We note in particular that we can accommodate sequences where
one of $T_\post$ or $N_\ttt$ is fixed, but not both.

\begin{assumption}\label{ass:sample_sizes} {\sc (Sample Sizes)}
We consider a sequence of populations where \\
$(i)$  the product  $N_\ttt \, T_\post$ goes to infinity, and both $N_\ccc$ and $T_\pre$ go to infinity,\\
$(ii)$ the ratio $T_\pre/N_\ccc$ is bounded and bounded away from zero,\\
$(iii)$  $N_\ccc / (N_\ttt T_\post \max(N_\ttt, T_\post)\log^2(N_\ccc)) \to \infty$.
\end{assumption}

We also need to make assumptions about the spectrum of $\bl$; in particular, $\bl$ cannot have too many large singular values,
although we allow for the possibility of  many small singular values. A sufficient, but not necessary, condition for the assumption
below is that the rank of $\bl$ is less than $\sqrt{\min(T_\pre,N_\ccc)}$. Notice that we do not assume any
lower bounds for non-zero singular values of $\bl$; in fact can accommodate arbitrarily many non-zero but
very small singular values, much like, {\it e.g.}, \citet*{belloni2014inference} can accommodate arbitrarily many
non-zero but very small signal coefficients in a high-dimensional inference problem.
We need that the $\sqrt{\min(T_\pre,N_\ccc)}$th singular value of $\bl_{\ccc,\pre}$ is sufficiently small. Formally:
\begin{assumption}\label{rank}{\sc (Properties of $\bl$)}
 Letting $\sigma_1(\ba), \sigma_2(\ba), \ldots$ denote the singular values of the matrix $\ba$ in decreasing order
and $R$ the largest integer less than $\sqrt{\min(T_\pre,N_\ccc)}$,
\begin{equation}
 \sigma_{R}(\bl_{\ccc,\pre})/R = o\left(\min\Big\{N_\ttt^{-1/2}\log^{-1/2}(N_\ccc), T_\post^{-1/2}\log^{-1/2}(T_\pre)\Bigr\}\right)
 \end{equation}
\end{assumption}

The last---and potentially most interesting---of our assumptions concerns the relation between the factor
structure $\bl$ and the assignment mechanism $\bw$. At a high level, it plays the role of an identifying assumption,
and guarantees that the oracle weights from \eqref{eq:row_oracle} and \eqref{eq:col_oracle} that are directly
defined in terms of $\bl$ are able to adequately cancel out $\bl$ via the weighted double-differencing strategy.
This requires that the optimization problems \eqref{eq:row_oracle} and \eqref{eq:col_oracle} accommodate
reasonably dispersed weights, and that the treated units and after periods not be too dissimilar from
the control units and the before periods respectively.

\begin{assumption}\label{weightss}{\sc (Properties of Weights and $\bl$)}
The oracle unit weights 
 $\oomega$ satisfy
 \begin{equation}
 \label{eq:weightss1}
 \begin{aligned}
 &\|\oomega_{\ccc}\|_2 = o([(N_\ttt T_\post)\log(N_\ccc)]^{-1/2}) \qquad \text{ and } \\
 &\norm{\oomega_0 + \oomega_{\ccc}^{\top}\bl_{\ccc,\pre} - \oomega_{\ttt}^{\top}\bl_{\ttt,\pre}}_2  = o( N_\ccc^{1/4} (N_\ttt T_\post\max(N_\ccc,T_\post))^{-1/4} \log^{-1/2}(N_\ccc)), 
 \end{aligned}
 \end{equation}
the oracle time weights $\olambda$  satisfy  
\begin{equation}
\label{eq:weightss2}
 \begin{aligned}
 &\|\olambda_{\pre}-\psi\|_2 = o([(N_\ttt T_\post)\log(N_\ccc)]^{-1/2}) \qquad \text{ and } \\
 &\norm{\olambda_0 + \bl_{\ccc,\pre}\olambda_\pre - \bl_{\ccc,\post}\olambda_\post}_2 = o(N_\ccc^{1/4} (N_\ttt T_\post)^{-1/8}), 
\end{aligned}
\end{equation}
and the oracle weights jointly satisfy
\begin{equation}
\label{eq:weightss3}
\oomega_\ttt^\top\bl_{\ttt,\post}\olambda_\post
-\oomega_\ccc^\top\bl_{\ccc,\post}\olambda_\post
-
\oomega_\ttt^\top\bl_{\ttt,\pre}\olambda_\pre
+\oomega_\ccc^\top\bl_{\ccc,\pre}\olambda_\pre
 =o\left((N_\ttt T_\post)^{-1/2}\right).
\end{equation}
\end{assumption}

Assumptions  \ref{ass:noise}-\ref{weightss} are substantially weaker than those used to establish asymptotic
normality of comparable methods.\footnote{In particular, note that our assumptions are
satisfied in the well-specified two-way fixed effect setting model. Suppose we have \smash{$L_{it}=\alpha_i+\beta_t$}
with uncorrelated and homoskedastic errors, and that the sample size restrictions in Assumption \ref{ass:sample_sizes} are satisfied. 
Then Assumption \ref{ass:noise} is automatically satisfied, and
the rank condition on  $\bl$ from Assumption \ref{rank} is satisfied with $R=2$.
Next, we see that the oracle unit weights satisfy \smash{$\tilde\omega_{\ccc,i}=1/N_\ccc$} so that
\smash{$\|\oomega\|_2=1/\sqrt N_\ccc$}, and the oracle time weights satisfy \smash{$\tilde\lambda_{\pre,i}=1/T_\pre$}
so that \smash{$\|\olambda-\psi\|_2=1/\sqrt N_\ccc$}. Thus if the restrictions on the rates at which the sample sizes increase
in Assumption \ref{ass:sample_sizes} are satisfied, then \eqref{eq:weightss1} and \eqref{eq:weightss2} are satisfied.
Finally, the additive structure of $\bl$ implies that, as long as the weights for the controls sum to one, 
\smash{$\oomega_\ttt^\top\bl_{\ttt,\post}\olambda_\post
-\oomega_\ccc^\top\bl_{\ccc,\post}\olambda_\post=0$},
and
\smash{$\oomega_\ttt^\top\bl_{\ttt,\pre}\olambda_\pre
+\oomega_\ccc^\top\bl_{\ccc,\pre}\olambda_\pre
 =0$},
 so that \eqref{eq:weightss3} is satisfied.}
We do not require that double differencing alone removes the individual and time effects
as the DID assumptions do. Furthermore, we do not require that unit comparisons alone are sufficient to remove the biases
in comparisons between treated and control units as the SC assumptions do. Finally, we do not require a low rank factor model
to be correctly specified,  as is often assumed in the analysis of methods that estimate $\bl$ explicitly
\citep[e.g.,][]{bai2009panel,moon2015linear,moon2017dynamic}. Rather, we only need the combination of the
three bias-reducing components in the SDID estimator, $(i)$ double differencing, $(ii)$ the unit weights, and $(iii)$ the time
weights, to reduce the bias to a sufficiently small level.

Our main formal result states that under these assumptions, our estimator is asymptotically normal.
Furthermore, its asymptotic variance is optimal, coinciding with the variance we would get 
if we knew $\bl$ and $\Sigma$ a-priori and could therefore estimate $\tau$ by a simple average
of $\tau_{it}$ plus unpredictable noise,  
$N_{\ttt}^{-1}\sum_{i=N_{\ccc}+1}^N [T_{\post}^{-1}\sum_{t=T_\pre+1}^T (\btau_{it} + \varepsilon_{it}) - \be_{i,\pre}\psi]$.

\begin{theorem}
\label{theo:asymptotic-linearity}
Under the model \eqref{eq:model} with $\bl$ and $\bw$ taken as fixed, suppose that we run the
SDID estimator \eqref{main_sdid} with regularization parameter $\zeta$
satisfying 
$(N_{\ttt}T_{\post})^{1/2} 
\log(N_\ccc) = o(\zeta^2)$.
Suppose moreover that Assumptions \ref{ass:noise}-\ref{weightss} hold. Then,
\begin{equation}
\label{eq:asymptotic-linearity}
\htau^\sdid-\tau=\frac{1}{N_\ttt}\sum_{i=N_{\ccc}+1}^N \p{\frac{1}{T_\post}\sum_{t=T_\pre+1}^T \varepsilon_{it} - \be_{i,\pre}\psi} 
    + o_p\left((N_\ttt T_\post)^{-1/2}\right),
\end{equation}
and consequently
\begin{equation}
\label{eq:CLT1}
\p{ \htau^\sdid - \tau} \,\big/\, {V_{\tau}^{1/2}}  \, \Rightarrow \, {\cal N}\left(0, \, 1\right), \ \ {\rm where}\ \
V_{\tau} = \frac{1}{N_\ttt} \Var{\frac{1}{T_\post}\sum_{t=T_\pre+1}^T \varepsilon_{it} -  \be_{i,\pre} \psi}.
\end{equation}
Here $V_\tau$ is on the order of $1/(N_{\ttt}T_{\post})$, i.e., $N_\ttt T_\post V_{\tau}$ is bounded and bounded away from zero.
\end{theorem}

\section{Large-Sample Inference}
\label{sec:inference}

The asymptotic result from  the previous section can be used to motivate practical methods for large-sample inference using SDID. Under appropriate conditions, the estimator is asymptotically normal and zero-centered; thus, if these conditions hold and 
we have a consistent estimator for its asymptotic variance $V_\tau$, we can use conventional confidence intervals 
\begin{equation}
\label{eq:GCI}
 \tau \in \htau^\sdid \pm z_{\alpha/2}\sqrt{\hV_\tau}  
\end{equation}
to conduct asymptotically valid inference. In this section, we discuss three approaches to variance estimation for use in confidence intervals
of this type.

\begin{algorithm}[t]
 \KwData{$\mathbf{Y},\mathbf{W}, B$}
 \KwResult{ Variance estimator $\hV^{cb}_{\tau}$ }
 \For{$i \leftarrow 1$ \KwTo $B$}{
    Construct a bootstrap dataset $(\mathbf{Y}^{(b)}, \mathbf{W}^{(b)})$ by sampling $N$ rows of \\ $(\mathbf{Y}, \mathbf{W})$ with replacement. \\ 
    \If{the bootstrap sample has no treated units or no control units}{ Discard and resample (\textbf{go to 2})}
    Compute the SDID estimator $\hat \tau^{(b)}$ based on $(\mathbf{Y}^{(b)}, \mathbf{W}^{(b)})$
}
Define $\hV^{b}_\tau=\frac{1}{B}\sum_{b=1}^{B} (\hat{\tau}^{(b)} - \frac{1}{B}\sum_{b=1}^B\hat{\tau}^{(b)})^2$\;
 \caption{Bootstrap Variance Estimation}
 \label{alg:boot}
\end{algorithm}

The first proposal we consider, described in detail in Algorithm \ref{alg:boot}, involves a clustered bootstrap
\citep{efron1979bootstrap} where we independently resample units. As argued in \citet*{Bertrand2004did},
unit-level bootstrapping presents a natural approach to inference with panel data when repeated observations
of the same unit may be correlated with each other. The bootstrap is simple to implement and, in our experiments,
appears to yield robust performance in large panels.
The main downside of the bootstrap is that it may be computationally costly as it involves running
the full SDID algorithm for each bootstrap replication, and for large datasets this can be prohibitively expensive.

\begin{algorithm}[t]
 \KwData{$\hat\omega, \hat \lambda, \mathbf{Y}, \mathbf{W}, \hat \tau$}
 \KwResult{ Variance estimator $\hV_{\tau}$ }
 \For{$i \leftarrow 1$ \KwTo $N$}{
  Compute 
 $\hat \tau^{(-i)}:\argmin_{\tau, \{\alpha_j,\beta_t\}_{j \ne i, t}}\sum_{j\ne i,t}\left(\mathbf{Y}_{jt} - \alpha_j - \beta_t - \tau \mathbf{W}_{it}\right)^2\hat\omega_j\hat \lambda_t$
 }
Compute $\hV^{\mathrm{jack}}_\tau=(N-1)N^{-1}\sum_{i=1}^{N} (\hat{\tau}^{(-i)} - \hat{\tau})^2$\;
 \caption{Jackknife Variance Estimation}
 \label{alg:jack}
\end{algorithm}

To address this issue we next consider an approach to inference that is more closely tailored to the SDID method and only involves
running the full SDID algorithm once, thus dramatically decreasing the computational burden.
Given weights $\homega$ and $\hlambda$ used to get the SDID point estimate, Algorithm \ref{alg:jack} applies the jackknife
\citep{miller1974jackknife} to the weighted SDID regression \eqref{main_sdid}, with the weights treated as fixed.
The validity of this procedure is not implied directly by asymptotic linearity as in \eqref{eq:asymptotic-linearity}; however, as shown below,
we still recover conservative confidence intervals under considerable generality.

\begin{theorem}
\label{theo:jack}
Suppose that the elements of $\bl$ are bounded.
Then, under the conditions of Theorem \ref{theo:asymptotic-linearity}, 
the jackknife variance estimator described in Algorithm \ref{alg:jack} yields conservative confidence intervals,
i.e., for any $0 < \alpha < 1$,
\begin{equation}
\label{eq:jack_cons}
\liminf \PP{\tau \in \htau^\sdid \pm z_{\alpha/2}\sqrt{\hV^{\mathrm{jack}}_\tau}} \geq 1 - \alpha.
\end{equation}
Moreover, if the treatment effects $\btau_{it} = \tau$ are constant\footnote{When
treatment effects are heterogeneous, the jackknife implicitly treats the estimand \eqref{eq:estimand} as
random whereas we treat it as fixed, thus resulting in excess estimated variance; see \citet{imbens2004}
for further discussion.}
and
\begin{equation}
\label{eq:regr_cons}
T_{\post} N_\ttt^{-1} \Norm{\hlambda_0 + \bl_{\ttt,\pre} \hlambda_{\pre} - \bl_{\ttt,\post}\hlambda_{\post}}_2^2 \rightarrow_p 0,
\end{equation}
i.e., the time weights $\hlambda$ are predictive enough on the exposed units,
then the jackknife yields exact confidence intervals and \eqref{eq:jack_cons} holds with equality.
\end{theorem}

In other words, we find that the jackknife is in general conservative and is exact 
when treated and control units are similar enough that time weights that fit the control units generalize
to the treated units. 
This result depends on specific structure of the SDID estimator, and does not hold for
related methods such as the SC estimator. In particular, an analogue to Algorithm \ref{alg:jack} for SC would be
severely biased upwards, and would not be exact even in the well-specified fixed effects model. Thus, we do
not recommend (or report results for) this type of jackknifing with the SC estimator. We do report results for
jackknifing DID since, in this case, there are no random weights $\homega$ or $\hlambda$ and so our jackknife
just amounts to the regular jackknife.

Now, both the bootstrap and jackknife-based methods discussed so far are designed with the setting of Theorem \ref{theo:asymptotic-linearity} in
mind, i.e., for large panels with many treated units. These methods may be less reliable when the number of treated units $N_\ttt$
is small, and the jackknife is not even defined when $N_\ttt = 1$. However, many applications of synthetic controls have $N_\ttt = 1$, e.g.,
the California smoking application from Section \ref{sec:calif}. To this end, we consider a third variance estimator that is motivated by placebo evaluations as often considered in the literature
on synthetic controls \citep*{Abadie2010,abadie2014}, and that can be applied with $N_\ttt =1$. The main idea of such placebo evaluations
is to consider the behavior of synthetic control estimation when we replace the unit that was exposed to the treatment with
different units that were not exposed.\footnote{Such a placebo test is closely connected to permutation tests in randomization inference;
however, in many synthetic controls applications, the exposed unit was not chosen at random, in which case placebo tests do not have
the formal properties of randomization tests \citep{firpo2018synthetic,hahn2016}, and so may need to be interpreted via a
more qualitative lens.}
Algorithm \ref{alg:placebo} builds on this idea, and uses placebo predictions using only the unexposed units to estimate
the noise level, and then uses it to get \smash{$\hV_\tau$} and build confidence intervals as in \eqref{eq:GCI}. See \citet{bottmer2021design} for a discussion of the properties of such placebo variance estimators in small samples.

\begin{algorithm}[t]
 \KwData{$\mathbf{Y}_{\ccc,\cdot}, N_{\ttt}, B$}
 \KwResult{ Variance estimator $\hV^{\mathrm{placebo}}_{\tau}$ }
 \For{$b \leftarrow 1$ \KwTo $B$}{
    Sample $N_{\ttt}$ out of the $N_{\ccc}$ control units without replacement to `receive the placebo'\;
    Construct a placebo treatment matrix $\mathbf{W}_{\ccc,\cdot}^{(b)}$ for the controls\;
    Compute the SDID estimator $\hat \tau^{(b)}$\
    based on $(\mathbf{Y}_{\ccc,\cdot}, \mathbf{W}_{\ccc,\cdot}^{(b)})$
    ;
}

Define $\hV^{\mathrm{placebo}}_\tau=\frac{1}{B}\sum_{b=1}^{B} (\hat{\tau}^{(b)} - \frac{1}{B}\sum_{b=1}^B\hat{\tau}^{(b)})^2$\;
 \caption{Placebo Variance Estimation}
 \label{alg:placebo}
\end{algorithm}

Validity of the placebo approach relies fundamentally on homoskedasticity across units, because if the exposed and unexposed units
have different noise distributions then there is no way we can learn \smash{$V_\tau$} from unexposed units alone. We also
note that non-parametric variance estimation for treatment effect estimators is in general impossible if we only have one
treated unit, and so homoskedasticity across units is effectively a necessary assumption in order for inference to be possible
here.\footnote{In Theorem \ref{theo:asymptotic-linearity}, we also assumed homoskedasticity. In contrast to the case of placebo
inference, however, it's likely that a similar result would also hold without homoskedasticity; homoskedasticity
is used in the proof essentially only to simplify notation and allow the use of concentration inequalities which have been 
proven in the homoskedastic case but can be generalized.}
Algorithm \ref{alg:placebo} can also be seen as an adaptation of the method of \citet{conley2011inference} for inference in DID models with
few treated units and assuming homoskedasticity, in that both rely on the empirical distribution of residuals for placebo-estimators
run on control units to conduct inference. We refer to \citet{conley2011inference} for a detailed analysis of this class of algorithms.

\begin{table}[t]
\begin{center}
\begin{tabular}{|l|rrr|rrr|rrr|}
\hline
&\multicolumn{3}{c|}{Bootstrap} & \multicolumn{3}{c|}{Jackknife}& \multicolumn{3}{c|}{Placebo}\\ 
& SDID & SC & DID & SDID & SC & DID& SDID & SC & DID \\ \hline
 Baseline & 0.96 & 0.93 & 0.89 & 0.93 & --- & 0.92 & 0.95 & 0.89 & 0.96 \\ 
   \hline
Gun Law & 0.97 & 0.96 & 0.93 & 0.94 & --- & 0.93 & 0.94 & 0.95 & 0.93 \\ 
  Abortion & 0.96 & 0.94 & 0.93 & 0.93 & --- & 0.95 & 0.97 & 0.91 & 0.96 \\ 
  Random & 0.96 & 0.96 & 0.92 & 0.93 & --- & 0.94 & 0.96 & 0.96 & 0.94 \\ 
  Hours & 0.92 & 0.96 & 0.94 & 0.89 & --- & 0.95 & 0.91 & 0.89 & 0.96 \\ 
  Urate & 0.91 & 0.90 & 0.57 & 0.86 & --- & 0.64 & 0.88 & 0.89 & 0.62 \\ 
   \hline
$T_{post}= 1$ &0.93 & 0.94 & 0.84 & 0.92 & --- & 0.88 & 0.93 & 0.91 & 0.92\\ 
$N_{tr} = 1$ & --- & --- & --- & --- & --- & --- & 0.97 & 0.95 & 0.96  \\ 
$T_{post}= N_{tr} = 1  $ & --- & --- & --- & --- & --- & --- & 0.95 & 0.94 & 0.94\\ 
   \hline
Resample, $N=200$ & 0.94 & 0.95 & 0.92 & 0.95 & --- & 0.93 & 0.96 & 0.94 & 0.94 \\ 
  Resample, $N=400$ & 0.95 & 0.92 & 0.96 & 0.96 & --- & 0.95 & 0.96 & 0.91 & 0.96\\ 
   \hline
Democracy & 0.93 & 0.96 & 0.55 & 0.94 & --- & 0.59 & 0.98 & 0.97 & 0.79 \\ 
Education &0.95 & 0.95 & 0.30 & 0.95 & --- & 0.34 & 0.99 & 0.90 & 0.94 \\ 
Random & 0.93 & 0.95 & 0.89 & 0.96 & --- & 0.91 & 0.95 & 0.94 & 0.91 \\ 
\hline
  \end{tabular}
\caption{Coverage results for nominal 95\% confidence intervals in the CPS and Penn World Table simulation setting from Tables \ref{table1} and  \ref{table2}.
The first three columns show coverage of confidence intervals obtained via  the Placebo method.
The second set of columns show coverage from the jackknife method.
The last set of columns show coverage from the clustered bootstrap.
Unless otherwise specified, all settings have $N = 50$ and $T = 40$
cells, of which at most $N_\ttt = 10$ units and $T_\post = 10$ periods are treated. In rows 7-9, we reduce the number
of treated cells. In rows 10 and 11, we artificially make the panel larger by adding rows, which makes the assumption that the number of treated units is small relative to the number of control units more accurate (we set $N_\ttt$ to $10\%$ of the total number of units).
We do not report jackknife and bootstrap coverage rates for $N_\ttt = 1$ because the estimators are not well-defined. We do not report jackknife coverage rates for SC because, as discussed in the text, the variance estimator is not well justified in this case.
All results are based on 400 simulation replications.}
\label{table_cps_cov}
 \end{center}								       
\end{table}

Table \ref{table_cps_cov} shows the coverage rates for the experiments described in Section \ref{sec:cps} and \ref{sec:penn},
using Gaussian confidence intervals \eqref{eq:GCI} with variance estimates obtained as described above.
In the case of the SDID estimation, the bootstrap estimator performs particularly well, yielding nearly nominal $95\%$ coverage,
while both placebo and jackknife variance estimates also deliver results that are close to the  nominal $95\%$ level.
This is encouraging, and aligned with our previous observation that the SDID estimator
appeared to have low bias. That being said, when assessing the performance of the placebo estimator, recall
that the data in Section \ref{sec:cps} was generated with noise that is both Gaussian and homoskedastic across
units---which were assumptions that are both heavily used by the placebo estimator.

In contrast, we see that coverage rates for DID and SC can be relatively low, especially in cases with significant bias
such as the setting with the  state unemployment rate as the outcome. This is again in line with what one may have expected
based on the distribution of the errors of each estimator as discussed in Section \ref{sec:cps}, e.g., in
Figure \ref{fig_mw_rand}: If the point estimates $\htau$ from DID and SC are dominated by bias, then we should not
expect confidence intervals that only focus on variance to achieve coverage.

\section{Related Work}
\label{sec:relworks}

Methodologically, our work draws most directly from the literature on SC methods, including
\citet{abadie2003}, \citet*{Abadie2010, abadie2014}, \citet{abadie2016},  \citet{doudchenko2016balancing}, and \citet*{ben2018augmented}.
Most methods in this line of work can be thought of as  focusing on constructing unit weights that create comparable
(balanced) treated and control units, without relying on any modeling or weighting across time.
\citet*{ben2018augmented} is an interesting exception. Their augmented synthetic control estimator,
motivated by the augmented inverse-propensity weighted estimator of \citet*{robins1994estimation},
combines synthetic control weights with a regression adjustment for improved accuracy (see also \citet*{kellogg2020combining} which explicitly connects SC to matching). 
They focus on the case of $N_\ttt = 1$ exposed units and $T_\post = 1$ post-exposure periods, and their method involves fitting a
model for the conditional expectation $m(\cdot)$ for $Y_{iT}$ in terms of the lagged outcomes $\by_{i,\pre}$, and
then using this fitted model to ``augment'' the basic synthetic control estimator as follows.
\begin{equation}
\label{eq:asc}
\begin{split}
 \htau_\asc
 &= Y_{NT} - \p{\sum_{i=1}^{N-1} \homega^\scc_i Y_{iT}+ \left(
\hat m(\by_{N,\pre})-\sumi\homega^\scc_i \hat m(\by_{i,\pre})\right)}.
\end{split}
\end{equation}
Despite their different motivations, the augmented synthetic control and synthetic difference in differences methods share an
interesting connection: with a linear model $m(\cdot)$, 
$\htau_\sdid$ and $\htau_\asc$ are very similar. In fact, had we fit $\homega^{\sdid}$ without intercept,
they would be equivalent for $\hat m(\cdot)$ fit by least squares on the controls,
imposing the constraint that its coefficients are nonnegative and to sum to one, that is, 
for $\hat m(\by_{i,\pre}) = \hlambda_0^{\sdid} + \by_{i,\pre}\hlambda_{\pre}^{\sdid}$.
This connection suggests that weighted two-way bias-removal methods are a natural way of working with
panels where we want to move beyond simple difference in difference approaches.

We also note recent work of \citet{roth2018pre} and \citet{rambachan2019honest},
who focus on valid inference in difference in differences settings when users look at past outcomes to
check for parallel trends. Our approach uses past data not only to check
whether the trends are parallel, but also to construct the weights to make them parallel. In this setting, we show
that one can still conduct valid inference, as long as $N$ and $T$ are large enough and the size of the treatment block is small.

In terms of our formal results, our paper fits broadly in the literature on panel models with interactive fixed effects and the matrix completion literature \citep{athey2017matrix, bai2009panel, moon2015linear,moon2017dynamic, robins1985comparison,xu2017generalized}.
Different types of problems of this form have a long tradition in the econometrics literature,
with early results going back to \citet*{ahn2001gmm}, \citet{chamberlain1992efficiency} and \citet*{holtz1988estimating} in the case
of finite-horizon panels (i.e., in our notation, under asymptotics where $T$ is fixed and only $N \rightarrow \infty$).
More recently, \citet{freyberger2018non} extended the work of \citet{chamberlain1992efficiency} to a setting that's
closely related to ours, and emphasized the role of the past outcomes for constructing moment restrictions in the fixed-$T$ setting.
\citet{freyberger2018non} attains identification by assuming that the errors $\be_{it}$ are uncorrelated,
and thus past outcomes act as valid instruments. In contrast, we allow for correlated errors within rows, and thus need to work in a large-$T$ setting.

Recently, there has considerable interest in models of type \eqref{basic_model} under asymptotics where both $N$ and $T$
get large. One popular approach, studied by \citet{bai2009panel} and \citet{moon2015linear,moon2017dynamic}, involves
fitting \eqref{basic_model} by ``least squares'', i.e., by minimizing squared-error loss while constraining $\widehat{\bl}$ to
have bounded rank $R$. While these results do allow valid inference for $\tau$, they require strong assumptions. First,
they require the rank of $\bl$ to be known a-priori (or, in the case of \citet{moon2015linear}, require a known upper bound
for its rank), and second, they require a $\beta_{\min}$-type condition whereby the normalized non-zero singular values of $\bl$
are well separated from zero. In contrast, our results require no explicit limit on the rank of $\bl$ and allow for $\bl$ to have
to have positive singular values that are arbitrarily close to zero, thus suggesting that the SDID method may be more robust than the
least squares method in cases where the analyst wishes to be as agnostic as possible regarding properties of $\bl$.\footnote{By
analogy, we also note that, in the literature on high-dimensional inference, methods that do no assume a uniform lower bound on
the strength of non-zero coefficients of the signal vector are generally considered more robust than ones that do
\citep*[e.g.,][]{belloni2014inference,zhang2014confidence}.}

\citet*{athey2017matrix}, \citet*{amjad2018robust}, \citet{moon2018nuclear} and \citet{xu2017generalized} build on this line of work, and replace the fixed-rank constraint with data-driven regularization on $\widehat{\bl}$. This innovation is very helpful from a computational perspective; however, results for inference about
$\tau$ that go beyond what was available for least squares estimators are currently not available. We also note recent papers that
draw from these ideas in connection to synthetic control type analyses, including \citet{chan2020pcdid} and  \citet{gobillon2016regional}.
Finally, in a paper contemporaneous to ours, \citet*{agarwal2019robustness} provide improved bounds from principal component
regression in an errors-in-variables model closely related to our setting, and discuss implications for estimation in synthetic
control type problems. Relative to our results, however, \citet{agarwal2019robustness} still require assumptions on the behavior
of the small singular values of $\bl$, and do not provide methods for inference about $\tau$.

In another direction, several authors have recently proposed various methods that implicitly control for the systematic component
$\bl$ in models of time \eqref{basic_model}. In one early example, \citet*{hsiao2012panel} start with a factor model similar to ours and
show that under certain assumptions it implies the moment condition
\begin{equation}
\label{eq:linear_model_stat}
Y_{Nt}= a+ \sum_{j=1}^{N-1}\beta_j Y_{jt} + \epsilon_{Nt}, \ \ \ \ \ \EE{\varepsilon_{Nt} \,\big|\, \{Y_{jt}\}_{j=1}^{N-1}}=0,
\end{equation}
for all $t = 1, \, \ldots, \, T$. The authors then estimate $\beta_j$ by (weighted) OLS.  This approach is further refined by \cite{li2017estimation},
who additionally propose to penalizing the coefficients $\beta_j$ using the lasso \citep{tibshirani1996regression}.
In a recent paper,  \citet*{chernozhukov2018inference} use the model \eqref{eq:linear_model_stat} as a starting point for inference.

While this line of work shares a conceptual connection with us, the formal setting is very different. In order to derive a representation of
the type \eqref{eq:linear_model_stat}, one essentially needs to assume a random specification for \eqref{basic_model} where both $\bl$ and
$\be$ are stationary in time. \citet{li2017estimation} explicitly assumes that the outcomes $\by$ themselves are weakly stationary, while 
\citet*{chernozhukov2018inference} makes the same assumption to derive the results that are valid under general misspecification.
In our results, we do not assume stationarity anywhere: $\bl$ is taken as deterministic and the errors $\be$ may be non-stationary.
Moreover, in the case of most synthetic control and difference in differences analyses, we believe stationarity to be a fairly restrictive
assumption. In particular, in our model, stationarity would imply that a simple pre-post comparison for exposed units would be an
unbiased estimator of $\tau$ and, as a result, the only purpose of the unexposed units would be to help improve efficiency.
In contrast, in our analysis, using unexposed units for double-differencing is crucial for identification.

\citet{ferman2019synthetic} analyze the performance of synthetic control estimator using essentially the same model as we do. They focus on the situations where $N$ is small, while $T_\pre$ (the number of control periods) is growing. They show that unless time factors have strong trends ({\it e.g.}, polynomial) the synthetic control estimator is asymptotically biased. Importantly \cite{ferman2019synthetic} focus on the standard synthetic control estimator, without time weights and regularization, but with an intercept in the construction of the weights.

Finally, from a statistical perspective, our approach bears some similarity to the work on ``balancing'' methods for
program evaluation under unconfoundedness, including \citet*{athey2018approximate}, \citet*{graham1},
\citet{hirshberg2017augmented}, \citet{imai2014covariate}, \citet{kallus2020generalized}, \citet{zhao2019covariate}
and \citet{zubizarreta2015stable}. One major result of this line of work is that, by algorithmically finding weights that balance
observed covariates across treated and control observations, we can derive robust estimators with good asymptotic properties
(such as efficiency). In contrast to this line of work, rather than balancing observed covariates, we here need to balance unobserved
factors $\ba$ and $\bbb$ in \eqref{basic_model} to achieve consistency; and accounting for this forces us to follow a different formal
approach than existing studies using balancing methods.


\bibliographystyle{plainnat}
\bibliography{references}

\newpage
\section{Appendix}
\subsection{Placebo Study Details}
\label{sec:placebo-study-details}
For the placebo studies we use three indicators $D_i$ to estimate the assignment model via
logistic regression as described in \eqref{eq:assignment_lr}.
The first is equal to an indicator that state $i$ has
 a minimum wage that is higher than the federal minimum
wage in the year 2000. This indicator was taken from \url{http://www.dol.gov/whd/state/stateMinWageHis.htm};
see \citet{barrios2012clustering} for details.
The second indicator comes from a state having an open-carry gun law. This was taken from
 \url{https://lawcenter.giffords.org/gun-laws/policy-areas/guns-in-public/open-carry/}.
  The third indicator comes from the state not having  a ban on partial birth abortions. This was taken from 
 \url{https://www.guttmacher.org/state-policy/explore/overview-abortion-laws}.
  Table \ref{laws} presents the values for these indicators.

\begin{table}[h]
\centering
\begin{adjustbox}{width=0.55\textwidth}
\begin{tabular}{lccc}
  \hline
State & Minimum Wage  & Unrestricted Open Carry & Abortion  \\ 
  \hline
  Alaska & 0 & 1 & 0\\
Alabama &  0 & 0 & 0 \\ 
  Arkansas & 0 & 1 & 0  \\ 
  Arizona & 0 & 1 & 0\\
  California & 1 & 0 & 1\\
  Colorado & 0 & 0 & 1\\ 
  Connecticut &  0 & 0 & 1\\ 
  Delaware & 1 & 1 & 1\\ 
  Florida & 0 & 0 & 0\\
  Georgia & 0 & 0 & 0\\ 
  Hawaii & 0 & 0 & 1\\
  Idaho &  0 & 1 & 0\\ 
  Illinois & 0 & 0 & 1\\ 
  Indiana & 0 & 0 & 0\\ 
  Iowa &  0 & 0 & 0\\ 
  Kansas &  0 & 1 & 0\\ 
  Kentucky &  0 & 1 & 0\\ 
  Louisiana & 0 & 1 & 0\\ 
  Massachusetts & 1 & 0 & 1\\
  Maine &  0 & 1 & 1\\ 
  Maryland & 0 & 0 & 1\\
  Michigan & 0 & 1 & 0\\
  Minnesota & 0 & 0& 1\\ 
  Mississippi & 0 & 1 & 0\\ 
  Missouri &  0 & 0 & 0\\ 
  Montana &  0 & 1 & 0\\ 
  Nebraska & 0 & 1 & 0\\ 
  Nevada &  0 & 1 & 1\\ 
  New Hampshire &  0 & 1 & 0\\ 
  New Mexico &  0 & 1 & 0\\ 
  North Carolina & 0 & 1 & 1\\ 
  North Dakota & 0 & 0 & 0\\ 
  New York & 0 & 0 & 1\\
  New Jersey & 0 & 0 & 0\\
  Ohio &  0 & 1 & 0\\ 
  Oklahoma &  0 & 0 & 0\\ 
  Oregon & 1 & 1 & 1\\
  Pennsylvania &  0 & 0 & 1\\ 
  Rhode Island &  1 & 0 & 0\\ 
  South Carolina &  0 & 0 & 0\\ 
  South Dakota &  0 & 1 & 0\\ 
  Tennessee &  0 & 0 & 0\\ 
  Texas & 0 & 0 & 0\\ 
  Utah &  0 & 0 & 0\\ 
  Vermont &  1 & 1 & 1\\ 
  Virginia &  0 & 0 & 0\\ 
  Washington & 1 & 0 & 1\\
  West Virginia &  0 & 1 & 0\\ 
  Wisconsin &0 & 1 & 0 \\ 
  Wyoming &  0 & 1 & 1\\ 
   \hline
\end{tabular}
\end{adjustbox}
\caption{State Regulations}
\label{laws}
\end{table}

\begin{table}[ht]
\centering
\begin{tabular}{|l|rr|rr|}
  \hline
 & SC  & SC (reg) & DIFP & DIFP (reg) \\ 
  \hline
Baseline & 0.037 & 0.078 & 0.032 & 0.036 \\ 
\hline
No Correlation & 0.038 & 0.079 & 0.032 & 0.036 \\ 
No $\mathbf{M}$ & 0.018 & 0.034 & 0.016 & 0.014 \\ 
No $\mathbf{F}$ & 0.023 & 0.025 & 0.032 & 0.036 \\ 
Only noise & 0.014 & 0.012 & 0.016 & 0.014 \\ 
No noise & 0.017 & 0.034 & 0.011 & 0.020 \\ 
\hline
Gun Law & 0.027 & 0.034 & 0.030 & 0.040 \\ 
Abortion & 0.031 & 0.065 & 0.027 & 0.035 \\ 
Random & 0.025 & 0.031 & 0.027 & 0.035 \\ 
\hline
Hours & 0.203 & 0.329 & 0.197 & 0.191 \\ 
U-rate & 0.184 & 0.259 & 0.187 & 0.288 \\ 
\hline
$T_{post} = 1$ & 0.059 & 0.065 & 0.054 & 0.050 \\ 
$N_{tr} = 1$ & 0.072 & 0.085 & 0.083 & 0.087 \\ 
$T_{post} = N_{tr} = 1$ & 0.124 & 0.124 & 0.117 & 0.112 \\ 
\hline
Resample, $N = 200$ & 0.017 & 0.016 & 0.018 & 0.018 \\ 
Resample, $N = 400$ & 0.014 & 0.011 & 0.015 & 0.012 \\ 
\hline
Democracy & 0.038 & 0.035 & 0.039 & 0.031 \\ 
Education & 0.053 & 0.062 & 0.039 & 0.029 \\ 
Random & 0.046 & 0.047 & 0.045 & 0.046 \\ 
   \hline
\end{tabular}
\caption{Comparison of SC and DIFP estimators without regularization and with the regularization parameter used to compute SDID unit weights. Simulation designs correspond to those of Table \ref{table1} and \ref{table2}. All results are based on 1000 simulations.}
 \label{table_reg_vs_unreg}
\end{table}

\FloatBarrier

\subsection{Unit/time weights for California}
\begin{table}[ht]
\centering
\begin{tabular}{r|rrr}
  \hline
 & DID & SC & SDID \\ 
  \hline
1988 & 0.053 & 0.000 & 0.427 \\ 
  1987 & 0.053 & 0.000 & 0.206 \\ 
  1986 & 0.053 & 0.000 & 0.366 \\ 
  1985 & 0.053 & 0.000 & 0.000 \\ 
  1984 & 0.053 & 0.000 & 0.000 \\ 
  1983 & 0.053 & 0.000 & 0.000 \\ 
  1982 & 0.053 & 0.000 & 0.000 \\ 
  1981 & 0.053 & 0.000 & 0.000 \\ 
  1980 & 0.053 & 0.000 & 0.000 \\ 
  1979 & 0.053 & 0.000 & 0.000 \\ 
  1978 & 0.053 & 0.000 & 0.000 \\ 
  1977 & 0.053 & 0.000 & 0.000 \\ 
  1976 & 0.053 & 0.000 & 0.000 \\ 
  1975 & 0.053 & 0.000 & 0.000 \\ 
  1974 & 0.053 & 0.000 & 0.000 \\ 
  1973 & 0.053 & 0.000 & 0.000 \\ 
  1972 & 0.053 & 0.000 & 0.000 \\ 
  1971 & 0.053 & 0.000 & 0.000 \\ 
  1970 & 0.053 & 0.000 & 0.000 \\ 
   \hline
\end{tabular}
\end{table}

\begin{table}[ht]
\centering
\begin{tabular}{r|rrr}
  \hline
 & DID & SC & SDID \\ 
  \hline
Alabama & 0.026 & 0.000 & 0.000 \\ 
  Arkansas & 0.026 & 0.000 & 0.003 \\ 
  Colorado & 0.026 & 0.013 & 0.058 \\ 
  Connecticut & 0.026 & 0.104 & 0.078 \\ 
  Delaware & 0.026 & 0.004 & 0.070 \\ 
  Georgia & 0.026 & 0.000 & 0.002 \\ 
  Idaho & 0.026 & 0.000 & 0.031 \\ 
  Illinois & 0.026 & 0.000 & 0.053 \\ 
  Indiana & 0.026 & 0.000 & 0.010 \\ 
  Iowa & 0.026 & 0.000 & 0.026 \\ 
  Kansas & 0.026 & 0.000 & 0.022 \\ 
  Kentucky & 0.026 & 0.000 & 0.000 \\ 
  Louisiana & 0.026 & 0.000 & 0.000 \\ 
  Maine & 0.026 & 0.000 & 0.028 \\ 
  Minnesota & 0.026 & 0.000 & 0.039 \\ 
  Mississippi & 0.026 & 0.000 & 0.000 \\ 
  Missouri & 0.026 & 0.000 & 0.008 \\ 
  Montana & 0.026 & 0.232 & 0.045 \\ 
  Nebraska & 0.026 & 0.000 & 0.048 \\ 
  Nevada & 0.026 & 0.204 & 0.124 \\ 
  New Hampshire & 0.026 & 0.045 & 0.105 \\ 
  New Mexico & 0.026 & 0.000 & 0.041 \\ 
  North Carolina & 0.026 & 0.000 & 0.033 \\ 
  North Dakota & 0.026 & 0.000 & 0.000 \\ 
  Ohio & 0.026 & 0.000 & 0.031 \\ 
  Oklahoma & 0.026 & 0.000 & 0.000 \\ 
  Pennsylvania & 0.026 & 0.000 & 0.015 \\ 
  Rhode Island & 0.026 & 0.000 & 0.001 \\ 
  South Carolina & 0.026 & 0.000 & 0.000 \\ 
  South Dakota & 0.026 & 0.000 & 0.004 \\ 
  Tennessee & 0.026 & 0.000 & 0.000 \\ 
  Texas & 0.026 & 0.000 & 0.010 \\ 
  Utah & 0.026 & 0.396 & 0.042 \\ 
  Vermont & 0.026 & 0.000 & 0.000 \\ 
  Virginia & 0.026 & 0.000 & 0.000 \\ 
  West Virginia & 0.026 & 0.000 & 0.034 \\ 
  Wisconsin & 0.026 & 0.000 & 0.037 \\ 
  Wyoming & 0.026 & 0.000 & 0.001 \\ 
   \hline
\end{tabular}
\end{table}

\FloatBarrier

\section{Staggered Adoption}
\label{sec:staggered}

In the paper so far we have focused on the case where some units start receiving the treatment at a common point in time, what \citet{athey2017matrix} call block assignment. Under block assignment the $N\times T$ matrix of treatment assignments $\bw$ has the form like the following matrix, where units 3-6 all adopt the treatment in period 5:
\[ \bw=\left(
		\begin{array}{ccccccccc}
		& 1 & 2 & 3 & 4 & 5 & 6 & 7\\
1&		\tick & \tick & \tick & \tick & \tick  & \tick & \tick \\
		2&\tick & \tick  & \tick & \tick   & \tick & \tick & \tick  \\
		3&\tick & \tick & \tick & \tick & \tock   & \tock & \tock  \\
		4&\tick  & \tick& \tick & \tick & \tock  & \tock & \tock  \\
		5&\tick & \tick & \tick & \tick & \tock   & \tock & \tock  \\
		6&\tick & \tick & \tick & \tick & \tock   & \tock & \tock  \\
		\end{array}
		\right)\,.
		\]
This is a common setting, but there are other settings that are of interest. Another important special case is that of {\it staggered adoption} ({\it e.g.,} \citet{athey2021design}) with multiple dates at which the treatment is started. For example,  in the following assignment matrix  units 5 and 6 adopt the treatment in period 3, and units 3 and 4 adopt the treatment in period 5 (and units 1 and 2 never adopt the treatment):
\[ \bw=\left(
		\begin{array}{ccccccccc}
		& 1 & 2 & 3 & 4 & 5 & 6 & 7\\
	1&	\tick & \tick & \tick & \tick & \tick  & \tick & \tick \\
	2&	\tick & \tick  & \tick & \tick   & \tick & \tick & \tick  \\
	3&	\tick & \tick & \tick & \tick & \tock   & \tock & \tock  \\
	4&			\tick & \tick & \tick & \tick & \tock   & \tock & \tock  \\
	5&	\tick & \tick & \tock & \tock & \tock   & \tock & \tock  \\
	6&	\tick & \tick & \tock & \tock & \tock   & \tock & \tock  \\
			\end{array}
		\right)\,.
		\]
With staggered adoption the weighted DID  regression approach in SDID does not work directly.
However, there are various alternatives. Here we discuss a simple modification to estimate the average treatment effect for the treated in that setting by applying the SDID estimator repeatedly, once for every adoption date. An alternative is the procedure developed in \citet{ben2019synthetic}. In the above example with two adoption dates, we can create two assignment matrices, $\bw^1$ and $\bw^2$, that both fit into the block assignment setting. We can then apply the SDID estimator to both samples, and calculated a weighted average of the two estimators, with the weight equal to the fraction of treated unit/time-period pairs in each of the two samples. In the above example, the first sample would consist of units 1, 2, 5 and 6, and the second sample would consist of units 1, 2, 3, and 4, as illustrated in the two assignment matrices below:
\[ \bw^{\rm 1}=\left(
		\begin{array}{ccccccccc}
		& 1 & 2 & 3 & 4 & 5 & 6 & 7\\
	1&	\tick & \tick & \tick & \tick & \tick  & \tick & \tick \\
	2&	\tick & \tick  & \tick & \tick   & \tick & \tick & \tick  \\
	5&	\tick & \tick & \tock & \tock & \tock   & \tock & \tock  \\
	6&	\tick & \tick & \tock & \tock & \tock   & \tock & \tock  \\
			\end{array}
		\right)\,.
		\hskip1cm
		\bw^{\rm 2}=\left(
		\begin{array}{ccccccccc}
		& 1 & 2 & 3 & 4 & 5 & 6 & 7\\
	1&	\tick & \tick & \tick & \tick & \tick  & \tick & \tick \\
	2&	\tick & \tick  & \tick & \tick   & \tick & \tick & \tick  \\
	3&	\tick & \tick & \tick & \tick & \tock   & \tock & \tock  \\
	4&			\tick & \tick & \tick & \tick & \tock   & \tock & \tock  \\
			\end{array}
		\right)\,.
		\]
Alternatively we can create the two samples by splitting the data up by time periods. In that case the first sample would consist of time periods 1, 2, 3, and 4, and the second sample would consist of time periods 1, 2, 5, 6, and 7, as illustrated below:
\[ \bw^{\rm 1}=\left(
		\begin{array}{ccccccccc}
		& 1 & 2 & 3 & 4 \\
	1&	\tick & \tick & \tick & \tick \\
	2&	\tick & \tick  & \tick & \tick   \\
	3&	\tick & \tick & \tick & \tick   \\
	4&			\tick & \tick & \tick & \tick   \\
	5&	\tick & \tick & \tock & \tock  \\
	6&	\tick & \tick & \tock & \tock  \\
			\end{array}
		\right)\,.
		\hskip1cm
		\bw^{\rm 2}=\left(
		\begin{array}{ccccccccc}
		& 1 & 2  & 5 & 6 & 7\\
	1&	\tick & \tick  & \tick  & \tick & \tick \\
	2&	\tick & \tick   & \tick & \tick & \tick  \\
	3&	\tick & \tick  & \tock   & \tock & \tock  \\
	4&			\tick & \tick & \tock   & \tock & \tock  \\
	5&	\tick & \tick  & \tock   & \tock & \tock  \\
	6&	\tick & \tick  & \tock   & \tock & \tock  \\
			\end{array}
		\right)\,.
		\]

\section{Formal Results}

In this section, we will outline the proof of Theorem~\ref{theo:asymptotic-linearity}. 
Recall from Section~\ref{sec:oracle-and-adaptive} the decomposition of the SDID estimator's error into three terms:
oracle noise, oracle confounding bias, and the deviation of the SDID estimator from the oracle.
Our main task is bounding the deviation term. To do this, we prove an abstract high-probability bound,
then derive a more concrete bound using results from a companion paper on penalized
high-dimensional least squares with errors in variable \citep{hirshberg2020leastsquares},
and then show that this bound is $o\left((N_\ttt T_\post)^{-1/2}\right)$ under the assumptions of Theorem~\ref{theo:asymptotic-linearity}.
Detailed proofs for each step are included in the next section.

\paragraph{Notation}
Throughout, each instance of $c$ will denote a potentially different universal constant; 
$a \lesssim b$, $a \ll b$, and $a \sim b$ will mean $a \le cb$, $a/b \to 0$, 
and $c \le a/b \le c$ respectively.
$\norm{v}$ and $\norm{A}$ will denote the Euclidean norm $\norm{v}_2$ for a vector $v$
and the operator norm $\sup_{\norm{v}_2 \le 1}\norm{Av}$ for a matrix $A$ respectively;
$\sigma_{1}(A), \sigma_{2}(A), \ldots$ will denote the singular values of $A$; 
$A_{i\cdot}$ and $A_{\cdot j}$ will denote the $i$th row and $j$th column of $A$;
$v'$ and $A'$ will denote the transposes of a vector $v$ and matrix $A$;
and $[v; w] \in \R^{m+n}$ will denote the concatenation of vectors $v \in \R^m$ and $w \in \R^n$.

\subsection{Abstract Setting}
\label{sec:abstract-setting}
We will begin by describing an abstract setting 
that arises as a condensed form of the setting considered in our formal results in Section~\ref{sec:formal}.
We observe an $N \times T$ matrix $Y$, which we will decompose as the sum $Y_{it}=L_{it}+ 1(i=N,j=T)\tau + \varepsilon$ of
a deterministic matrix $L$ and a random matrix $\varepsilon$.
We will refer to four blocks,
\[ Y = \begin{pmatrix} Y_{::} & Y_{:T} \\ Y_{N:} & Y_{NT} \end{pmatrix}, \]
where $Y_{::}$ is a submatrix that omits the last row and column, 
$Y_{N:}$ is the last row omitting its last element, and
$Y_{:T}$ is the last column omitting its last element. 
We will use analogous notation for the parts of $L$ and $\varepsilon$ and let $N_0=N-1$ and $T_0=T-1$.

We assume that rows of $\varepsilon$ are independent and subgaussian and that for $i \le N_0$ they are identically distributed with 
linear post-on-pretreatment autoregression function $\E[\varepsilon_{iT} \mid \varepsilon_{i:}] = \varepsilon_{i:}\psi$ and
covariance $\Sigma = \E \varepsilon_{i\cdot}' \varepsilon_{i\cdot}$ and let $\Sigma^N$ be the covariance matrix of $\varepsilon_{N\cdot}$.
We will refer to the covariance of the subvectors $\varepsilon_{i:}$ and $\varepsilon_{N:}$ as $\Sigma_{::}$ and $\Sigma_{::}^N$ respectively.

Our abstract results involve a bound $K$ characterizing the concentration of the rows $\varepsilon_{i\cdot}$.
\begin{equation}
\label{eq:K-conditions}
\begin{aligned}
K &\ge \max\p{1,\ \norm{\varepsilon_{1:}\Sigma_{::}^{-1/2}}_{\psi_2},\ \norm{\varepsilon_{N:}(\Sigma^N_{::})^{-1/2}}_{\psi_2} \ 
\frac{\norm{\varepsilon_{1T} - \varepsilon_{1:}\psi}_{\psi_2 \mid \varepsilon_{1:}}}{\norm{\varepsilon_{1T} - \varepsilon_{1:}\psi}_{L_2}}}, \\
&P\p{\abs{\norm{\varepsilon_{1:}}^2 - \E \norm{\varepsilon_{1:}}^2 } \ge u} \
    \le c\exp\p{-c \min\p{\frac{u^2}{K^4 \E\norm{\varepsilon_{1:}}^2},\  \frac{u}{K^2\norm{\Sigma_{::}}}}} \ \text{ for all }\ u \ge 0.
\end{aligned}
\end{equation}
Here we follow the convention \citep[e.g.,][]{vershynin2018high} that the subgaussian norm of a random vector $\xi$ is
$\norm{\xi}_{\psi_2} := \sup_{\norm{x}\le 1}\norm{x'\xi}_{\psi_2}$.
The conditional subgaussian norm $\norm{\cdot}_{\psi_2 \mid Z}$ is defined like the subgaussian norm the conditional distribution given $Z$.
When the rows of $\varepsilon$ are gaussian vectors, these conditions are satisfied for $K$ equal to a sufficiently large universal constant.
In the gaussian case, $\varepsilon_{1T} - \varepsilon_{1:}\psi$ is independent of $\varepsilon_{i:}$, 
the squared subgaussian norm of a gaussian random vector is bounded by a multiple of the operator norm of its covariance,
and the concentration of $\norm{\varepsilon_{1:}}^2$ as above is implied by the Hanson-Wright inequality 
\citep[e.g.,][Theorem 6.2.1]{vershynin2018high}.

\subsection{Concrete Setting}
\label{sec:concrete-setting}
We map from the setting considered in Section~\ref{sec:formal} to our condensed form by averaging within blocks as follows.
\[
\begin{pmatrix} Y_{::} & Y_{:T} \\ Y_{N:} & Y_{NT} \end{pmatrix} =
\begin{pmatrix} \by_{\ccc,\pre} & \by_{\ccc,\post}\lambda_{\post} \\ \omega_{\ttt}'\by_{\ttt,\pre} & \omega_{\ttt}'\by_{\ttt,\post}\lambda_{\post} \end{pmatrix}.
\]
Here $\lambda_{\post} \in \R^{T_{\post}}$ and $\omega_{\ttt} \in \R^{N_{\ttt}}$ 
are vectors with equal weight $1/T_{\post}$ and $1/N_{\ttt}$ respectively. When working with this condensed form,
we write $\omega$ and $\lambda$ for what is rendered $\omega_{\ccc}$ and $\lambda_{\ttt}$ in Section~\ref{sec:formal}.
We will also use $\Omega$ and $\Lambda$ to denote the sets that would be written 
$\set{\omega_{\ccc} : \omega \in \Omega}$ and $\set{\lambda_{\pre} : \lambda \in \Lambda}$ 
in the notation used in Equations~\ref{unit_weights} and \ref{time_weights}. Note that these sets $\Omega$ and $\Lambda$ 
are the unit simplex in $\R^{N_0}=\R^{N_{\ccc}}$ and $\R^{T_0}=\R^{T_{\pre}}$ respectively. 

In this condensed form, rows $\varepsilon_{i\cdot}$ are independent gaussian vectors with mean zero
and covariance matrix $\Sigma$ for $i \le N_0$ and $N_{\ttt}^{-1}\Sigma$ for $i=N$. This matrix $\Sigma$ satisfies,
with quantities on the right defined as in Section~\ref{sec:formal},
\[ \Sigma = \begin{pmatrix} 
	    \Sigma_{\pre,\pre}			       & \Sigma_{\pre,\post}\lambda_{\post} \\
	    \lambda_{\post}'\Sigma_{\post,\pre}	       & \lambda_{\post}' \Sigma_{\post,\post} \lambda_{\post} 
	    \end{pmatrix}. \]
Note that because all rows have the same covariance up to scale, they have the same autoregression vector,
$\psi = \argmin_{v \in \R^{T_0}}\E(\varepsilon_{i:}v - \varepsilon_{iT})^2$. 
This definition is equivalent to the one given in Section~\ref{sec:formal}. 
And this characterization of $\varepsilon_{i:}\psi$ as a least squares projection implies that $\varepsilon_{i:}\psi - \varepsilon_{iT}$ 
and $\varepsilon_{i:}$ are uncorrelated and, being jointly normal, therefore independent.

That the eigenvalues of non-condensed-form $\Sigma$ 
are bounded and bounded away from zero implies that the eigenvalues of the submatrix $\Sigma_{::}=\Sigma_{\pre,\pre}$
are bounded and bounded away from zero. Furthermore, it implies 
the variance of $\varepsilon_{i:}\psi - \varepsilon_{iT}$ is on the order of $1/T_{\post}$.

To show this, we establish an upper and lower bound of that order. We will write $\sigma_{\min}(\Sigma)$ and $\sigma_{\max}(\Sigma)$ for the smallest
and largest eigenvalues of $\Sigma$. 
For the lower bound, we calculate its variance $\E\p{\varepsilon_{i\cdot} \cdot [\psi;\ -\lambda_{\post}]}^2= [\psi;\ -\lambda_{\post}]\ \Sigma \ [\psi;\ -\lambda_{\post}]$,
and observe that this is at least $\norm{[\psi; -\lambda_{\post}]}^2 \sigma_{\min}(\Sigma)$.
This implies an order $1/T_{\post}$ lower bound, as $\norm{[\psi; -\lambda_{\post}]}^2 \ge \norm{\lambda_{\post}}^2 = 1/T_{\post}$. 
For the upper bound, observe that because $\varepsilon_{iT} - \varepsilon_{i:}\psi$ is the orthogonal projection of $\varepsilon_{iT}$ on a subspace,
specifically the subspace orthogonal to $\set{\varepsilon_{i:}v : v \in \R^{T_{\pre}}}$,
its variance is bounded by that of $\varepsilon_{iT}$. 
This is $[0;\lambda_{\post}] \ \Sigma\ [0; \lambda_{\post}] \le \sigma_{\max}(\Sigma) \norm{\lambda_{\post}}^2 = \sigma_{\max}(\Sigma)/T_{\post}$.
\subsection{Theorem~\ref{theo:asymptotic-linearity} in Condensed Form}

In the abstract setting we've introduced above, we can write a weighted difference-in-differences treatment effect estimator 
as the difference between our (aggregate) treated observation $Y_{NT}$ and an estimate $\hat{Y}_{NT}$ of the corresponding (aggregate) control potential outcome.
In the concrete setting considered in Section~\ref{sec:formal}, this coincides with the estimator defined in \eqref{eq:tau-weighted-differences}. 
\begin{equation}
\label{eq:treatment-effect-condensed-form}
\hat \tau(\lambda,\omega) = Y_{NT} - \hat Y_{NT}(\lambda,\omega) \ \text{ where }\ \hat Y_{NT}(\lambda,\omega) := Y_{N:}\lambda + \omega' Y_{:T} - \omega' Y_{::} \lambda. 
\end{equation}
And the following weights coincide with the definitions used in Section~\ref{sec:formal}.
\begin{equation}
\label{eq:weights-condensed-form}
\begin{aligned}
&\homega_0,  \homega  = \argmin_{\omega_0,  \omega \in  \R \times \Omega} \norm{\omega_0 + \omega' Y_{::} - Y_{N:}}^2 + \zeta^2 T_0 \norm{\omega}^2,\\ 
&\somega_0,  \somega  = \argmin_{\omega_0,  \omega \in  \R \times \Omega} \norm{\omega_0 + \omega' L_{::} - L_{N:}}^2 + (\zeta^2+\sigma^2)T_0 \norm{\omega}^2, \\
&\hlambda_0, \hlambda = \argmin_{\lambda_0, \lambda \in \R \times \Lambda} \norm{\lambda_0 + Y_{::}\lambda - Y_{:T}}^2,\\ 
&\slambda_0, \slambda = \argmin_{\lambda_0, \lambda \in \R \times \Lambda} \norm{\lambda_0 + L_{::}\lambda - L_{:T}}^2 + N_0 \norm{\Sigma_{::}^{1/2}(\lambda - \psi)}^2.
\end{aligned}
\end{equation}

The following assumptions on the condensed form hold in the setting considered in Theorem~\ref{theo:asymptotic-linearity}.
The first summarizes our condensed-form model. 
The second is implied by Assumption~\ref{ass:noise} for $N_1=N_{\ttt}$ and $T_1 \sim T_{\post}$ as described above in Section~\ref{sec:concrete-setting}. 
And the remaining three are condensed-form restatements of Assumptions~\ref{ass:sample_sizes}-\ref{weightss}, 
differing only in that we substitute $T_1 \sim T_{\post}$ for $T_{\post}$ itself.

\begin{assumption}[Model]
\label{ass:model-condensed}
We observe $Y_{it}=L_{it} + 1(i=N, t=T) \tau + \varepsilon_{it}$ for
deterministic $\tau \in \R$ and $L \in \R^{N \times T}$ and random $\varepsilon \in \R^{N \times T}$.
And we define $N_0=N-1$ and $T=T_0-1$.
\end{assumption}

\begin{assumption}[Properties of Errors]
\label{ass:error-condensed}
The rows $\varepsilon_{i\cdot}$ of the noise matrix are independent gaussian vectors with mean zero
and covariance matrix $\Sigma$ for $i \le N_0$ and $N_1^{-1}\Sigma$ for $i=N$ where the eigenvalues of 
$\Sigma_{::}$ are bounded and bounded away from zero. Here $N_1 > 0$ can be arbitrary
and we define $T_1= 1/\var[\varepsilon_{i:} \psi - \varepsilon_{iT}]$ and 
$\psi = \argmin_{v \in \R^{T_0}}E(\varepsilon_{i:}v - \varepsilon_{iT})^2$.
\end{assumption}

\begin{assumption}[Sample Sizes]
\label{ass:sample-sizes-condensed}
We consider a sequence of problems where $T_0/N_0$ is bounded and bounded away from zero,
$T_1$ and $N_1$ are bounded away from zero, and
$N_0 / (N_1 T_1 \max(N_1, T_1)\log^2(N_0)) \to \infty$.
\end{assumption}

\begin{assumption}[Properties of $L$]
\label{ass:rank-condensed}
For the largest integer $K \le \sqrt{\min(T_0,N_0)}$,
\[ \sigma_{K}(L_{::})/K \ll \min(N_1^{-1/2}\log^{-1/2}(N_0), T_1^{-1/2}\log^{-1/2}(T_0)).  \]
\end{assumption}

\begin{assumption}[Properties of Oracle Weights]
\label{ass:properties-of-oracle-weights-condensed}
We use weights as in \eqref{eq:weights-condensed-form} for \\
$\zeta \gg (N_1T_1)^{1/4}\log^{1/2}(N_0)$ and the oracle weights satisfy
 \begin{align*} 
(i)\   &\max(\norm{\somega}, \norm{\slambda-\psi}) \ll (N_1 T_1)^{-1/2} \log^{-1/2}(N_0), \\
(ii.\omega)\ &\norm{\somega_0 + \somega'L_{::} - L_{N:}}  \ll N_0^{1/4} (N_1 T_1\max(N_1,T_1))^{-1/4} \log^{-1/2}(N_0), \\
(ii.\lambda)\ &\norm{\slambda_0 + L_{::}\slambda - L_{:T}} \ll N_0^{1/4} (N_1 T_1)^{-1/8}, \\
(iii)\ &L_{NT} - \somega' L_{:T} - L_{N:}\slambda + \somega' L_{::} \slambda \ll (N_1 T_1)^{-1/2}. 
\end{align*}
\end{assumption}

\noindent The following condensed form asymptotic linearity result  
implies Theorem~\ref{theo:asymptotic-linearity}.
\begin{theorem}
\label{theorem:asymptotic-linearity-condensed}
If Assumptions~\ref{ass:model-condensed}-\ref{ass:properties-of-oracle-weights-condensed} hold,
then $\hat \tau(\hlambda, \homega) - \tau = \varepsilon_{NT} - \varepsilon_{N:}\psi + o_p((N_1 T_1)^{-1/2}).$
\end{theorem}

The following lemma reduces its proof to
demonstrating the negligibility of the difference $\Delta_{oracle} := \hat \tau(\homega,\hlambda) - \hat\tau(\somega,\slambda)$
    between the SDID estimator and the corresponding oracle estimator. Its proof is a straightforward calculation.
 Note that the bounds it requires on the oracle weights are looser than what is required by Assumption~\ref{ass:properties-of-oracle-weights-condensed}(i);
those tighter bounds are used to control $\Delta_{oracle}$. 
\begin{lemma}
\label{lemma:reduction-to-oracle-deviation}
If deterministic $\somega,\slambda$ satisfy $\norm{\somega} = o(N_1^{-1/2})$ and $\norm{\slambda - \psi} = o(T_1^{-1/2})$
and Assumptions~\ref{ass:model-condensed}, \ref{ass:error-condensed}, and \ref{ass:properties-of-oracle-weights-condensed}(iii) hold,
then $\hat \tau(\somega,\slambda) - \tau = \varepsilon_{NT} - \varepsilon_{N:}\psi + o_p((N_1 T_1)^{-1/2})$.
\end{lemma}

To show that this difference $\Delta_{oracle}$ is small, 
we use bounds on the difference between the estimated and oracle weights
based on \citet[Theorem 1]{hirshberg2020leastsquares}. We summarize these bounds
in Lemma~\ref{lemma:weight-consistency} below.

\begin{lemma}
\label{lemma:weight-consistency}
If Assumptions~\ref{ass:model-condensed}, \ref{ass:error-condensed}, and \ref{ass:rank-condensed} hold;
$T_1$ and $N_1$ are bounded away from zero;
$N_0,T_0 \to \infty$ with $N_0 \ge \log^2(T_0)$ and $T_0 \ge \log^2(N_0)$;
and we choose weights as in \eqref{eq:weights-condensed-form} for 
unit simplices $\Omega \subseteq \R^{N_0}$ and $\Lambda \subseteq \R^{T_0}$,
then the following bounds hold on an event of probability $1-c\exp(-c\min(N_0^{1/2}, T_0^{1/2}, N_0/\norm{L_{::}\slambda + \slambda_0 - L_{:T}}, T_0/\norm{\somega'L_{::} + \somega_0 - L_{N:}}))$:
\begin{align*}
&\norm{\hlambda_0 - \slambda_0  + L_{::} (\hlambda - \slambda)} \le cvr_{\lambda}, 
&&\norm{\hlambda - \slambda} \le cvN_0^{-1/2}r_{\lambda}, \\
&\norm{\homega_0 - \somega_0  + L_{::}' (\homega - \somega)} \le cvr_{\omega}, 
&&\norm{\homega - \somega} \le cv(\eta^2 T_0)^{-1/2}r_{\omega}
\end{align*}
for $\eta^2 = \zeta^2 + 1$, some universal constant $c$, and
\begin{align*}
r_{\lambda}^2 &= (N_0/\Teff)^{1/2} \sqrt{\log(T_0)}
				+ \norm{L_{::}\slambda + \slambda_0 - L_{:T}} \sqrt{\log(T_0)}, && \Teff^{-1/2} = \norm{\slambda - \psi} + T_1^{-1/2} \\
r_{\omega}^2 &= (T_0/\Neff)^{1/2} \sqrt{\log(N_0)} 
				+ \norm{L_{::}'\somega + \somega_0 - L_{N:}'} \sqrt{\log(N_0)}, && \Neff^{-1/2} = \norm{\somega} + N_1^{-1/2}. 
\end{align*}
\end{lemma}
When Assumptions~\ref{ass:sample-sizes-condensed} and \ref{ass:properties-of-oracle-weights-condensed}(i-ii) hold as well,
these bounds hold with probability $1-c\exp(-cN_0^{1/2})$, as together 
those assumptions they imply the lemma's conditions on $N_0,T_0,N_1,T_1$ and that
$N_0/\norm{L_{::}\slambda + \slambda_0 - L_{:T}} \gg N_0^{3/4}$
and $T_0/\norm{\somega'L_{::} + \somega_0 - L_{N:}} \gg N_0^{3/4}$.

We conclude by using bounds of this form, in conjunction with the
first order orthogonality of the weighted difference-in-differences estimator $\hat \tau(\lambda,\omega)$
to the weights $\lambda$ and $\omega$, to control $\Delta_{oracle}$.
We do this abstractly in Lemma~\ref{lemma:lowrank-general},
then derive from it a simplified bound from which it will
be clear that $\Delta_{oracle} = o_p((N_1 T_1)^{-1/2})$ under our assumptions.

\begin{lemma}
\label{lemma:lowrank-general}
In the setting described in Section~\ref{sec:abstract-setting}, let $\Lambda \subseteq \R^{T_0}$ and $\Omega \subseteq \R^{N_0}$ be sets with the property that 
$\sum_{t \le T_0} \lambda_t = \sum_{i \le N_0} \omega_i = 1$ for all $\lambda \in \Lambda$ and $\omega \in \Omega$.
Let $\hlambda_0,\hlambda \in \R \times \Lambda$ and $\homega_0,\homega \in \R \times \Omega$ be random and
$\slambda_0,\slambda \in \R \times \Lambda$ and $\slambda_0,\slambda \in \R \times \Omega$ be deterministic.
On the intersection of an event of probability $1-c\exp(-u^2)$ and one on which
\begin{equation}
\label{eq:lowrank-general-consistency}
\begin{aligned} 
\sigma \norm{\omega-\somega} \le s_{\lambda}
\quad &\text{ and } \quad \norm{\homega_0 - \somega_0 + (\homega - \somega)'L_{::}} \le r_{\omega}, \\
\norm{\Sigma_{::}^{1/2}(\hlambda-\slambda)} \le s_{\omega}
\quad &\text{ and } \quad  \norm{\hlambda_0-\slambda_0 + L_{::}(\hlambda - \slambda)} \le r_{\lambda},
\end{aligned}
\end{equation}
the corresponding treatment effect estimators defined in \eqref{eq:treatment-effect-condensed-form}
are close in the sense that  
\begin{align*}
\smallabs{\hat \tau(\hlambda,\homega) - \hat \tau(\slambda,\somega)} 
&\le cuK [\Neff^{-1/2}s_{\lambda} + \Teff^{-1/2} s_{\omega} + \sigma^{-1}s_{\omega}s_{\lambda}] \\
&+ cK[ (\norm{\somega} + \sigma^{-1}s_{\omega})\width(\Sigma_{::}^{1/2}\sLambda_{s_{\lambda}}) + 
       (\norm{\Sigma_{::}^{1/2}(\psi-\slambda)} + s_{\lambda})\width(\sOmega_{s_{\omega}})] \\
& +     \sigma^{-1}s_{\omega}\min_{\lambda_0 \in \R}\norm{S_{\lambda}^{1/2} (L_{::}\slambda + \lambda_0 - L_{:T})} 
  +	s_{\lambda}\min_{\omega_0 \in \R}\norm{S_{\omega}^{1/2}\Sigma_{::}^{-1/2}(L_{::}'\somega + \omega_0 - L_{N:}')} \\ 
&+     \min\p{\norm{\Sigma_{::}^{-1/2}} r_{\omega} s_{\lambda},\  
	      \sigma^{-1} s_{\omega} r_{\lambda},\   
	      \min_{k \in \mathbb{N}}\sigma_k(L_{::}^c)^{-1} r_{\lambda} r_{\omega} + \sigma^{-1}\norm{\Sigma_{::}^{-1/2}}\sigma_{k+1}(L_{::}^c) s_{\lambda} s_{\omega}}
\end{align*}
Here $c$ is a universal constant, $\width(S)$ is the gaussian width of the set $S$, and 
\begin{align*}
&\Teff^{-1/2} = \sigma^{-1}(\norm{\Sigma_{::}^{1/2}(\slambda - \psi)} + \norm{\tilde\varepsilon_{iT}}_{L_2}), 
&&\Neff^{-1/2} = \norm{\somega} + \norm{(\Sigma_{::}^{N})^{1/2}\Sigma_{::}^{-1/2}}, \\
&\sLambda_{s} = \set{\lambda - \slambda : \lambda \in \sLambda, \norm{\Sigma_{::}^{1/2}(\lambda - \slambda)} \le s}, 
&&\sOmega_{s} = \set{\omega - \somega : \omega \in \sOmega, \sigma\norm{\omega - \somega} \le s}, \\
&S_{\lambda} = I-L_{::}(L_{::}'L_{::} + (\sigma r_{\omega}/s_{\omega})^2 I)^{-1}L_{::}', 
&&S_{\omega} = I-\Sigma_{::}^{-1/2}L_{::}'(L_{::}\Sigma_{::}^{-1}L_{::}' + (r_{\lambda}/s_{\lambda})^2 I)^{-1}L_{::}\Sigma_{::}^{-1/2}, \\
&L_{::}^{c} = L_{::} - N_0^{-1}1_{N_0} 1_{N_0}' L_{::} - L_{::} T_0^{-1} 1_{T_0} 1_{T_0}'. &&
\end{align*}
\end{lemma}

We simplify this using bounds $s_{\omega},s_{\lambda},r_{\omega},r_{\lambda}$ from Lemma~\ref{lemma:weight-consistency} 
and bounds $\width(\sOmega_{s_\omega}) \lesssim \sqrt{\log(N_0)}$ 
and $\width(\sLambda_{s_\lambda}) \lesssim \sqrt{\log(T_0)}$ that hold for the specific sets $\Omega,\Lambda$ used in our concrete setting 
\citep[Example 1]{hirshberg2020leastsquares}. 

\begin{corollary}
\label{cor:lowrank-simplified}
Suppose Assumptions~\ref{ass:model-condensed}, \ref{ass:error-condensed}, and \ref{ass:rank-condensed} hold
with $T_0 \sim N_0$ and that $\log(N_0)$, $T_1$ and $N_1$ are bounded away from zero. Let $m_0=N_0$, $m_1=\sqrt{N_1 T_1}$, and $\bar m_1=\max(N_1, T_1)$.
Consider the weights defined in \eqref{eq:weights-condensed-form} 
with $\Omega \subseteq \R^{N_0}$ and $\Lambda \subseteq \R^{T_0}$ taken to be the unit simplices and $\zeta \gg  m_1^{1/2} \log^{1/2}(m_0)$.
With probability $1-2\exp(-\min(T_1\log(T_0), N_1\log(N_0))) - c\exp(-cN_0^{1/2}))$,
$\hat \tau(\homega,\hlambda) - \hat\tau(\slambda,\somega) = o_p((N_1 T_1)^{-1/2})$ if
\begin{align*} 
\max(\norm{\somega},\norm{\psi-\slambda}) &\ll m_1^{-1} \log^{-1/2}(m_0), \\
\norm{\somega_0 + \somega'L_{::} - L_{N:}}  &\ll m_0^{1/4} m_1^{-1/2} \bar m_1^{-1/4} \log^{-1/2}(m_0), \\
\norm{\slambda_0 + L_{::}\slambda - L_{:T}} &\ll m_0^{1/4} m_1^{-1/4}, 
\end{align*}
and the latter two bounds go to infinity.
\end{corollary}
These assumptions are implied by Assumptions~\ref{ass:model-condensed}-\ref{ass:properties-of-oracle-weights-condensed}.
Assumption~\ref{ass:sample-sizes-condensed} states
our assumptions $T_0 \sim N_0$, $\log(N_0),T_1,N_1 \not\to 0$, and that
the (fourth power of) the second bound above goes to infinity; when the second bound does go to infinity, so does the third. 
As Assumption~\ref{ass:sample-sizes-condensed} implies that that $T_0 \sim N_0 \to \infty$, it implies the probability stated in the lemma above goes to one. And Assumption~\ref{ass:properties-of-oracle-weights-condensed}(i-ii) states that the bound above hold.

As our assumptions imply the conclusions of Lemma~\ref{lemma:reduction-to-oracle-deviation} and Corollary~\ref{cor:lowrank-simplified},
and those two results imply the conclusions of Theorem~\ref{theorem:asymptotic-linearity-condensed}, this concludes our proof.

\section{Proof Details}
In this section, we complete our proof by proving the lemmas used in the sketch above.

\subsection{Proof of Lemma~\ref{lemma:reduction-to-oracle-deviation}}
First, consider the oracle estimator's bias,
\[ \E \hat\tau(\slambda,\somega) - \tau =  (L_{NT} + \tau) - \somega' L_{:T} - L_{N:}\slambda + \somega' L_{::} \slambda - \tau. \]
Assumption~\ref{ass:properties-of-oracle-weights-condensed}(iii) is that this is $o_p((N_1 T_1)^{-1/2})$.

Now consider the oracle estimator's variation around its mean,
\begin{align*}
\hat\tau(\slambda,\somega)- \E \hat\tau(\slambda,\somega) 
&= \varepsilon_{NT} - \varepsilon_{N:}\slambda + \somega' \varepsilon_{:T} + \somega' \varepsilon_{::} \slambda \\
&= (\varepsilon_{NT} - \varepsilon_{N:}\slambda) - \somega'(\varepsilon_{:T} - \varepsilon_{::} \slambda) \\
&= (\varepsilon_{NT} - \varepsilon_{N:}\psi) - \somega'(\varepsilon_{:T} - \varepsilon_{::} \psi) 
- \varepsilon_{N:}(\slambda - \psi) + \somega'\varepsilon_{::}(\slambda - \psi).
\end{align*}
The conclusion of our lemma holds if all but the first term in the decomposition above are $o_p((N_1 T_1)^{-1/2})$.
We do this by showing that each term has $o((N_1 T_1)^{-1})$ variance.
\begin{align*}
&\E (\somega'(\varepsilon_{:T} - \varepsilon_{::} \psi))^2 
= \norm{\somega}^2 \E (\varepsilon_{1T} - \varepsilon_{i:}\psi)^2 = \norm{\somega}^2/T_1, \\
&\E (\varepsilon_{N:}(\slambda - \psi))^2 
= (\slambda-\psi)'(\E \varepsilon_{N:}' \varepsilon_{N:})  (\slambda-\psi) \le \norm{\slambda - \psi}^2 \norm{\Sigma_{::}}/N_1, \\
& \E (\somega'\varepsilon_{::}(\slambda - \psi))^2 = \norm{\somega}^2 \E (\varepsilon_{1:}(\slambda-\psi))^2 \le \norm{\somega}^2\norm{\slambda}^2 \norm{\Sigma_{::}}.
\end{align*}
Our assumption that $\norm{\Sigma_{::}}$ is bounded and our assumed bounds on $\norm{\somega}$ and $\norm{\slambda}$ 
imply that each of these is $o((N_1 T_1)^{-1})$ as required.

\subsection{Proof of Lemma~\ref{lemma:weight-consistency}}
The bounds involving $\lambda$ follow from the application of \citet[Theorem 1]{hirshberg2020leastsquares} with $\eta^2=1$, $A=L_{::}$, $b=L_{:T}$, 
and $[\varepsilon, \nu]= [\varepsilon_{::}, \varepsilon_{:T}]$ with independent rows, using the bound 
$\width(\sLambda_s)\lesssim \sqrt{\log(T_0)}$ mentioned in its Example 1.
The bounds for $\omega$ follow from the application of the same theorem with $\eta^2=1+\zeta^2/\sigma^2$
for $\sigma^2=\trace(\Sigma_{::})/T_0$, $A=L_{::}'$, $b=L_{N:}'$, and $[\varepsilon, \nu]=\varepsilon_{::}', \varepsilon_{N:}']$
with independent columns, using the analogous bound $\width(\sOmega_s) \lesssim \sqrt{\log(N_0)}$.

In the first case, \citet[Theorem 1]{hirshberg2020leastsquares} gives bounds of the claimed form for 
\begin{align*} 
&r_{\lambda}^2 = [(N_0/\Teff)^{1/2} + \norm{L_{::}\slambda + \slambda_0 - L_{:T}}] \sqrt{\log(T_0)} + 1 \quad \text{ holding with probability } \\
&1-c\exp\p{-c\min(N_0\log(T_0)/r_{\lambda}^2, v^2 R, N_0)} \text{ if }\ \sigma_{R+1}(L_{::})/R \le cv T_1^{-1/2}\log^{-1/2}(T_0) \quad \text{ and } \\
&R \le \min(v^2 (N_0 \Teff)^{1/2}, v^2 N_0 / \log(T_0), cN_0).
\end{align*}
To see this, ignore constant order factors of $\phi\ (\ge 1)$ and $\norm{\Sigma}$ in \citet[Theorem 1]{hirshberg2020leastsquares} and
substitute $s^2 = cv^2 r_{\lambda}^2 / (\eta^2 n)$ for problem-appropriate parameters $\eta^2=1$, $n=N_0$, $n_{eff}^{-1/2}=\Teff^{-1/2}\ (\ge T_1^{-1/2})$, and
$\bar \width(\Theta_s) = \sqrt{\log(T_0)}$.

In the second case, \citet[Theorem 1]{hirshberg2020leastsquares} gives bounds of the claimed form for 
\begin{align*} 
&r_{\omega}^2 = [(T_0/\Neff)^{1/2} + \norm{\somega'L_{::} + \somega_0 - L_{N:}}] \sqrt{\log(N_0)} + \log(N_0) \quad \text{ holding with probability } \\
&1-c\exp\p{-c\min(\eta^2 T_0\log(N_0)/r_{\omega}^2, v^2 R, T_0)} \text{ if }\ \sigma_{R+1}(L_{::})/R \le cv N_1^{-1/2}\log^{-1/2}(N_0) \quad \text{ and } \\
&R \le \min(v^2 (T_0 \Neff)^{1/2}, v^2 \eta^2 T_0 / \log(N_0), cT_0).
\end{align*}
To see this, ignore constant order factors of $\phi\ (\ge 1)$ and $\norm{\Sigma}$ in \citet[Theorem 1]{hirshberg2020leastsquares} and 
substitute $s^2 = cv^2 r_{\lambda}^2 / (\eta^2 n)$ for problem-appropriate parameters $\eta^2=1+\zeta^2/\sigma^2$, $n=T_0$, $n_{eff}^{-1/2}=\Neff^{-1/2}\ (\ge N_1^{-1/2})$,
and $\bar \width(\Theta_s) = \sqrt{\log(N_0)}$.

We will now simplify our conditions on $R$. As we have assumed that $N_1$ and $T_1$ and therefore $\Neff$ and $\Teff$ are bounded away from zero,
we can choose $v$ of constant order with $v \ge \max( c/\Teff, c/\Neff, 1)$, so our upper bounds on $R$ simplify to
\[ R \le \min(N_0^{1/2}, N_0/\log(T_0), cN_0) \quad  \text{ and } \quad R \le \min(T_0^{1/2}, \eta^2 T_0 / \log(N_0), T_0) \]
respectively. Having assumed that that $N_0, T_0 \to \infty$ with $N_0 \ge \log^2(T_0)$ and $T_0 \ge \log^2(N_0)$,
these conditions simplify to $R \le N_0^{1/2}$ and $R \le T_0^{1/2}$. Thus, it suffices that the largest integer $R \le \min(N_0,T_0)^{1/2}$
satisfy $\sigma_{R+1}(L_{::})/R \le c\min(N_1^{-1/2}\log^{-1/2}(N_0), T_1^{-1/2}\log^{-1/2}(T_0))$.
This is implied, for any constant $c$, by Assumption~\ref{ass:rank-condensed}.

We conclude by simplifing our probability statements. As noted above, we take $R \sim \min(N_0,T_0)^{1/2}$, so we may make this substitution.
Furthermore, again using our assumption that $\Neff$ and $\Teff$ are bounded away from zero,
\begin{align*} 
\frac{N_0\log(T_0)}{r_{\lambda}^2} 
&\gtrsim \min\p{ \frac{ N_0\log(T_0)}{(N_0/\Teff)^{1/2}\sqrt{\log(T_0)}},\ 
		 \frac{ N_0\log(T_0)}{ \norm{L_{::}\slambda + \slambda_0 - L_{:T}}\sqrt{\log(T_0)}},
		 \frac{ N_0\log(T_0)}{1}}\\
&\gtrsim \min\p{\sqrt{N_0},\ N_0/\norm{L_{::}\slambda + \slambda_0 - L_{:T}}}, \\
\frac{T_0\log(N_0)}{r_{\omega}^2}
&\gtrsim \min\p{\frac{T_0\log(N_0)}{(T_0/\Neff)^{1/2}\sqrt{\log(N_0)}}, \ 
		\frac{T_0\log(N_0)}{\norm{\somega'L_{::} + \somega_0 - L_{N:}} \sqrt{\log(N_0)}}, \
		\frac{T_0\log(N_0)}{\log(N_0)}} \\
&\gtrsim \min\p{\sqrt{T_0},\ T_0/\norm{\somega'L_{::} + \somega_0 - L_{N:}}}.
\end{align*}
Thus, each bound holds with probability at least 
$1-c\exp(-c\min(N_0^{1/2}, T_0^{1/2}, N_0/\norm{L_{::}\slambda + \slambda_0 - L_{:T}}, T_0/\norm{\somega'L_{::} + \somega_0 - L_{N:}}))$.
And by the union bound, doubling our leading constant $c$, both simultaneously with such a probability.

\subsection{Proof of Lemma~\ref{lemma:lowrank-general}}
We begin with a decomposition of the difference between the SDID estimator and the oracle.
\begin{align*}
&\tau(\slambda,\somega) - \hat \tau(\hlambda,\homega) \\ 
&=\hat Y_{NT}(\hlambda,\homega) - Y_{NT}(\slambda,\somega) \\
&=\sqb{Y_{N:}\hlambda + \homega' Y_{:T} - \homega' Y_{::} \hlambda}  -  \sqb{ Y_{N:} \slambda + \somega' Y_{:T} - \somega' Y_{::} \slambda } \\
&=Y_{N:}(\hlambda-\slambda) + (\homega-\somega)' Y_{:T} - \sqb{(\homega-\somega)' Y_{::} (\hlambda-\slambda) + \somega' Y_{::} (\hlambda - \slambda) +  (\homega-\somega)' Y_{::} \slambda } \\
&=(Y_{N:} - \somega'Y_{::})(\hlambda-\slambda) + (\homega-\somega)' (Y_{:T} - Y_{::}\slambda) - (\homega-\somega)' Y_{::} (\hlambda-\slambda). 
\end{align*}
We bound these terms. As $Y_{it}= L_{it} + 1(i=N, t=T)\tau + \varepsilon$, we can decompose each of these three terms into two parts, 
one involving $L$ and the other $\varepsilon$. We will begin by treating the parts involving $\varepsilon$.
\begin{enumerate}
\item The first term is a sum $\varepsilon_{N:}(\hlambda-\slambda) - \somega'\varepsilon_{::}(\hlambda-\slambda)$. 
Because $\hlambda$ is independent of $\varepsilon_{N:}$, the first of these is subgaussian conditional on $\hlambda$,
with conditional subgaussian norm $\norm{\varepsilon_{N:} (\hlambda-\slambda)}_{\psi_2 \mid \hlambda} 
\le \norm{\varepsilon_{N:}(\Sigma^N_{::})^{-1/2}}_{\psi_2}\norm{(\Sigma^N_{::})^{1/2}\Sigma_{::}^{-1/2}}\norm{\Sigma_{::}^{1/2}(\hlambda-\slambda)}$.
It follows that it satisfies a subgaussian tail bound $\smallabs{\varepsilon_{N:}(\hlambda-\slambda)} 
\le cu\norm{\varepsilon_{N:}(\Sigma^N_{::})^{-1/2}}_{\psi_2}$ $\norm{(\Sigma^N_{::})^{1/2}\Sigma_{::}^{-1/2}} \norm{\Sigma_{::}^{1/2}(\hlambda-\slambda)}$
 with conditional probability $1-2\exp(-u^2)$. This implies that the same bound holds unconditionally on an event of probability $1-2\exp(-u^2)$.

Furthermore, via generic chaining \citep[e.g.,][Theorem 8.5.5]{vershynin2018high}, on an event of probability $1-2\exp(-u^2)$,
either $\Sigma_{::}^{1/2}(\hlambda-\slambda) \not\in \sLambda_{s_{\lambda}}$ or $\smallabs{\somega'\varepsilon_{::}(\hlambda-\slambda)} \le 
c\norm{\somega'\varepsilon_{::}\Sigma_{::}^{-1/2}}_{\psi_2}(\width(\Sigma_{::}^{1/2}\sLambda_{s_{\lambda}}) + u\rad(\Sigma_{::}^{1/2}\sLambda_{s_{\lambda}})) 
\le c \norm{\varepsilon_{i:}\Sigma_{::}^{-1/2}}_{\psi_2} \norm{\somega}(\width(\Sigma_{::}^{1/2}\sLambda_{s_{\lambda}}) + u s_{\lambda})$. 
The second comparison here follows from Hoeffding's inequality \citep[e.g.,][Theorem 2.6.3]{vershynin2018high}.
Thus, by the union bound, on the intersection of an event of probability $1-c\exp(-u^2)$ and one on which \eqref{eq:lowrank-general-consistency} holds,
\begin{align*}  
&\smallabs{(\varepsilon_{N:} - \somega'\varepsilon_{::})(\hlambda-\slambda)} \\
&\le cu\norm{\varepsilon_{N:}(\Sigma^N_{::})^{-1/2}}_{\psi_2}\norm{(\Sigma_{::}^N)^{1/2} \Sigma_{::}^{-1/2}} s_{\lambda}
+    c\norm{\varepsilon_{1:}\Sigma^{-1/2}}_{\psi_2}\norm{\somega}(\width(\Sigma_{::}^{1/2}\sLambda_{s_{\lambda}}) + us_{\lambda}) \\
&\le cuK\Neff^{-1/2}s_{\lambda} +
     cK\norm{\somega}\width(\Sigma_{::}^{1/2}\sLambda_{s_{\lambda}}).
\end{align*}
 
\item The second term is similar to the first. It is a sum $(\homega-\somega)'\tilde\varepsilon_{:T} + (\homega - \somega)' \varepsilon_{::}(\psi - \slambda)$ for $\tilde \varepsilon_{:T} = \varepsilon_{:T} - \varepsilon_{::}\psi$. Because $\homega$ is a function of $\varepsilon_{::}, \varepsilon_{N:}$
and $\tilde\varepsilon_{:T}$ is mean zero conditional on them, the first of these terms is 
a weighted average of conditionally independent mean-zero subgaussian random variables.
Applying Hoeffding's inequality conditionally, it follows that 
its magnitude is bounded by $cu \norm{\homega-\somega} \max_{i<N} \norm{\tilde\varepsilon_{iT}}_{\psi_2 \mid \varepsilon{::}, \varepsilon_{N:}} \le 
cuK \norm{\homega-\somega} \norm{\tilde\varepsilon_{1T}}_{L_2}$ on an event of probability $1-2\exp(-u^2)$. In the second comparison, 
we've used the independence of rows $\varepsilon_{i\cdot}$,
the identical distribution of rows for $i < N$, 
and the assumption that $\norm{\tilde\varepsilon_{1T}}_{\psi_2 \mid \varepsilon_{1:}} \le K \norm{\tilde\varepsilon_{1T}}_{L_2}$.

Furthermore, via generic chaining, on an event of probability $1-c\exp(-u^2)$,
either $(\homega - \somega) \not\in \sOmega_{s_{\omega}}$ or $\smallabs{(\homega-\somega)\varepsilon_{::}(\psi-\slambda)} \le 
c\norm{\varepsilon_{::}(\psi - \slambda)}_{\psi_2}(\width(\sOmega_{s_{\omega}}) + u\rad(\sOmega_{s_{\omega}})) 
\le cK\norm{\Sigma_{::}^{1/2}(\psi-\slambda)}(\width(\sOmega_{s_{\omega}}) + u \rad(\sOmega_{s_{\omega}}))$. 
The second comparison here follows from Hoeffding's inequality.
Thus, by the union bound, on the intersection of an event of probability $1-c\exp(-u^2)$ and one on which \eqref{eq:lowrank-general-consistency} holds,
\begin{align*}
&\smallabs{(\homega-\somega)'(\varepsilon_{:T} - \varepsilon_{::}\slambda)} \\
&\le cuK \norm{\tilde\varepsilon_{1T}}_{L_2} \sigma^{-1}s_{\omega}
    +  cuK\norm{\Sigma_{::}^{1/2}(\psi-\slambda)}\sigma^{-1}s_{\omega} + cK\norm{\Sigma_{::}^{1/2}(\psi-\slambda)}\width(\sOmega_{s_{\omega}}) \\
&\le cuK \Teff^{-1/2} s_{\omega}
    + cK\norm{\Sigma_{::}^{1/2}(\psi-\slambda)}\width(\sOmega_{s_{\omega}}).
\end{align*}

\item Via Chevet's inequality \citep[Lemma~3]{hirshberg2020leastsquares}, on an event of probability $1-c\exp(-u^2)$, either $(\homega - \somega) \not\in \sOmega_{s_{\omega}}$,
$(\hlambda - \slambda) \not\in \sLambda_{s_{\lambda}}$, or
$\smallabs{(\homega - \somega)'\varepsilon_{::} (\hlambda - \slambda)} \le 
cK[\width(\sOmega_{s_{\omega}})\rad(\Sigma_{::}^{1/2}\sLambda_{s_{\lambda}}) +
   \rad(\sOmega_{s_{\omega}})\width(\Sigma_{::}^{1/2}\sLambda_{s_{\lambda}}) + u \rad(\sOmega_{s_{\omega}})\rad(\Sigma_{::}^{1/2}\sLambda_{s_{\lambda}})] 
\le cK[\width(\sOmega_{s_{\omega}})s_{\lambda} + \width(\Sigma_{::}^{1/2}\sLambda_{s_{\lambda}})\sigma^{-1}s_{\omega} + u \sigma^{-1}s_{\omega}s_{\lambda} ].$
On the intersection of this event and one on which \eqref{eq:lowrank-general-consistency} holds, 
the first two possibilities are ruled out and our bound on $\smallabs{(\homega - \somega)'\varepsilon_{::} (\hlambda - \slambda)}$ holds.
\end{enumerate}

\noindent By the union bound, these three bounds are satisfied on the intersection of one of probability 
$1-c\exp(-u^2)$ and one on which \eqref{eq:lowrank-general-consistency} holds. 
And by the triangle inequality, adding our bounds yields a bound on our terms involving $\varepsilon$.
\begin{equation}
\label{eq:epsilon-terms}
\begin{aligned} 
&\smallabs{(\varepsilon_{N:} - \somega'\varepsilon_{::})(\hlambda-\slambda) 
+ (\homega-\somega)' (\varepsilon_{:T} - \varepsilon_{::}\slambda) 
- (\homega-\somega)' \varepsilon_{::} (\hlambda-\slambda)} \\ 
    &\le cuK [\Neff^{-1/2}s_{\lambda} + \phi \Teff^{-1/2} s_{\omega} + \sigma^{-1}s_{\omega}s_{\lambda}] \\
    &+ cK[ (\norm{\somega} + \sigma^{-1}s_{\omega})\width(\Sigma_{::}^{1/2}\sLambda_{s_{\lambda}}) + 
	      (\norm{\Sigma_{::}^{1/2}(\psi-\slambda)} + s_{\lambda})\width(\sOmega_{s_{\omega}})] 
\end{aligned}
\end{equation}

We now turn our attention to the terms involving $L$. For any $\omega_0, \omega \in \R \times \R^{N_0}$, 
$(L_{N:} - \somega'L_{::})(\hlambda-\slambda)=(L_{N:} - \omega' L_{::} - \omega_0) (\hlambda - \slambda) + (\omega - \somega)'L_{::}(\hlambda - \slambda)$.
The value of the constant $\omega_0$ does not affect the expression because the sum of the elements of $\hlambda - \slambda$ is zero. 
By the Cauchy-Schwarz and triangle inequalities, it follows that  
\[ \smallabs{ (L_{N:} - \somega'L_{::})(\hlambda-\slambda)} \le  \norm{(L_{N:} - \omega' L_{::} - \omega_0)\Sigma_{::}^{-1/2}}\norm{\Sigma_{::}^{1/2}(\hlambda - \slambda)} + \norm{\omega-\somega}\norm{L_{::}(\hlambda - \slambda)} \]
Furthermore, substituting bounds implied by \eqref{eq:lowrank-general-consistency} and using the elementary bound $x + y \le 2\sqrt{x^2+y^2}$, 
we get a quantity that we can minimize explicitly over $\omega$. The following result; for $A=\Sigma_{::}^{-1/2}L_{::}'$, 
$b=\Sigma_{::}^{-1/2}(L_{N:}'-\omega_0 1)$, $\alpha=s_{\lambda}$, and $\beta=r_{\lambda}$ satisfying $\beta/\alpha=cN_0^{1/2}$;
implies the bound
\begin{align*}
\smallabs{ (L_{N:} - \somega'L_{::})(\hlambda-\slambda)} &\le 
2 s_{\lambda}\min_{\omega_0} \norm{S_{\omega}^{1/2}\Sigma_{::}^{-1/2}(L_{::}'\somega + \omega_0 - L_{N:}')} \\
S_{\omega} &= I-\Sigma_{::}^{-1/2}L_{::}'(L_{::}\Sigma_{::}^{-1}L_{::}' + (r_{\lambda}/s_{\lambda})^2 I)^{-1}L_{::}\Sigma_{::}^{-1/2}.
\end{align*}

\begin{lemma}
For any real matrix $A$ and appropriately shaped vectors $\tilde x$ and $b$,
$\min_x \alpha^2 \norm{Ax - b}^2 + \beta^2 \norm{x - \tilde x}^2 = \alpha^2 \norm{S^{1/2} (A\tilde x - b)}^2$
for $S=I-A(A'A + (\beta/\alpha)^2 I)^{-1}A'$. If $\beta=0$, the same holds for
$S=I-A(A'A)^{\dagger}A$.
\end{lemma}
\begin{proof}
Reparameterizing in terms of $y=x-\tilde x$ and defining $v=A\tilde x - b$ and $\lambda^2=\beta^2/\alpha^2$,
 this is $\alpha^2$ times $\min_y \norm{v + Ay}^2 + \lambda^2 \norm{y}^2 = \min_y \norm{v}^2 + 2y'A'v + y'(A'A + \lambda^2 I)y$. 
Setting the derivative of the expression to zero, we solve for the minimizer $y=-(A'A + \lambda^2 I)^{-1}A'v$ 
and the minimum $v'[I - A(A'A + \lambda^2 I)^{-1}A']v$, then multiply by $\alpha^2$.
\end{proof}

Analogously, for any $\lambda_0, \lambda \in \R \times \R^{T_0}$,
\[ \smallabs{  (\homega-\somega)' (L_{:T} - L_{::}\slambda)} \le \norm{L_{:T} - L_{::}\lambda - \lambda_0}\norm{\homega - \somega} + \norm{\lambda-\slambda}\norm{(\homega - \somega)'L_{::}}. \]
and therefore, when \eqref{eq:lowrank-general-consistency} holds,
\begin{align*} 
\smallabs{  (\homega-\somega)' (L_{:T} - L_{::}\slambda)} &\le
2\sigma^{-1}s_{\omega}\min_{\lambda_0}\norm{S_{\lambda}^{1/2} (L_{::}\slambda - \lambda_0 - L_{:T})} \\
S_{\lambda} &= I-L_{::}(L_{::}'L_{::} + (\sigma r_{\omega}/s_{\omega})^2 I)^{-1}L_{::}'.
\end{align*}

Finally, we can take the minimum of two Cauchy-Schwarz bounds on the third term, 
\begin{align*}
\smallabs{(\homega-\somega)' L_{::} (\hlambda-\slambda)}  
&= \smallabs{[(\homega_0 - \somega_0) + (\homega-\somega)' L_{::}] (\hlambda-\slambda)} \\ 
&\le \norm{  (\homega_0 - \somega_0) + (\homega-\somega)' L_{::}}\norm{\Sigma_{::}^{-1/2}}\norm{\Sigma_{::}^{1/2}(\hlambda - \slambda)}, \\
\smallabs{(\homega-\somega)' L_{::} (\hlambda-\slambda)}  
&= \smallabs{(\homega-\somega)' [(\hlambda_0 - \slambda_0) + L_{::} (\hlambda-\slambda)]} \\ 
&\le \norm{\homega-\somega} \norm{(\hlambda_0 - \slambda_0) + L_{::} (\hlambda-\slambda)}.
\end{align*} 
As above, the inclusion of either intercept does not effect the value of the expression because $\hlambda - \slambda$ and $\homega - \somega$ sum to one.
This implies that on an event on which the bounds \eqref{eq:lowrank-general-consistency} hold, 
\begin{equation}
\label{eq:L-terms}
\begin{aligned} 
&\smallabs{(L_{N:} - \somega'L_{::})(\hlambda-\slambda) 
+ (\homega-\somega)' (L_{:T} - L\slambda) 
- (\homega-\somega)' L_{::} (\hlambda-\slambda)} \\ 
& \le  2 s_{\lambda} \min_{\omega_0}\norm{S_{\omega}^{1/2}\Sigma_{::}^{-1/2}(L_{::}'\somega + \omega_0 - L_{N:}')}
+     2 \sigma^{-1} s_{\omega} \min_{\lambda_0}\norm{S_{\lambda}^{1/2} (L_{::}\slambda - \lambda_0 - L_{:T})} \\
&\quad + \min\p{\norm{\Sigma_{::}^{-1/2}} r_{\omega} s_{\lambda}, \sigma^{-1} s_{\omega} r_{\lambda}}.
\end{aligned}
\end{equation}
We can include in the minimum in the third term above another bound on $\smallabs{(\homega-\somega)' L_{::} (\hlambda-\slambda)}$.
We will use one that exploits a potential gap in the spectrum of $L_{::}$,
e.g., a bound on the smallest nonzero singular value of $L_{::}$. The abstract bound we will use is one on the inner product $x' A y$: 
given bounds $\norm{x'A} \le r_{x}$, $\norm{Ay} \le r_y$, $\norm{x} \le s_{x}$, $\norm{y} \le s_y$, it is no larger than 
$\min_{k} \sigma_k(A)^{-1} r_x r_y + \sigma_{k+1}(A) s_x s_y$. To show this, we first observe that 
without loss of generality, we can let $A$ be square, diagonal, and nonnegative with decreasing elemnts on the diagonal: in terms of its singular value decomposition $A=USV'$
and $x_U = U' x$ and $y_V = V' y$, $x'Ay = x_U' S y_V$ where $\norm{x_U' S} \le r_x$, $\norm{S y_V} \le r_y$, $\norm{x_U} \le s_x$, $\norm{y_V} \le s_y$.
In this simplified diagonal case, letting $a_i:=A_{ii}$ and $R=\rank(A)$,
\begin{align*} 
\smallabs{x'Ay} &= \smallabs{\sum_{i=1}^{R} x_i y_i a_i} \\
     &\le \smallabs{\sum_{i=1}^{k} x_i y_i a_i} + \smallabs{\sum_{i=k+1}^{R} x_i y_i a_i} \\
     &\le \sqrt{\sum_{i=1}^k x_i^2 a_i^2 \sum_{i=1}^k y_i^2} + \sqrt{\sum_{i=k+1}^{R} x_i^2 a_i^2 \sum_{i=k+1}^R y_i^2} \\
     &\le a_k^{-1} \sqrt{\sum_{i=1}^k x_i^2 a_i^2 \sum_{i=1}^k y_i^2 a_i^2} + a_{k+1} \sqrt{\sum_{i=k+1}^{R} x_i^2 \sum_{i=k+1}^R y_i^2} \\
     &\le a_k^{-1} r_x r_y + a_{k+1} s_x s_y.
\end{align*}
We apply this with $x=\homega-\somega$, $y=\hlambda-\slambda$, and $A = L_{::} - N_0^{-1} 1_{N_0} 1_{N_0}' L_{::} - L_{::} T_0^{-1} 1_{T_0} 1_{T_0}'$;
because $(\homega - \somega)'1_{N_0} = 0$ and $1_{T_0}' (\hlambda - \slambda) = 0$, $(\homega-\somega)' L_{::} (\hlambda-\slambda) = (\homega-\somega)' A (\hlambda - \slambda) = x'Ay$. When the bounds in \eqref{eq:lowrank-general-consistency} hold,
$\norm{x' A} \le r_{\omega}$ and $\norm{A y} \le r_{\lambda}$, as
\[ \norm{(\homega - \somega)'A}^2 = \sum_{t=1}^{T_0} \sqb{(\homega - \somega)'L_{:t} - T_{0}^{-1} \sum_{t=1}^{T_0}(\homega - \somega)' L_{:t}}^2 
				  = \min_{\delta \in \R}\norm{ (\homega - \somega)'L_{::} - \delta}^2 \le r_{\omega}^2. \]
These bounds also imply $\norm{x} \le \sigma^{-1}s_{\omega}$ and $\norm{y} \le \norm{\Sigma_{::}^{-1/2}}s_{\lambda}$, so our third term is bounded by 
\begin{align*}
\smallabs{(\homega-\somega)' L_{::} (\hlambda-\slambda)} 
&\le \min_{k} \sigma_k(A)^{-1} r_{\lambda} r_{\omega} + \sigma^{-1}\norm{\Sigma_{::}^{-1/2}}\sigma_{k+1}(A) s_{\lambda} s_{\omega} \\
\end{align*}

Adding together \eqref{eq:epsilon-terms} and \eqref{eq:L-terms}, including this additional bound in the minimum in 
the third term of \eqref{eq:L-terms}, we get the claimed bound on $\smallabs{\tau(\slambda,\somega) - \hat \tau(\hlambda,\homega)}$.

\subsection{Proof of Corollary~\ref{cor:lowrank-simplified}}
We begin with the bound from Lemma~\ref{lemma:lowrank-general}.
As the claimed bound is stated up to an unspecified universal constant,
we can ignore universal constants throughout. We can ignore $K$ as well;
as discussed in Section~\ref{sec:abstract-setting}, as in the gaussian case we consider, it can be taken to be a universal constant.
Furthermore, we can ignore all appearances of powers of $\sigma$, $\Sigma_{::}$, and $S_{\theta}$ for $\theta \in \set{\lambda,\omega}$,
using bounds $\width(\Sigma_{::}^{k}\cdot) \le  \norm{\Sigma_{::}^{k}}\width(\cdot)$,
$\norm{\Sigma_{::}^{k}\cdot} \le \norm{\Sigma^{k}}\norm{\cdot}$, and
$\norm{S_{\theta}^{1/2}\cdot} \le \norm{S_{\theta}^{1/2}}\norm{\cdot}$
and observing that $\norm{S_{\theta}}\le 1$ by construction
and, under Assumption~\ref{ass:error-condensed}, $\norm{\Sigma_{::}}$ and $\norm{\Sigma_{::}^{-1}}$ 
are bounded by universal constants. And we bound minima over $\omega_0$ and $\slambda_0$ 
by substituting $\somega_0$ and $\slambda_0$. Then, as $\width(\sLambda_{s_{\lambda}}) \lesssim \sqrt{\log(T_0)}$
and $\width(\sOmega_{s_{\omega}}) \lesssim \sqrt{\log(N_0)}$, 
Lemma~\ref{lemma:weight-consistency} and Lemma~\ref{lemma:lowrank-general}
together (taking $\sigma=1$ in the latter), 
imply that on an event of probability $1-c\exp(-u^2)-c\exp(-v)$ for $v$ as in Lemma~\ref{lemma:weight-consistency}, 
the following bound holds for $\eta^2=1+\zeta^2$.

\begin{align*}
\smallabs{\hat \tau(\hlambda,\homega) - \hat \tau(\slambda,\somega)} 
&\lesssim u[\Neff^{-1/2} N_0^{-1/2}r_{\lambda} + \Teff^{-1/2} (\eta^2T_0)^{-1/2} r_{\omega} + (\eta^2 N_0 T_0)^{1/2}r_{\omega}r_{\lambda}] \\
&+ (\norm{\somega} + (\eta^2 T_0)^{-1/2}r_{\omega})\log^{1/2}(T_0) + 
       (\norm{\psi-\slambda} + N_0^{-1/2}r_{\lambda})\log^{1/2}(N_0) \\
&+     (\eta^2 T_0)^{-1/2}r_{\omega}E_{\lambda} + N_0^{-1/2}r_{\lambda}E_{\omega} + r_{\omega}r_{\lambda} M \qquad \text { for any } 
\end{align*}
\[
M \ge \min\p{N_0^{-1/2},\ (\eta^2 T_0)^{-1/2}, \
				   \min_{k \in \mathbb{N}}\sigma_k(L_{::}^c)^{-1} + \sigma_{k+1}(L_{::}^c) (\eta^2 N_0 T_0)^{-1/2}} \ \text{ and }\ \]
\begin{align*}
r_{\lambda} &= \log^{1/4}(T_0)[(N_0/\Teff)^{1/4} + E_{\lambda}^{1/2}], & E_{\lambda} = \norm{L_{::}\slambda + \slambda_0 - L_{:T}}, 
\ \ & \Teff^{-1/2} = \norm{\slambda - \psi} + T_1^{-1/2}, \\
r_{\omega}  &= \log^{1/4}(N_0)[(T_0/\Neff)^{1/4} + E_{\omega}^{1/2}],  & E_{\omega} = \norm{L_{::}'\somega + \somega_0 - L_{N:}'},
\ \ & \Neff^{-1/2} = \norm{\somega} + N_1^{-1/2}. 
\end{align*}
Taking $u = \min(\Teff^{1/2}\log^{1/2}(T_0), \Neff^{1/2}\log^{1/2}(N_0), (\eta^2 N_0 T_0)^{1/2}M)$,
we can ignore the first line in the bound above, 
as its three terms are bounded by the second term in the second line, the first term in the second line,
and the final term respectively. Grouping terms with common powers of $r_{\omega},r_{\lambda}$;
redefining $\E_{\lambda} = \max(E_{\lambda},1)$ and $\E_{\omega}  = \max(E_{\omega}, 1)$,
and expanding $r_{\omega},r_{\lambda}$ yields the following bound.
\begin{equation}
\label{eq:lowrank-simplified-expanded-bound}
\begin{aligned}
&\norm{\somega} \log^{1/2}(T_0) + \norm{\psi-\slambda} \log^{1/2}(N_0) \\
&+ (\eta^2 T_0)^{-1/2}[(T_0/\Neff)^{1/4} + E_{\omega}^{1/2}] E_{\lambda} \log^{1/2}(N_0) \\
&+         N_0^{-1/2} [(N_0/\Teff)^{1/4} + E_{\lambda}^{1/2}]E_{\omega}  \log^{1/2}(T_0)  \\
&+ M				    [(N_0T_0 / \Neff\Teff)^{1/4} +
				     (N_0/\Teff)^{1/4}E_{\omega}^{1/2} +
				     (T_0/\Neff)^{1/4}E_{\lambda}^{1/2} +
				     (E_{\omega}E_{\lambda})^{1/2}] \log^{1/4}(N_0)\log^{1/4}(T_0). 
\end{aligned}
\end{equation}
Each term is multiplied by either $\log^{1/2}(T_0)$, $\log^{1/2}(N_0)$, or their geometric mean.
For simplicity, we will substitute a common upper bound of $\ell^{1/2}$ for $\ell =\log(\max(N_0,T_0))$.
To establish our claim, we must show that each term is $o((N_{1} T_{1})^{-1/2})$.

The first line of our bound is small enough, $\Neff \sim N_1$, and $\Teff \sim T_1$, if
\begin{equation}
\label{eq:sample-size-conditions}
\max(\norm{\somega}, \norm{\slambda - \psi}) \ll (N_1T_1)^{-1} \ell^{-1/2},\quad \min(N_1,T_1) \gtrsim 1,
\end{equation}
If the following bound holds, the remaining terms that do no involve $M$ are small enough.
\begin{equation}
\label{eq:E-conditions}
\begin{aligned} 
E_{\omega}  &\ll N_0^{1/4}N_1^{-1/2}T_1^{-1/4} \ell^{-1/2}, \\
E_{\lambda} &\ll \eta T_0^{1/4} N_1^{-1/4} T_1^{-1/2}\ell^{-1/2}, \\
(E_{\omega}E_{\lambda})^{1/2} &\ll \min(N_0^{3/8}T_1^{-3/8}N_1^{-1/4},\ \eta^{1/2}T_0^{3/8}N_1^{-3/8}T_1^{-1/4})\ell^{-1/4}.
\end{aligned}
\end{equation}
To see this, multiply the square root of the first bound by the first part of the third when bounding the term involving $E_{\lambda}^{1/2}E_{\omega}$ 
and the square root of the second by the second part of the third when bounding the term involving $E_{\omega}^{1/2}E_{\lambda}$.
Note that because our `redefinition' of $E_{\omega}, E_{\lambda}$ requires that they be no smaller than one,
these upper bounds must go to infinity, and so long as they do we can interpret them as bounds on 
$\norm{L_{::}'\somega + \somega_0 - L_{N:}'}$, $\norm{L_{::}\slambda + \slambda_0 - L_{:T}}$,
and their geometric mean respectively.

By substituting the bounds \eqref{eq:E-conditions} into the term with a factor of $M$ in \eqref{eq:lowrank-simplified-expanded-bound}, 
we can derive a sufficent condition for it to be small enough. To see that it is sufficient, we bound
first multiple of $M$ in \eqref{eq:lowrank-simplified-expanded-bound} using the first bound on $M$ below,
the second using the second in combination with our bound on $E_{\omega}$,
the third using the third in combination with our bound on $E_{\lambda}$,
and the fourth using the second in combination with our first bound on $(E_{\omega}E_{\lambda})^{1/2}$.
\begin{equation}
\label{eq:m-conditions}
M \ll \min\p{ (N_0 T_0 N_1 T_1 \ell)^{-1/4},\ N_0^{-3/8}N_1^{-1/4}T_1^{-1/8},\ \eta^{-1/2}T_0^{-3/8}T_1^{-1/4}N_1^{-1/8} }\ell^{-1/4}. 
\end{equation}
Equations~\ref{eq:sample-size-conditions}, \ref{eq:E-conditions}, and \ref{eq:m-conditions},
so long as the bounds in \eqref{eq:E-conditions} all go to infinity, are sufficient to imply our claim. 
Note that because every vector $\omega$ in the unit simplex in $\R^{N_0}$ 
satisfies $\norm{\omega} \ge N_0^{-1/2}$, \eqref{eq:sample-size-conditions} 
implies an additional constraint on the dimensions of the problem, $N_0 \gg N_1 T_1 \ell$.

Having established these bounds on $E_{\omega}$ and $E_{\lambda}$,
we are now in a position to characterize the probability that our result holds by lower bounding the ratios
$N_0/E_{\lambda}$ and $T_0/E_{\omega}$ that appear in the probability statement of 
Lemma~\ref{lemma:weight-consistency}. As $N_0/E_{\lambda} \gg N_0^{3/4}$ and $T_0/E_{\omega} \gg T_0^{3/4}$,
the claims of Lemma~\ref{lemma:weight-consistency} hold with probability $1-c\exp(-v)$ for $v=c\min(N_0,T_0)^{1/2}$.
Thus, recalling from above that we are working on an event of probability $1-c\exp(-u^2)-c\exp(-v)$
for $u = \min(\Teff^{1/2}\log^{1/2}(T_0), \Neff^{1/2}\log^{1/2}(N_0), (\eta^2 N_0 T_0)^{1/2}M)$
and that $\Neff \sim N_1$ and $\Teff \sim T_1$,
this is probability at least $1-2\exp(-\min(T_1\log(T_0), N_1\log(N_0), \eta^2 N_0 T_0 M^2)) - c\exp(-c\min(N_0^{1/2}, T_0^{1/2}))$.

We will now derive simplfied sufficient conditions under the assumption that $N_0 \sim T_0$.
Let $m_0 = N_0$, $m_1 = (N_1 T_1)^{1/2}$, and $\bar m_1 = \max(N_1,T_1)$. Then \eqref{eq:m-conditions} holds if 
\[ M \ll \min(m_0^{-1/2} m_1^{-1/2} \ell^{-1/2},\ \eta^{-1/2}m_0^{-3/8}m_1^{-1/4} \bar m_1^{-1/4} \ell^{-1/4}). \]

This is not satisfiable with $M=N_0^{-1/2} \sim m_0^{1/2}$. But with $M=(\eta T_0)^{-1/2} \sim \eta^{-1} m_0^{-1/2}$, 
it is satisfied for $\eta \gg \max(1,\ m_0^{-1/4}\bar m_1^{1/2}) m_1^{1/2} \ell^{1/2}$.
For such $\eta$, \eqref{eq:E-conditions} hold when
\begin{align*} 
E_{\omega}  &\ll m_0^{1/4} m_1^{-1/2}\bar m_1^{-1/4}  \ell^{-1/2}, \\
E_{\lambda} &\ll \max(m_0^{1/4} \bar m_1^{-1/4},\ \bar m_1^{1/4}) \\
(E_{\omega}E_{\lambda})^{1/2} &\ll m_0^{3/8}m_1^{-1/2}\bar m_1^{-1/8} \ell^{-1/4}.
\end{align*}
To keep the statement of our lemma simple, we use the simplified bound 
$E_{\lambda} \ll m_0^{1/4} \bar m_1^{-1/4}$. 
Then the geometric mean of our bounds on $E_{\omega}$ and $E_{\lambda}$ bounds their geometric mean,
and it is $m_0^{1/4}m_1^{-1/4}\bar m_1^{-1/4} \ell^{-1/4}$.
Thus, our explicit bound on the geometric mean above is redundant as long as the ratio of these two bounds,
$m_0^{1/4}m_1^{-1/4}\bar m_1^{-1/4} \ell^{-1/4} /  m_0^{3/8}m_1^{-1/2}\bar m_1^{-1/8} \ell^{-1/4}$,
is bounded. As this ratio simplifies to $m_0^{-1/8} m_1^{1/4} \bar m_1^{-1/8} \le (m_1/m_0)^{1/8}$
and $m_0 \gg m_1$, it is redundant. And taking $M \sim \eta^{-1} m_0^{-1/2}$ in our probability statement above,
our claims hold with probability $1-2\exp(-\min(T_1\log(T_0), N_1\log(N_0))) - c\exp(-cm_0^{1/2}))$.

To avoid complicating the statement of our result, we will not explore refinements made possible by a nontrivially large gap in the spectrum of $L_{::}^c$,
i.e., the case that $M=\min_k \sigma_k(L_{::}^c)^{-1} + \sigma_{k+1}(L_{::}^c) (\eta^2 N_0 T_0)^{-1/2}$. However, in models with no weak factors,
this quantity will be very small, and as a result, Equations~\ref{eq:sample-size-conditions} and \ref{eq:E-conditions} 
will essentially be sufficient to imply our claim.
As we make $\eta$ large only to control $M$ when it is equal to $(\eta T_0)^{-1/2}$, 
this provides some justification for the use of weak regularization ($\zeta$ small) or no regularization ($\zeta = 0$)
when fitting the synthetic control $\homega$. 

We conclude by observing that the lower bound on $\zeta$ above simplifies to $\zeta \gg m_1^{1/2} \ell^{1/2}$ under our stated assumptions.
We begin with the assumption that the above upper bound on $E_{\omega}$ goes to infinity. Observing that the other lower bound on $\zeta$ as stated above is $\bar m_1^{1/4}$ times the reciprocal of the 
this infinity-tending bound on $E_{\omega}$, it follows that it must be $o(\bar m_1^{1/4})$. As $m_1^{1/2}=\bar m_1^{1/4} \min(N_1,T_1)^{1/4}$ and the latter factor and $\ell^{1/2}$
are bounded away from zero by assumption, $\bar m_1^{1/4} = O(m_1^{1/2}\ell^{1/2})$, 
so this other lower bound is indeed smaller than the (other) one that we retain. 

\section{Proof of Theorem~\ref{theo:jack}}

Throughout this proof, we will assume constant treatment effects $\tau_{ij} = \tau$. When treatment effects are not constant,
the jackknife variance estimate will include an additional nonnegative term that depends on the amount of treatment heterogeneity, 
making the inference conservative. 

We will write $a \sim_p b$ meaning $a/b \rightarrow_p 1$, 
$a \lesssim_p b$ meaning $a = O_p(b)$, $a \ll_p b$ meaning $a = o_p(b)$,
$\sigma_{\min}(\Sigma)$ and $\sigma_{\max}(\Sigma)$ for the smallest and largest eigenvalues of a matrix $\Sigma$,
and $1_n \in \R^n$ for a vector of ones. And we write $\hlambda^\star$ to denote the concatenation of $\hlambda_\pre$ and $-\hlambda_\post$.

Now recall that, as discussed in Section \ref{sec:double_diff}, 
\begin{equation}
\begin{split}
\htau &=  \homega_{\ttt}' Y_{\ttt,\post} \hlambda_{\post} - \homega_{\ccc}' Y_{\ccc,\post} \hlambda_{\post} - \homega_{\ttt}' Y_{\ttt,\pre} \hlambda_{\pre} +
 \homega_{\ccc}' Y_{\ccc,\pre} \hlambda_{\pre} \\
&= \hmu_{\ttt} - \hmu_{\ccc} \quad \text{ where } \\
&\hmu_{\ccc} = \sum_{i = 1}^{N_{\ccc}} \homega_i \hDelta_i, 
\ \ \hmu_{\ttt} = \sum_{i = N_{\ccc} + 1}^{N} \homega_i \hDelta_i, 
\ \ \hDelta_i = Y_{i,\cdot} \hlambda^\star.
 \end{split}
\end{equation}
In the jackknife variance estimate defined in Algorithm \ref{alg:jack},
\begin{equation}
\label{eq:jack-estimates}
\hat \tau^{(-i)} =
    \begin{cases} 
	\hmu_{\ttt} - \frac{\sum_{k \le N_{\ccc}, k \neq i} \homega_k \Delta_k}{1-\homega_i}
         = \hmu_{\ttt} - \p{\hmu_{\ccc} - \frac{\homega_i (\Delta_i - \hmu_{\ccc})}{1-\homega_i}} & \text{ for }\ i \le N_{\ccc} \\
	\frac{\sum_{k \ge N_{\ccc}, k \neq 1}\homega_k \Delta_k}{1-\homega_i} - \hmu_{\ccc} 
	 =  \p{\hmu_{\ttt} - \frac{\homega_i (\Delta_i - \hmu_{\ttt})}{1-\homega_i}} - \hmu_{\ccc} & \text{ for }\  i > N_{\ccc}.
\end{cases}
\end{equation}
Thus, the jackknife variance estimate defined in Algorithm \ref{alg:jack} is
\begin{equation}
\label{eq:jack1}
\hV^{\mathrm{jack}}_\tau = \frac{N-1}{N} \p{\sum_{i = 1}^{N_\ccc} \p{\frac{\homega_i \p{\hDelta_i - \hmu_\ccc}}{1 - \homega_i}}^2 +
\sum_{i = N_\ccc + 1}^{N} \p{\frac{\homega_i \p{\hDelta_i - \hmu_\ttt}}{1 - \homega_i}}^2}.
\end{equation}
A few simplifications are now in order. We use the bound $\norm{\homega_{\ccc}}^2 \ll (N_\ttt T_\post \log(N_\ccc))^{-1}$
derived in Section~\ref{sec:bounding-homega} below. This bound implies that 
the denominators $1-\homega_i$ appearing in the expression above all lie in the interval
$[1-\max(\norm{\homega_\ccc}, N_{\ttt}^{-1}),\ 1] = [1-o_p(1),\ 1]$. 
As each term in that expression is nonnegative, it follows that the ratio
between it and the expression below, derived by replacing these denominators with $1$,
is in this interval and therefore converges to one.
\begin{equation}
\label{eq:jack2}
\hV^{\mathrm{jack}}_\tau \sim_p \sum_{i = 1}^{N_\ccc} \homega_i^2 \p{\hDelta_i - \hmu_\ccc}^2 +
\sum_{i = N_\ccc + 1}^{N} \homega_i^2 \p{\hDelta_i - \hmu_\ttt}^2.
\end{equation}
We will simplify this further by showing that the first term is negligible relative to the second. We'll start by lower bounding the second term.
This is straightforward because for $i > N_{\ccc}$, the unit weights $\homega_i$ are equal to the constant $1/N_{\ttt}$
and the time weights $\hlambda$ are independent of $Y_{i,\cdot}$. 
\begin{align*}
\E \sum_{i = N_{\ccc}+1}^N \homega_i^2 \p{\hDelta_i - \hmu_\ttt}^2  
&= N_{\ttt}^{-2}  \sum_{i=N_{\ccc}+1}^N \E ( (Y_{i,\cdot} - \homega_{\ttt}'Y_{\ttt,\cdot})\hlambda^\star)^2 
&& \\ 
&\ge N_{\ttt}^{-2} \sum_{i=N_{\ccc}+1}^N \E( (\varepsilon_{i,\cdot} - \homega_{\ttt}'\varepsilon_{\ttt,\cdot})\hlambda^\star)^2  
&& \\ 
&= N_{\ttt}^{-1} \E \hlambda_{\star}'\ (1-N_{\ttt}^{-1})\Sigma\ \hlambda^\star \ 
&&\text{as } \Cov{\varepsilon_{i,\cdot} - \homega_{\ttt}'\varepsilon_{\ttt,\cdot}} = (1-N_\ttt^{-1})\Sigma \\
&\ge N_{\ttt}^{-1} \norm{\hlambda^\star}^2 (1-N_{\ttt}^{-1})\sigma_{\min}(\Sigma) 
&& \\
&\ge (N_{\ttt} T_{\post})^{-1} (1-N_{\ttt}^{-1})\sigma_{\min}(\Sigma)  
&&\text{as } \norm{\hlambda^\star}^2 \ge \norm{\hlambda_\ttt}^2 = T_{\post}^{-1}.
\end{align*}
As $\sigma_{\min}(\Sigma)$ is bounded away from zero, it follows that the mean of the second term in \eqref{eq:jack2}
is on the order of $(N_{\ttt} T_{\post})^{-1}$ or larger.
We'll now show that the first term in \eqref{eq:jack2} is $o_p((N_{\ttt} T_{\post})^{-1})$,
so \eqref{eq:jack2} is equivalent to a variant in which we have dropped its first term. 

By H\"older's inequality and the bound 
$\norm{\homega_\ccc}^2 \ll (N_\ttt T_\post \log(N_\ccc))^{-1}$ derived in Section~\ref{sec:bounding-homega},
\[ \sum_{i = 1}^{N_\ccc} \homega_i^2 \p{\hDelta_i - \hmu_\ccc}^2 
\le \norm{\homega_{\ccc}}^2 \max_{i \le N_\ccc} \p{\hDelta_i - \hmu_\ccc}^2 
\ll (N_\ttt T_\post \log(N_\ccc))^{-1} \max_{i \le N_\ccc} \p{\hDelta_i - \hmu_\ccc}^2. \]
Thus, it suffices to show that $\max_{i \le N_\ccc} (\hDelta_i - \hmu_\ccc)^2 \ll \log(N_\ccc)$.
And it suffices to show that $\max_{i \le N_\ccc} \hDelta_i^2  \ll \log(N_\ccc)$, 
as $(\hDelta_i - \hmu_\ccc)^2 \le 2\hDelta_i^2 + 2\hmu_\ccc^2$ 
and $\hmu_\ccc$ is a convex combination of $\hDelta_1 \ldots \hDelta_{\ccc}$.
This bound holds because, by H\"older's inequality,
\[ \max_{i \le N_\ccc}\abs{\hDelta_i} = \max_{i \le N_\ccc}\abs{Y_{i,\cdot} \hlambda^\star} \le \Norm{\hlambda^\star}_{1} \cdot \max_{i \le N_\ccc,\ j \le T} \abs{Y_{ij}}
    \lesssim_p \sqrt{\log(N_\ccc)}.\]
In our last comparison above, we use the properties that $\norm{\hlambda^\star}_1 = \norm{\hlambda_\pre}_1 + \norm{\hlambda_\post}_1 = 2$,
that the elements of $L$ are bounded,
and that the maximum of $K=N_\ccc T$ 
gaussian random variables $\varepsilon_{it}$ is $O_p(\sqrt{\log(K)})$,
as well as Assumption~\ref{ass:sample_sizes}, which implies that $T \sim N_\ccc$ so $\log(K) \lesssim \log(N_\ccc)$.
Summarizing, 
\begin{equation}
\label{eq:jack3}
\hV^{\mathrm{jack}}_\tau \sim_p \frac{1}{N_\ttt^2} \sum_{i = N_\ccc+1}^{N} \p{\hDelta_i - \hmu_\ttt}^2.
\end{equation}
This simplification is as we would hope given that, under the conditions of Theorem \ref{theo:asymptotic-linearity},
we found that all the noise in $\htau$ comes from the exposed units. Now, focusing further on \eqref{eq:jack3} we
note that, when treatment effects are constant across units, we can verify that they do not contribute to $\hV^{\mathrm{jack}}_\tau$
and so
\begin{equation}
\label{eq:jack4}
\begin{split}
&\frac{1}{N_\ttt^2} \sum_{i = N_\ccc + 1}^{N}  \p{\hDelta_i - \hmu_\ttt}^2
 = \frac{1}{N_\ttt^2} \sum_{i = N_\ccc + 1}^{N}  \p{\hDelta_i(L) - \hmu_{\ttt}(L) + \hDelta_i(\varepsilon) - \hmu_{\ttt}(\varepsilon)}^2, \\
&\ \ \ \ \ \ \ \ \ \ \ \  \hDelta_i(L) = L_{i,\cdot} \hlambda^\star \ \ \ \ \
 \hDelta_i(\varepsilon) = \varepsilon_{i,\cdot} \hlambda^\star,
 \end{split}
\end{equation}
where $\hmu_{\ttt}(L)$ and $\hmu_{\ttt}(\varepsilon)$ are averages of $\hDelta_i(L)$ and $\hDelta_i(\varepsilon)$
respectively over the exposed units. Now, by construction, $\hlambda$ is only a function of the unexposed units and
so, given that there is no cross-unit correlation, $\hlambda$ is independent of $\varepsilon_{i,.}$ for all $i > N_\ccc$. Thus, the
cross terms between $\hDelta_i(L) - \hmu_{\ttt}(L)$ and $\hDelta_i(\varepsilon) - \hmu_{\ttt}(\varepsilon)$ in \eqref{eq:jack4}
are mean-zero and concentrate out, and so
 \begin{equation}
\label{eq:jack5}
\hV^{\mathrm{jack}}_\tau \sim_p \frac{1}{N_\ttt^2} \sum_{i = N_\ccc + 1}^{N}  \p{\hDelta_i(L) - \hmu_{\ttt}(L)}^2 
+ \frac{1}{N_\ttt^2} \sum_{i = N_\ccc + 1}^{N} \p{\hDelta_i(\varepsilon) - \hmu_{\ttt}(\varepsilon)}^2.
\end{equation}
We will now show that the second term is equivalent to a variant in which $\slambda$ replaces $\hlambda$. 
We denote by $\tilde\Delta$ and $\tilde\mu_{\ttt}$ the corresponding variants of $\hDelta$ and $\hmu_{\ttt}$.
First consider the second term in \eqref{eq:jack5}. $\hDelta_i(\varepsilon) = \tilde\Delta_i(\varepsilon) + \varepsilon_{i,\pre}(\hlambda_\pre - \slambda_\pre)$, so 
\begin{align*}
\p{\hDelta_i(\varepsilon) - \hmu_{\ttt}(\varepsilon)}^2 
&= \p{ [\tilde\Delta_i(\varepsilon) - \tilde\mu_{\ttt}(\varepsilon)] + (\varepsilon_{i,\pre} - \homega_{\ttt}'\varepsilon_{\ttt,\pre}) (\hlambda_\pre - \slambda_\pre) }^2 
 \\
&= \p{\tilde \Delta_i(\varepsilon) - \tilde \mu_{\ttt}(\varepsilon)}^2 \\
&+ 2[\tilde\Delta_i(\varepsilon) - \tilde\mu_{\ttt}(\varepsilon)] (\varepsilon_{i,\pre} - \homega_{\ttt}'\varepsilon_{\ttt,\pre})  (\hlambda_\pre - \slambda_\pre) \\
&+ ((\varepsilon_{i,\pre} - \homega_{\ttt}'\varepsilon_{\ttt,\pre}) (\hlambda_\pre - \slambda_\pre))^2.
\end{align*}
By the Cauchy-Schwarz inequality, the second and third terms in this decomposition are negligible relative to the first
if $\E_{\ttt} ((\varepsilon_{i,\pre} - \homega_{\ttt}'\varepsilon_{\ttt,\pre}) (\hlambda_\pre - \slambda_\pre))^2 \ll_p \E_{\ttt} (\tilde \Delta_i(\varepsilon) - \tilde \mu_{\ttt}(\varepsilon))^2$ where $\E_{\ttt}$ denotes expectation conditional on $\varepsilon_{\ccc,\cdot}$.  We calculate both quantities and compare.
\begin{align*} 
&\E_{\ttt} ((\varepsilon_{i,\pre} - \somega_{\ttt}'\varepsilon_{\ttt,\pre}) (\hlambda_\pre - \slambda_\pre))^2 
= (\hlambda_\pre - \slambda_\pre)' (1-N_{\ttt}^{-1}) \Sigma  (\hlambda_\pre - \slambda_\pre). \\
&\E_{\ttt} (\tilde \Delta_i(\varepsilon) - \tilde \mu_{\ttt}(\varepsilon))^2 = 
         \E_{\ttt}( (\varepsilon_{i,\cdot}- \somega_{\ttt}'\varepsilon_{\ttt,\pre})' \tlambda^\star)^2 
	=  \tlambda'\ (1-N_{\ttt}^{-1})\Sigma \ \tlambda.
\end{align*}
In Section~\ref{sec:bounding-hlambda}, we show that the first is $\lesssim_p N_\ccc^{-1/2} T_\post^{-1/2} \log^{1/2}(N_{\ccc})$,
and the second is $\gtrsim \norm{\tlambda^{\star}}^2 \ge T_{\post}^{-1}$ because $\sigma_{\min}(\Sigma)$ is bounded away from zero.
Thus, because $N_\ccc^{-1/2} \ll T_{\post}^{-1/2} \log^{-1/2}(N_{\ccc})$ under Assumption~\ref{ass:sample_sizes},
the first quantity is negligible relative to the second. As discussed, it follows that 
\begin{equation}
\label{eq:jack-noise}
\frac{1}{N_\ttt^2} \sum_{i = N_\ccc + 1}^{N} \p{\hDelta_i(\varepsilon) - \hmu_{\ttt}(\varepsilon)}^2 \sim_p 
\frac{1}{N_\ttt^2} \sum_{i = N_\ccc + 1}^{N} \p{\tilde \Delta_i(\varepsilon) - \tilde\mu_{\ttt}(\varepsilon)}^2.
\end{equation}
By the law of large numbers, the right side is equivalent ($\sim_p$) to its mean 
$N_\ttt^{-1}\tlambda'\ (1-N_{\ttt}^{-1})\Sigma \ \tlambda$ and therefore
to  $N_\ttt^{-1}\tlambda' \Sigma \tlambda$. It is shown that
$N_\ttt^{-1}\tlambda' \Sigma \tlambda \sim_p V_{\tau}$ in
the proof of Lemma~\ref{lemma:reduction-to-oracle-deviation}, so
\begin{equation}
\label{eq:jack6}
\hV^{\mathrm{jack}}_\tau \sim_p \frac{1}{N_\ttt^2} \sum_{i = N_\ccc + 1}^{N}  \p{\hDelta_i(L) - \hmu_{\ttt}(L)}^2 + V_{\tau}.
\end{equation}
Because the first term is nonnegative, our variance estimate is asymptotically either unbiased or upwardly biased, 
so our confidence intervals are conservative as claimed. In the remainder, we derive a sufficient condition
for the first term to be asymptotically negligible relative to $V_{\tau}$, so our confidence intervals have asymptotically nominal coverage.

We bound this term using the expansion $\hmu_{\ttt}(L) = N_\ttt^{-1}1_{N_\ttt}'(L_{\ttt,\post}\hlambda_\post - L_{\ttt,\pre}\hlambda_{\pre})$.
\begin{align*}
N_\ttt^{-2} \sum_{i = N_\ccc + 1}^{N}  \p{\hDelta_i(L) - \hmu_{\ttt}(L)}^2 
&= N_\ttt^{-2} \norm{ (I - N_{\ttt}^{-1}1_{N_{\ttt}}1_{N_{\ttt}}') (L_{\ttt,\pre}\hlambda_{\pre} + \hlambda_0 1_{N_\ttt} - L_{\ttt,\post}\hlambda_{\post}) }^2 \\
&\le N_\ttt^{-2} \norm{L_{\ttt,\pre}\hlambda_{\pre} + \hlambda_0  - L_{\ttt,\post}\hlambda_{\post}}^2.
\end{align*}
This comparison holds because $\norm{I- N_{\ttt}^{-1}1_{N_{\ttt}}1_{N_{\ttt}}} \le 1$.
By Assumption \eqref{eq:regr_cons}, this bound is $o_P((N_{\ttt}T_{\post})^{-1})$ and therefore negligible relative to $V_{\tau}$. 
We conclude by proving our claims about $\norm{\homega_{\ccc}}$ and $\norm{\Sigma_{\pre}^{1/2}(\hlambda_{\ccc} - \slambda_{\ccc})}$.

\subsubsection{Bounding $\norm{\homega_{\ccc}}$}
\label{sec:bounding-homega}
Here we will show that $\norm{\homega_\ccc}^2 \ll (N_\ttt T_\post \log(N_\ccc))^{-1}$ under the assumptions of Theorem \ref{theo:asymptotic-linearity}.
\begin{equation*}
\begin{aligned}
&\norm{\homega_{\ccc}-\somega_{\ccc}}^2  
\lesssim_p \ \zeta^{-2} N_\ccc^{-1} [ N_{\ccc}^{1/2}N_\ttt^{-1/2} + \norm{\somega_{\ccc}'L_{\ccc,\pre} + \somega_0 - \somega_{\ttt}'L_{\ttt,\pre}}] \log^{1/2}(N_\ccc) \\
&\ll      [N_\ttt^{1/2} T_\post^{1/2} \log(N_\ccc)]^{-1}  N_\ccc^{-1/2} N_\ttt^{-1/2} \log^{1/2}(N_\ccc) \\ 
&+        [N_\ttt^{1/2} T_\post^{1/2} \max(N_\ttt,T_\post)^{1/2} N_\ccc^{-1/4} \log(N_\ccc)]^{-1} N_{\ccc}^{-3/4} N_{\ttt}^{-1/4} T_{\post}^{-1/4} \max(N_{\ttt}, T_{\post})^{-1/4} \\
&\ll N_\ccc^{-1/2} N_\ttt^{-1} T_\post^{-1/2} \\
&\ll (N_\ttt T_\post \log(N_\ccc))^{-1}.
\end{aligned}
\end{equation*}
Our first bound follows from Lemma~\ref{lemma:weight-consistency}, in which we can take $N_{eff}^{-1/2} \sim N_{\ttt}^{-1/2}$
because $\norm{\somega_\ccc} \lesssim N_{\ttt}^{-1/2}$ under Assumption~\ref{weightss}. To derive our second,
we substitute the upper bound $N_{\ccc}^{1/4} N_{\ttt}^{-1/4}$ $T_{\post}^{-1/4} \max(N_{\ttt}, T_{\post})^{-1/4} \log^{-1/2}(N_{\ccc}) \gg \norm{\somega_{\ccc}'L_{\ccc,\pre} + \somega_0 - L_{\ttt,\pre}}$ from Assumption~\ref{weightss} and substitute (in brackets) two lower bounds on $\zeta^{2}$ chosen as in Theorem~\ref{theo:asymptotic-linearity}:
the first is implied by squaring the lower bound $\zeta \gg (N_{\ttt}T_{\post})^{1/4} \log^{1/2}(N_\ccc)$ 
and the second by multiplying this lower bound by an alternative lower bound, $\zeta \gg (N_\ttt T_\post)^{1/4}$ $\max(N_\ttt,T_\post)^{1/2}N_0^{-1/4} \log^{1/2}(N_\ccc)$.
The third is a simplification, and the fourth follows because $T_\post \log^2(N_\ccc) \ll N_\ccc$ under Assumption~\ref{ass:sample_sizes}.
Furthermore, as $\norm{\somega_\ccc}^2 \ll (N_\ttt T_\post \log(N_\ccc))^{-1}$ under Assumption~\ref{weightss},
by the triangle inequality, $\norm{\homega_\ccc}^2 \ll (N_\ttt T_\post \log(N_\ccc))^{-1}$ as claimed.
 
\subsubsection{Bounding $\norm{\Sigma_{\pre,\pre}(\hlambda_{\ccc} - \slambda_{\ccc})}$}
\label{sec:bounding-hlambda}
Here we will show that $\norm{\Sigma_{\pre,\pre}(\hlambda_{\ccc} - \slambda_{\ccc})}^2 \lesssim_p N_\ccc^{-1/2} T_\post^{-1/2} \log^{1/2}(N_{\ccc})$.
Because Assumption~\ref{ass:noise} implies that $\norm{\Sigma_{\pre,pre}}$ is bounded, it suffices to bound
$\norm{\hlambda_{\ccc} - \slambda_{\ccc}}$. 
\begin{align*}
\norm{\hlambda_{\ccc} - \slambda_{\ccc}}^2 
&\lesssim_p N_\ccc^{-1} [ N_{\ccc}^{1/2}T_\post^{-1/2} + \norm{L_{\ccc,\pre}\slambda_{\pre} + \slambda_0 - L_{\ccc,\post}\slambda_{\post}}] \log^{1/2}(N_\ccc) \\
&\lesssim   N_\ccc^{-1/2} T_\post^{-1/2} \log^{1/2}(N_{\ccc}) + N_\ccc^{-3/4}N_\ttt^{-1/8}T_\post^{-1/8} \log^{1/2}(N_{\ccc}) \\
&\lesssim   N_\ccc^{-1/2} T_\post^{-1/2} \log^{1/2}(N_{\ccc}).
\end{align*}
Our first bound follows from Lemma~\ref{lemma:weight-consistency}, in which we can take $T_{eff}^{-1/2} \sim T_{\post}^{-1/2}$
because $\norm{\slambda_\pre - \psi} \lesssim T_{\post}^{-1/2}$ under Assumption~\ref{weightss}. 
To derive our second, we substitute the upper bound 
$N_\ccc^{1/4} N_\ttt^{-1/8} T_\post^{-1/8} \gg \norm{L_{\ccc,\pre}\slambda_{\pre} + \slambda_0 - L_{\ccc,\post}\lambda_\post}$ from Assumption~\ref{weightss}.
The third follows because $N_{\ccc}^{-1/4} \ll N_\ttt^{-1/4}T_\post^{-1/4}\max(N_\ttt, T_\post)^{-1/4} \le N_\ttt^{-3/8}T_\post^{-3/8}$ 
under Assumption~\ref{ass:sample_sizes}.

\end{document}